\documentclass[a4paper,UKenglish,cleveref, autoref, thm-restate]{lipics-v2021}

\hideLIPIcs  

\usepackage[utf8]{inputenc}
\usepackage[T1]{fontenc}
\usepackage{amsfonts}
\usepackage{amsmath}
\usepackage{amsthm}
\usepackage{cmll}
\usepackage{stmaryrd}
\usepackage{amssymb}
\usepackage{bussproofs}
\usepackage{graphicx}
\usepackage{rotating}
\usepackage{float}
\usepackage{caption}
\usepackage{xcolor}
\usepackage{hyperref}
\usepackage{xspace}
\usepackage
[disable]
{todonotes}

\usepackage{tikz}
\usepackage{tikz-cd}
\usetikzlibrary{arrows,calc,matrix,shapes}
\usetikzlibrary{positioning,calc,arrows,decorations.pathmorphing,shapes.multipart}
\usetikzlibrary{automata}
\usetikzlibrary{matrix,shapes,fit,decorations.markings}

\date{}

\hypersetup{                    
    colorlinks=true,                
    breaklinks=true,                
    }

\newcommand*{\defname}[1]{\emph{\bf #1}}

\newcommand*{\logsys}[1]{\textnormal{\textsf{#1}}}
\newcommand*{\A}{\Gamma}
\newcommand*{\B}{\Delta}

\def\C{\mathcal{C}} 
\def\Exp{\mathcal{N}}

\newcommand*{\ocProofs}{\C^{\oc}}
\newcommand*{\ocsProofs}{\C^{\oc_s}}

\newcommand*{\ocgProofs}{\C^{\ocg}}

\newcommand*{\ocfProofs}{\C^{\ocf}}

\newcommand*{\ocuProofs}{\C^{\ocu}}
\newcommand*{\Var}{\mathcal{V}}
\newcommand*{\Atom}{\mathcal{A}}

\newcommand*{\redseq}{\rightarrow}
\newcommand\rmcutpar{\ensuremath{\mathsf{mcut(\iota, \perp\!\!\!\perp)}}\xspace}
\newcommand*{\rmcutparprime}{{\scriptsize\ensuremath{\mathsf{mcut(\iota', \perp\!\!\!\perp')}}}\xspace}
\newcommand*{\rmcutpardouble}{{\scriptsize\ensuremath{\mathsf{mcut(\iota'', \perp\!\!\!\perp'')}}}\xspace}
\newcommand*{\cutrel}{\perp\!\!\!\perp}
\newcommand*{\True}{\text{true}}
\newcommand*{\False}{\text{false}}

\newcommand*{\LL}{\logsys{\ensuremath{\text{LL}}}}
\newcommand*{\ELL}{\logsys{\ensuremath{\text{ELL}}}}
\newcommand*{\muELLinf}{\logsys{\ensuremath{\mu \text{ELL}^\infty}}}
\newcommand*{\LLL}{\logsys{\ensuremath{\text{LLL}}}}
\newcommand*{\SLL}{\logsys{\ensuremath{\text{SLL}}}}
\newcommand*{\seLL}{\logsys{\ensuremath{\text{seLL}}}}
\newcommand*{\MALL}{\logsys{\ensuremath{\text{MALL}}}}

\newcommand*{\superLL}{\logsys{\ensuremath{\text{superLL}}}}

\newcommand*{\muMALLinf}{\logsys{\ensuremath{\mu \text{MALL}^{\infty}}}}

\newcommand*{\muLLinf}{\logsys{\ensuremath{\mu \text{LL}^{\infty}}}}
\newcommand*{\muLLi}{\logsys{\ensuremath{\mu \text{LL}^{\infty}}}}

\newcommand*{\muLKinf}{\logsys{\ensuremath{\mu \text{LK}^{\infty}}}}

\newcommand*{\muLKmodinf}{\logsys{\ensuremath{\mu \text{LK}_{\Box}^{\infty}}}}
\newcommand*{\muLLmodinf}{\logsys{\ensuremath{\mu \text{LL}_{\Box}^{\infty}}}}

\newcommand*{\musuperLLinf}{\logsys{\ensuremath{\mu \text{superLL}^{\infty}}}}
\newcommand*{\musuperLLinfSig}{\logsys{\ensuremath{\mu \text{superLL}^{\infty}_{\Sig}}}}

\newcommand*{\ax}{\text{ax}}
\newcommand*{\exch}{\text{ex}}
\newcommand*{\mcut}{\text{mcut}}
\newcommand*{\cut}{\text{cut}}

\newcommand*{\wnwk}{\wn_{\text{w}}}
\newcommand*{\wnde}{\wn_{\text{d}}}
\newcommand*{\wncontr}{\wn_{\text{c}}}

\newcommand*{\ocprom}{\oc_{\text{p}}}
\newcommand*{\diacontr}{\lozenge_{\text{c}}}
\newcommand*{\diawk}{\lozenge_{\text{w}}}

\newcommand*{\boxprom}{\Box_{\text{p}}}
\newcommand*{\ocpromloz}{\ocprom^{\lozenge}}

\newcommand*{\mpx}[1]{\wn_{\text{m}_{#1}}}
\newcommand*{\contr}[1]{\wn_{\text{c}_{#1}}}

\newcommand*{\size}[1]{\#(#1)}

\newcommand*{\trans}[1]{{#1}^{\bullet}}

\newcommand*{\Sig}{\mathcal{E}}

\newcommand*{\Sigfun}[1]{[#1]}
\newcommand*{\e}{\sigma}
\newcommand*{\f}{\tau}
\newcommand*{\g}{\rho}
\newcommand*{\opmpx}{{{\scriptstyle\mathcal{O}}_{\text{mpx}}}}

\newcommand*{\perm}{\text{perm}}

\newcommand{\axgmpx}{(\mathsf{Ax^{g}_{m})}}
\newcommand{\axfumpx}{(\mathsf{Ax^{fu}_{m})}}
\newcommand{\axcontr}{(\mathsf{Ax_{c})}}
\newcommand{\axTrans}{(\mathsf{Ax_{trans})}}
\newcommand{\axleqgs}{(\mathsf{Ax_\leq^{gs})}}
\newcommand{\axleqfu}{(\mathsf{Ax_\leq^{fu})}}
\newcommand{\axleqfg}{(\mathsf{Ax_\leq^{fg})}}
\newcommand{\axlequs}{(\mathsf{Ax_\leq^{us})}}
\newcommand{\refgmpxAx}{\hyperref[gmpxAx]{$\axgmpx$}}
\newcommand{\refgmpx}{\hyperref[gmpxAx]{$\axgmpx$}}
\newcommand{\reffumpxAx}{\hyperref[fumpxAx]{$\axfumpx$}}
\newcommand{\reffumpx}{\hyperref[fumpxAx]{$\axfumpx$}}
\newcommand{\refcontrAx}{\hyperref[contrAx]{$\axcontr$}}
\newcommand{\refleqTrans}{\hyperref[leqTrans]{$\axTrans$}}
\newcommand{\refleqgs}{\hyperref[leqgs]{$\axleqgs$}}
\newcommand{\refleqfu}{\hyperref[leqfu]{$\axleqfu$}}
\newcommand{\refleqfg}{\hyperref[leqfg]{$\axleqfg$}}
\newcommand{\reflequs}{\hyperref[lequs]{$\axlequs$}}

\newcommand{\AIC}[1]{\AxiomC{\ensuremath{#1}}}
\newcommand{\ZIC}[1]{\AxiomC{}\UnaryInfC{\ensuremath{#1}}}
\newcommand{\UIC}[1]{\UnaryInfC{\ensuremath{#1}}}
\newcommand{\BIC}[1]{\BinaryInfC{\ensuremath{#1}}}
\newcommand{\TIC}[1]{\TrinaryInfC{\ensuremath{#1}}}
\newcommand{\QIC}[1]{\QuaternaryInfC{\ensuremath{#1}}}

\newcommand{\RL}[1]{\RightLabel{\ensuremath{#1}}}
\newcommand{\DP}{\DisplayProof}

\def\proofref#1{
\marginpar{\vspace{-0.2cm}
\colorbox{cyan}{\begin{minipage}{1.8cm}
{\scriptsize{#1}}
\end{minipage}}}}
\def\defref#1{
\marginpar{\vspace{-0.2cm}
\colorbox{lightgray}{\begin{minipage}{1.8cm}
{\scriptsize{#1}}
\end{minipage}}}}

\newtheorem{prop}{Proposition}
\newtheorem{defi}{Definition}
\newtheorem{lem}{Lemma}
\newtheorem{nota}{Notation}
\newtheorem{exa}{Example}
\newtheorem{coro}{Corollary}
\newtheorem{thm}{Theorem}
\newtheorem{rem}{Remark}

\newcommand*{\commEsaie}[1]{}
\newcommand*{\commAlex}[1]{}
\newcommand*{\todoa}[1]{}

\newcommand*{\leqg}{\leq_{\text{g}}}
\newcommand*{\leqf}{\leq_{\text{f}}}
\newcommand*{\lequ}{\leq_{\text{u}}}
\newcommand*{\ocg}{{\oc_{\text{g}}}}
\newcommand*{\ocf}{{\oc_{\text{f}}}}
\newcommand*{\ocu}{{\oc_{\text{u}}}}

\newcommand*{\muEAL}{\logsys{\ensuremath{\text{EAL}_\mu}}}

{}

\newcommand{\orig}[2]{
\begin{tikzpicture}[remember picture]
\node[inner sep=0pt,outer sep=1pt] (#1) {\ensuremath{#2}};\end{tikzpicture}}

\makeatletter
\tikzset{%
  remember picture with id/.style={%
    remember picture,
    overlay,
    save picture id=#1,
  },
  save picture id/.code={%
    \edef\pgf@temp{#1}%
    \immediate\write\pgfutil@auxout{%
      \noexpand\savepointas{\pgf@temp}{\pgfpictureid}}%
  },
  if picture id/.code args={#1#2#3}{%
    \@ifundefined{save@pt@#1}{%
      \pgfkeysalso{#3}%
    }{
      \pgfkeysalso{#2}%
    }
  }
}

\def\savepointas#1#2{%
  \expandafter\gdef\csname save@pt@#1\endcsname{#2}%
}

\def\tmk@labeldef#1,#2\@nil{%
  \def\tmk@label{#1}%
  \def\tmk@def{#2}%
}

\tikzdeclarecoordinatesystem{pic}{%
  \pgfutil@in@,{#1}%
  \ifpgfutil@in@%
    \tmk@labeldef#1\@nil
  \else
    \tmk@labeldef#1,(0pt,0pt)\@nil
  \fi
  \@ifundefined{save@pt@\tmk@label}{%
    \tikz@scan@one@point\pgfutil@firstofone\tmk@def
  }{%
  \pgfsys@getposition{\csname save@pt@\tmk@label\endcsname}\save@orig@pic%
  \pgfsys@getposition{\pgfpictureid}\save@this@pic%
  \pgf@process{\pgfpointorigin\save@this@pic}%
  \pgf@xa=\pgf@x
  \pgf@ya=\pgf@y
  \pgf@process{\pgfpointorigin\save@orig@pic}%
  \advance\pgf@x by -\pgf@xa
  \advance\pgf@y by -\pgf@ya
  }%
}

\newcommand\tikzmark[2][]{%
\tikz[remember picture with id=#2] #1;}
\makeatother

\makeatletter
\tikzset{%
  remember picture with id/.style={%
    remember picture,
    overlay,
    save picture id=#1,
  },
  save picture id/.code={%
    \edef\pgf@temp{#1}%
    \immediate\write\pgfutil@auxout{%
      \noexpand\savepointas{\pgf@temp}{\pgfpictureid}}%
  },
  if picture id/.code args={#1#2#3}{%
    \@ifundefined{save@pt@#1}{%
      \pgfkeysalso{#3}%
    }{
      \pgfkeysalso{#2}%
    }
  }
}

\def\savepointas#1#2{%
  \expandafter\gdef\csname save@pt@#1\endcsname{#2}%
}

\def\tmk@labeldef#1,#2\@nil{%
  \def\tmk@label{#1}%
  \def\tmk@def{#2}%
}

\tikzdeclarecoordinatesystem{pic}{%
  \pgfutil@in@,{#1}%
  \ifpgfutil@in@%
    \tmk@labeldef#1\@nil
  \else
    \tmk@labeldef#1,(0pt,0pt)\@nil
  \fi
  \@ifundefined{save@pt@\tmk@label}{%
    \tikz@scan@one@point\pgfutil@firstofone\tmk@def
  }{%
  \pgfsys@getposition{\csname save@pt@\tmk@label\endcsname}\save@orig@pic%
  \pgfsys@getposition{\pgfpictureid}\save@this@pic%
  \pgf@process{\pgfpointorigin\save@this@pic}%
  \pgf@xa=\pgf@x
  \pgf@ya=\pgf@y
  \pgf@process{\pgfpointorigin\save@orig@pic}%
  \advance\pgf@x by -\pgf@xa
  \advance\pgf@y by -\pgf@ya
  }%
}

\title{A uniform cut-elimination theorem for linear logics with fixed points and super exponentials} 

\titlerunning{Super exponentials with fixed-points} 

\author{Esaïe \textbf{B\footnotesize{AUER}} \& Alexis \textbf{S\footnotesize{AURIN}}}{Université Paris Cité \& CNRS \& INRIA, Pl. Aurélie Nemours,  75013 Paris, France}{alexis.saurin@irif.fr}{https://www.irif.fr/users/saurin/index}{}

\authorrunning{E. \textbf{B\footnotesize{AUER}} \& A. \textbf{S\footnotesize{AURIN}}} 

\Copyright{CC-BY} 


\category{} 

\relatedversion{} 

\nolinenumbers 

\EventEditors{John Q. Open and Joan R. Access}
\EventNoEds{2}
\EventLongTitle{42nd Conference on Very Important Topics (CVIT 2016)}
\EventShortTitle{CVIT 2016}
\EventAcronym{CVIT}
\EventYear{2016}
\EventDate{December 24--27, 2016}
\EventLocation{Little Whinging, United Kingdom}
\EventLogo{}
\SeriesVolume{42}
\ArticleNo{23}

\begin{document}

\maketitle

 \begin{abstract}
In the realm of light logics deriving from linear logic, a number of variants of exponential rules have been investigated. The profusion of such proof systems induces the need for cut-elimination theorems for each logic, the proof of which may be redundant. A number of approaches in proof theory have been adopted to cope with this need. In the present paper, we consider this issue from the point of view of enhancing linear logic with least and greatest fixed-points and considering such a variety of exponential connectives.

Our main contribution is to provide a uniform cut-elimination theorem for a parametrized system with fixed-points by combining two approaches: cut-elimination proofs by reduction (or translation) to another system and the identification of sufficient conditions for cut-elimination. 

More precisely, we examine a broad range of systems, taking inspiration from Nigam and Miller's subexponentials and from Bauer and Laurent's super exponentials. Our work is motivated on the one hand by Baillot's work on light logics with recursive types and on the other hand by Bauer and Saurin's recent work on the proof theory of the modal $\mu$-calculus.

\end{abstract}


\todo[inline,color=cyan]{A FAIRE:
- Finir intro

- faire schéma résultats

- discuter comparaison notre superLL et celui de BL: (i) pas de digging, pourquoi compliqué à modéliser avec points fixes? 
(ii) on a limité la forme des règles de contraction, mais uniquement pour simplifier la présentation (iii) on traite trois types de promotions, donc permet de modéliser LLL: donner un exemple avec LLL, peut-être SLL pour montrer le multiplexing en action? -> voir les exemples qu'a Esaïe.

- Dire que !u sert pour LLL?

- introduire full modal mu-calculus et faire noter qu'on étend donc le résultat de FOSSACS au multi-modal mu-calcul.




- discuter des extensions, notamment de BLL si on a les contractions plus élaborées de BL21, mais aussi de la question du digging et de traiter d'autres logiques modales de cette manière.}

\section{Introduction}

%

\noindent{\bf On the redundancy of cut-elimination proofs.}
While cut-elimination is certainly a cornerstone of structural proof theory since Gentzen's introduction of the sequent calculus, an annoying fact is that a slight change in a proof system induces the need to reprove globally the cut-elimination property. Such redundant new proofs are usually quite boring and fastidious, often lacking any new insight: cut-elimination results lack modularity. This results in the need of reestablishing a theorem which differ only very marginally from a previously proven one, even though the details are very technical and the failure of cut-elimination may hide in those small variants.
There are mainly two directions to try and make cut-elimination results more uniform, reduction and axiomatization:

\begin{description}
\item[Cut-elimination by reduction] The first option consists in proving a new cut-elimination result 
by means of translation between proof systems, allowing to reduce the cut-elimination property of a given system to that of another one for which the property is already known. Very frequent in term-calculi such as the variants of the $\lambda$-calculus, this approach is also applied in proof theory, for instance in translations between classical, intuitionistic and linear logics~\cite{DanosJS97, girard87} where linear translations come with simulation results. A more recent application of this approach is the second author's proof of cut-elimination for \muLLinf, the infinitary proof system for linear logic extended with least and greatest fixed-points, which is proved~\cite{TABLEAUX23} by a reduction to the cut-elimination property of the exponential-free fragment of the logic~\cite{CSL16}.
\item[Axiomatizing systems eliminating cuts]
The second option consists in abstracting properties ensuring that cut-elimination holds in a sequent calculus, and to provide sufficient conditions for cut-elimination to hold. For instance, after Miller and Nigam's work on subexponentials~\cite{Nigam09} providing a family of logics extending \LL{} with exponential admitting various structural rules, Bauer and Laurent provided a systematic and generic setting that captures most of the light logics to be found in the literature~\cite{lll,sll}, \superLL, for which they provided a uniform proof of cut-elimination based on an axiomatization stating a set of sufficient conditions for cut-elimination to hold~\cite{TLLA21}. Another line of work, more algebraic, establishing sufficient conditions for cut-elimination is that of Terui {\it et al.}\cite{Ciabattoni06,Terui_2007,Ciabattoni09,Terui_11,CIABATTONI12} which established modular and systematic cut-elimination results by combining methods from proof theory and algebra.
\end{description}

We will see in the present paper that the two approaches can be mixed in order to provide a uniform cut-elimination proof for a large family of logics that we call \musuperLLinf{} and that extends both \muLLinf{} and super exponentials: we shall obtain a single proof for a large class of proof systems and, by relying on a proof translation-method, we shall not need to design a new termination measure but we will simply rely on simulation results from one logic to another.
  
\medskip
\noindent{\bf Linear modal $\mu$-calculus.}
One of our motivations originated in a recent work, where we established a cut-elimination theorem for the classical modal $\mu$-calculus with infinite proofs~\cite{FOSSACS25}.
A key step in this work consisted in proving cut-elimination  of \muLLmodinf{}, a linear variant of the classical modal $\mu$-calculus,
to which we could reduce cut-elimination of the classical modal $\mu$-calculus. 
Indeed linear logic offers powerful tools for translating systems like \muLKinf{} from~\cite{TABLEAUX23} and \muLKmodinf{}~\cite{DBLP:journals/tcs/Kozen83} into linear systems making the transfer of properties of those system to others logics efficient. 
Proving cut-elimination for  \muLLmodinf{} we were led to consider a more systematic treatment of exponentials and modalities revisiting a previous work by the first author with Laurent~\cite{TLLA21} and introducing  \musuperLLinf{}.


\medskip

\noindent{\bf Light logics with least and greatest fixed points.}
Taming the deductive power of linear logic's exponential connectives allows one to get complexity bounds on the cut-elimination process~\cite{lll, sll}. Adding fixed points in such logic enriches the study of complexity classes~\cite{BAILLOT2015, BrunelTerui10, RoversiVercelli10, dal2006light}, as well as the study of light $\lambda$-calculus enriched with fixpoints as in~\cite{BAILLOT18}.

In~\cite{BAILLOT2015}, enriching \emph{elementary affine logic} with fixed points allows one to refine the complexity results from \ELL{}, and to characterize a hierarchy of the elementary complexity classes. In~\cite{Nguyen19}, it is even shown that the fixed-point-free version of this logic gets a very different characterization of complexity bounds for similar types.

The systems  defined in the present article differ from those discussed in the previous paragraph: they are based on recursive types rather than extremal fixed-points (ie. inductive and coinductive types), we base our study on potentially infinite and regular derivation trees, etc. 
However, both systems have strong similarities that we shall discuss in a later section,
which makes a stronger link between our systems and light systems from the literature.

\medskip

\noindent{\bf Organization and contributions of the paper.}
The main contribution of this paper is a syntactic cut-elimination result for a large class of (parametrized) linear systems wit least and greatest fixed-points coming with a notion of non-wellfounded and regular proofs. 
In \Cref{section:background}, we recall some definitions and results about infinitary rewriting theory and linear logic; we also give definitions of a variant of Bauer and Laurent's system of super exponentials~\cite{TLLA21}. 
We set up in \Cref{section:musuperll} a parametrized system, \musuperLLinf{}, which is \superLL{} extended with fixed-points and non-wellfounded proofs. 
Finally, in \Cref{section:cut-elimination}, we define the cut reduction system that achieve syntactic cut-elimination and provide the proof of our main theorem, the cut-elimination of \musuperLLinf{}, through an encoding into \muLLinf{}.

\todo[inline,color=green]{Prevoir une schéma qui décrit musuperLL, superLL, quelques instances et les traductions? Pour illustrer les résultats de l'article???

Dire aussi qu'on discute des instances que capture notre système.}

    \section{Background on \LL, fixed-points and non-wellfounded proofs}
\label{section:background}


In this paper, we will study proof theory of different systems of linear logic (\LL). It is much more convenient to work on one-sided sequents systems as proofs as well as the description of these systems are more compact than the two-sided version.
However, The results for the two-sided systems can be retrieved systematically from the one-sided systems with translations between them as in~\cite{TABLEAUX23} for instance.

\subsection{Formulas, sequent calculi and non-wellfounded proofs}
\label{section:definitionsOfKnownSequentCalculi}
\label{sequentCalculiSection}

Let $\Var$ and $\Atom$ be two disjoint sets of \emph{fixed-point variables} and \emph{atomic formulas} respectively.
The \emph{(pre-)formulas} of linear logic with fixed-points are defined inductively as ($a\in\Atom, X\in\Var$):\\
$F, G \hspace{0.5em}::= \hspace{0.5em} a  \mid a^\perp \mid X \mid \mu X. F \mid \nu X. F \mid F \parr G \mid F \otimes G \mid \bot \mid 1 \mid F \oplus G \mid F \with G \mid 0 \mid \top \mid \wn F \mid \oc F.$\\
\emph{Formulas} of \muLLinf{} are such closed pre-formulas ($\mu$ and $\nu$ being binders for variables in $\Var$).
By considering the $\mu, \nu, X$-free formulas of this system, we get \LL{}, the usual formulas of linear logic~\cite{girard87}. By considering the $\oc, \wn$-free formulas of it, we get the formulas $\muMALLinf{}$ \emph{the multiplicative and additive linear logic with fixed points}~\cite{CSL16}. By considering the intersection of these two subset of formulas, we get the formulas of \MALL{} the \emph{multiplicative and additive linear logic}. The $\wn, \oc$-fragment is called \emph{the exponential fragment} of linear logic.

\begin{defi}[Negation]\label{negationDef}
We define $(-)^\perp$ to be the involution on formulas satisfying:\\
$\begin{array}{clclclcl}
\bot^\perp & = 1~~ & X^\perp & = X~~ & (A_1\otimes A_2)^\perp & = A_1^\perp\parr A_2^\perp~~ & (A_1\with A_2)^\perp & = A_1^\perp \oplus A_2^\perp \\[-2pt]
\top^\perp & = 0 & {a^\perp}^\perp & = a & (\mu X. F)^\perp & = \nu X. F^\perp & (\wn F)^\perp & = \oc F^\perp
\end{array}$
\end{defi}

The sequent calculi that we consider in this paper are built on one-sided sequents: A \defname{sequent} is a list of formulas $\A$, that we usually write $\vdash \A$.
Usually, in the literature, derivation rules are defined as a scheme of one \defname{conclusion sequent} and a (possibly empty) list of \defname{hypotheses sequents}.
In our system, the derivation rules come equipped with an \defname{ancestor relation} linking each formula in the conclusion to zero, one or several formulas of the hypotheses.
When defining our rules, we provide this link by drawing the ancestor relation with colors. (See \Cref{fig:MALLrules,fig:ExponentialRules,fig:fixFragment}.)
As usual, some formulas may be distinguished as \defname{principal formulas}: both formulas in the conclusion of an axiom rule are principal, no formula is principal in the conclusion of an $(\exch)$ or $(\cut)$ inference while in other rules of \Cref{fig:MALLrules,fig:ExponentialRules,fig:fixFragment} the leftmost occurrence of each conclusion sequent is principal.

\begin{defi}[\MALL, \LL{} and \muLLinf{} inference rules]
\Cref{fig:MALLrules} defines \MALL{} inference rules.
\begin{figure}
\centering
\scalebox{0.9}{
\hspace{-1cm}$
\begin{array}{c}
\AIC{}
\RL{\ax}
\UIC{\vdash F, F^{\bot}}
\DP\qquad
\AIC{\vdash F, \tikzmark{cutc12}\A}
\AIC{\vdash F^{\bot}, \B\tikzmark{cutc22}}
\RL{\cut}
\BIC{\vdash \tikzmark{cutc11}\A, \B\tikzmark{cutc21}}
\DP\qquad
\AIC{\vdash \tikzmark{llexchc3}\A, \tikzmark{llexchf3}G, F\tikzmark{llexchf4}, \B\tikzmark{llexchc4}}
\RL{\exch}
\UIC{\vdash \tikzmark{llexchc1}\A, \tikzmark{llexchf1}F, \tikzmark{llexchf2}G, \B\tikzmark{llexchc2}}
\DP\qquad
\AIC{\vdash \tikzmark{parrf2}F, \tikzmark{parrf3}G, \tikzmark{parrc2}\A}
\RL{\parr}
\UIC{\vdash \tikzmark{parrf1}F\parr G, \tikzmark{parrc1}\A}
\DP \qquad
\AIC{\vdash \tikzmark{otimesf2}F, \B_1\tikzmark{otimesc12}}
\AIC{\vdash \tikzmark{otimesf3}G, \B_2\tikzmark{otimesc22}}
\RL{\otimes}
\BIC{\vdash \tikzmark{otimesf1}F\otimes G, \tikzmark{otimesc11}\B_1, \tikzmark{otimesc21}\B_2}
\DP
\\[2ex]
\AIC{\vdash \tikzmark{oplus1f2}F_1, \tikzmark{oplus1c2}\A}
\RL{\oplus^1}
\UIC{\vdash \tikzmark{oplus1f1}F_1 \oplus F_2, \tikzmark{oplus1c1}\A}
\DP
\qquad 
\AIC{\vdash \tikzmark{oplus2f2}F_2, \tikzmark{oplus2c2}\A}
\RL{\oplus^2}
\UIC{\vdash \tikzmark{oplus2f1} F_1 \oplus F_2, \tikzmark{oplus2c1}\A}
\DP
\qquad
\AIC{\vdash \tikzmark{withf2}F_1, \tikzmark{withc2}\A}
\AIC{\vdash \tikzmark{withf3}F_2, \tikzmark{withc3}\A}
\RL{\with}
\BIC{\vdash \tikzmark{withf1}F_1 \with F_2, \tikzmark{withc1}\A}
\DP\qquad
\AIC{}
\RL{1}
\UIC{\vdash 1}
\DP \qquad
\AIC{\vdash \tikzmark{botc2}\A}
\RL{\bot}
\UIC{\vdash \bot, \tikzmark{botc1}\A}
\DP
\qquad
\AIC{}
\RL{\top}
\UIC{\vdash \top, \A}
\DP
\end{array}
\begin{tikzpicture}[overlay,remember picture,-,line cap=round,line width=0.1cm]
   \draw[rounded corners, smooth=2,red, opacity=.4] ([xshift=1mm] pic cs:cutc11) to ([yshift=2mm] pic cs:cutc12);
   \draw[rounded corners, smooth=2,red, opacity=.4] ([xshift=-2mm] pic cs:cutc21) to ([yshift=2mm] pic cs:cutc22);
   \draw[rounded corners, smooth=2,red, opacity=.4] ([xshift=1mm] pic cs:llexchc1) to ([yshift=2mm] pic cs:llexchc3);
   \draw[rounded corners, smooth=2,red, opacity=.4] ([xshift=-1mm, yshift=1mm] pic cs:llexchc2) to ([xshift=-1mm,yshift=1mm] pic cs:llexchc4);
   \draw[rounded corners, smooth=2,green, opacity=.4] ([xshift=1mm] pic cs:llexchf1) to ([yshift=2mm] pic cs:llexchf4);
   \draw[rounded corners, smooth=2,green, opacity=.4] ([xshift=1mm] pic cs:llexchf2) to ([yshift=2mm] pic cs:llexchf3);
   \draw[rounded corners, smooth=2,red, opacity=.4] ([xshift=1mm] pic cs:botc1) to ([yshift=2mm] pic cs:botc2);
   \draw[rounded corners, smooth=2,green, opacity=.4] ([xshift=4mm] pic cs:otimesf1) to ([yshift=2mm] pic cs:otimesf2);
   \draw[rounded corners, smooth=2,green, opacity=.4] ([xshift=4mm] pic cs:otimesf1) to ([xshift=2mm, yshift=2mm] pic cs:otimesf3);
   \draw[rounded corners, smooth=2,red, opacity=.4] ([xshift=1mm] pic cs:otimesc11) to ([xshift=-2mm, yshift=2mm] pic cs:otimesc12);
   \draw[rounded corners, smooth=2,red, opacity=.4] ([xshift=1mm] pic cs:otimesc21) to ([yshift=2mm] pic cs:otimesc22);
   \draw[rounded corners, smooth=2,green, opacity=.4] ([xshift=4mm] pic cs:parrf1) to ([yshift=2mm] pic cs:parrf2);
   \draw[rounded corners, smooth=2,green, opacity=.4] ([xshift=4mm] pic cs:parrf1) to ([xshift=2mm, yshift=2mm] pic cs:parrf3);
   \draw[rounded corners, smooth=2,red, opacity=.4] ([xshift=1mm] pic cs:parrc1) to ([yshift=2mm] pic cs:parrc2);
   \draw[rounded corners, smooth=2,green, opacity=.4] ([xshift=4mm] pic cs:oplus1f1) to ([xshift=2mm, yshift=2mm] pic cs:oplus1f2);
   \draw[rounded corners, smooth=2,red, opacity=.4] ([xshift=1mm] pic cs:oplus1c1) to ([yshift=2mm] pic cs:oplus1c2);
   \draw[rounded corners, smooth=2,green, opacity=.4] ([xshift=4mm] pic cs:oplus2f1) to ([xshift=2mm, yshift=2mm] pic cs:oplus2f2);
   \draw[rounded corners, smooth=2,red, opacity=.4] ([xshift=1mm] pic cs:oplus2c1) to ([yshift=2mm] pic cs:oplus2c2);
   \draw[rounded corners, smooth=2,green, opacity=.4] ([xshift=4mm] pic cs:withf1) to ([yshift=2mm] pic cs:withf2);
   \draw[rounded corners, smooth=2,green, opacity=.4] ([xshift=4mm] pic cs:withf1) to ([xshift=2mm, yshift=2mm] pic cs:withf3);
   \draw[rounded corners, smooth=2,red, opacity=.4] ([xshift=1mm] pic cs:withc1) to ([yshift=2mm] pic cs:withc2);
   \draw[rounded corners, smooth=2,red, opacity=.4] ([xshift=1mm] pic cs:withc1) to ([xshift=1mm, yshift=2mm] pic cs:withc3);
\end{tikzpicture}
$
}
\caption{one-sided \MALL{} rules}\label{fig:MALLrules}
\end{figure}
 \LL{} inferences are obtained by considering  \Cref{fig:MALLrules,fig:ExponentialRules}.
\begin{figure}
\centering
$
\AIC{\vdash \tikzmark{llwkc2}\A}
\RL{\wnwk}
\UIC{\vdash \wn F, \tikzmark{llwkc1}\A}
\DP
\quad
\AIC{\vdash \tikzmark{llcontrf2}\wn F, \tikzmark{llcontrf3}\wn F, \tikzmark{llcontrc2}\A}
\RL{\wncontr}
\UIC{\vdash \tikzmark{llcontrf1}\wn F, \tikzmark{llcontrc1}\A}
\DP
\quad
\AIC{\vdash \tikzmark{lldef2}F, \tikzmark{lldec2}\A}
\RL{\wnde}
\UIC{\vdash \tikzmark{lldef1}\wn F, \tikzmark{lldec1}\A}
\DP
\quad
\AIC{\vdash \tikzmark{llpromf2}F, \tikzmark{llpromc2}\wn\A}
\RL{\ocprom}
\UIC{\vdash \tikzmark{llpromf1}\oc F, \tikzmark{llpromc1}\wn\A}
\DP
$
\begin{tikzpicture}[overlay,remember picture,-,line cap=round,line width=0.1cm]
   \draw[rounded corners, smooth=2,red, opacity=.4] ([xshift=1mm, yshift=0mm] pic cs:llwkc1) to ([xshift=1mm, yshift=2mm] pic cs:llwkc2);
   \draw[rounded corners, smooth=2,green, opacity=.4] ([xshift=1mm, yshift=0mm] pic cs:llcontrf1) to ([xshift=2mm, yshift=2mm] pic cs:llcontrf2);
   \draw[rounded corners, smooth=2,green, opacity=.4] ([xshift=1mm, yshift=0mm] pic cs:llcontrf1) to ([xshift=2mm, yshift=2mm] pic cs:llcontrf3);
   \draw[rounded corners, smooth=2,red, opacity=.4] ([xshift=0mm, yshift=0mm] pic cs:llcontrc1) to ([xshift=1mm, yshift=2mm] pic cs:llcontrc2);
   \draw[rounded corners, smooth=2,green, opacity=.4] ([xshift=2mm, yshift=0mm] pic cs:lldef1) to ([xshift=1mm, yshift=2mm] pic cs:lldef2);
   \draw[rounded corners, smooth=2,red, opacity=.4] ([xshift=1mm, yshift=0mm] pic cs:lldec1) to ([xshift=1mm, yshift=2mm] pic cs:lldec2);
   \draw[rounded corners, smooth=2,green, opacity=.4] ([xshift=2mm, yshift=0mm] pic cs:llpromf1) to ([xshift=1mm, yshift=2mm] pic cs:llpromf2);
   \draw[rounded corners, smooth=2,red, opacity=.4] ([xshift=2mm, yshift=0mm] pic cs:llpromc1) to ([xshift=2mm, yshift=2mm] pic cs:llpromc2);
\end{tikzpicture}
\caption{one-sided exponential fragment of \LL}\label{fig:ExponentialRules}
\end{figure}
Finally, inference rules for \muMALLinf{} and \muLLinf{} are obtained by adding rules of  \Cref{fig:fixFragment} to \MALL{} and \LL{} inferences.
\begin{figure}[t]
    \centering
$
\AIC{\vdash \tikzmark{muf2}F[X:=\mu X.F], \tikzmark{muc2}\A}
\RL{\mu}
\UIC{\vdash \tikzmark{muf1}\mu X.F, \tikzmark{muc1}\A}
\DP\qquad\qquad
\AIC{\vdash \tikzmark{nuf2}F[X:=\nu X.F], \tikzmark{nuc2}\A}
\RL{\nu}
\UIC{\vdash \tikzmark{nuf1}\nu X.F,\tikzmark{nuc1}\A}
\DP
$
\begin{tikzpicture}[overlay,remember picture,-,line cap=round,line width=0.1cm]
\draw[rounded corners, smooth=2,green, opacity=.4] ([xshift=4mm] pic cs:muf1) to ([xshift=2mm, yshift=2mm] pic cs:muf2);
   \draw[rounded corners, smooth=2,red, opacity=.4] (pic cs:muc1) to ([xshift=2mm, yshift=2mm] pic cs:muc2);
   \draw[rounded corners, smooth=2,green, opacity=.4] ([xshift=4mm] pic cs:nuf1) to ([xshift=2mm, yshift=2mm] pic cs:nuf2);
   \draw[rounded corners, smooth=2,red, opacity=.4] (pic cs:nuc1) to ([xshift=2mm, yshift=2mm] pic cs:nuc2);   
\end{tikzpicture}
    \caption{Rules for the fixed-point fragment\label{fig:fixFragment}}
\end{figure}
\end{defi}

In the rest of the article, we will not write the exchange rules explicitly: one can assume that every rule is preceded and followed by a finite number of instances of $(\exch)$.
While proofs for \MALL{} and \LL{} are the usual trees inductively generated by the inference rules, defining non-wellfounded proofs for fixed-point logics requires some definitions:

\begin{defi}[Pre-proofs]
Given a set of derivation rules, we define \defname{pre-proofs} to be the trees co-inductively generated by rules of each of those systems.
\defname{Regular (or circular) pre-proofs} are those pre-proofs having a finite number of sub-proofs. 
\end{defi}

We represent regular proofs with back-edges as in the following example:

\begin{exa}[Regular proof]
    \label{exa:regProof}
We give an example of circular proof:
\hfill $
\AIC{\vdash \nu X. \oc X, \orig{circularExas}{\wn 0}}
\RL{\ocprom}
\UIC{\vdash \oc \nu X. \oc X, \wn 0}
\RL{\nu}
\UIC{\vdash \nu X. \oc X, \orig{circularExat}{\wn 0}}
\DP
$
\begin{tikzpicture}[remember picture,overlay]
     \draw [->,>=latex] ([yshift=1mm] circularExas.east) .. controls +(15:2cm) and +(-11:2cm) .. (circularExat.east);
    \end{tikzpicture}
\end{exa}



From that, we define the proofs as a subset of the pre-proofs:
\begin{defi}[Validity and proofs]
Let $b=(s_i)_{i\in\omega}$ be a sequence of sequents defining an infinite branch in a pre-proof $\pi$.
A \defname{thread} of $b$ is a sequence $(F_i\in s_i)_{i>n}$ of formula occurrences such that for each $j$, $F_j$ and $F_{j+1}$ are satisfying the ancestor relation.
We say that a thread of $b$ is \defname{valid} if the minimal recurring formula
of this sequence, for sub-formula ordering, exists and is a $\nu$-formula and that the formulas of this threads are infinitely often principal. A branch $b$ is \defname{valid} if there exists a valid thread of $b$.
A pre-proof is \defname{valid} and is a \defname{proof} if each of its infinite branches is valid.
\end{defi}

\begin{exa}
\label{ex:prom-encoding}
Given a formula $A$, let us consider $\trans{\wn} A = \mu X. ({A} \oplus (\bot \oplus (X\parr X)))$ and $\trans{\oc} A = \nu X. ({A} \with (1 \with (X\otimes X)))$.
Assuming a context $\Gamma$ and a valid proof $\pi$ of $\vdash A, \wn\Gamma$, the following is a valid proof of  $\vdash \trans{\oc} A, \wn\Gamma$:\\
(In every infinite branch along\\ 
the 2 back-edges, $\trans{\oc} A$ is the \\
minimal recurring formula.)

\vspace{-1.5cm}~\hfill $
\AIC{\pi}
\noLine
\UIC{\vdash A, \wn \Gamma}
\AIC{}
\RL{1}
\UIC{\vdash 1}
\RL{\wnwk^\star}
\UIC{\vdash 1, \wn\Gamma}
\AIC{\vdash \trans{\oc} A, \orig{circularProms1}{\wn\Gamma}}
\AIC{\vdash \trans{\oc} A, \orig{circularProms2}{\wn\Gamma}}
\RL{\otimes}
\BIC{\vdash \trans{\oc} A \otimes \trans{\oc} A, \wn\Gamma,  \wn\Gamma}
\RL{\wncontr^\star}
\UIC{\vdash \trans{\oc} A \otimes \trans{\oc} A, \wn\Gamma}
\RL{\with,\with}
\TIC{\vdash {A} \with (1 \with (\trans{\oc} A \otimes \trans{\oc} A)), \wn\Gamma}
\RL{\nu}
\UIC{\vdash \trans{\oc}A, \orig{circularPromt}{\wn \Gamma}}
\DP
$
\begin{tikzpicture}[remember picture,overlay]
     \draw [->,>=latex] ([yshift=1mm] circularProms1.east) .. controls +(16:5.5cm) and +(-11:5cm) .. (circularPromt.east);
     \draw [->,>=latex] ([yshift=1mm] circularProms2.east) .. controls +(10:2.4cm) and +(-11:4cm) .. ([yshift=1mm] circularPromt.east);
    \end{tikzpicture}
\end{exa}

\subsection{Cut-elimination for linear logic with fixed-point}
\label{multicutdef}
Cut-elimination holds for \muMALLinf{} and \muLLinf{} in the form of the infinitary weak normalization of a multicut-reduction relation: a new rule, the \defname{multicut (\mcut)}, is introduced, that corresponds to an abstraction of several cuts. \defref{Details in appendix~\ref{app:multicutdef}.}
This rule has an arbitrary number of premises:\quad
    $
    \AIC{\vdash \A_1}
    \AIC{\dots}
    \AIC{\vdash \A_n}
    \RL{\rmcutpar}
    \TIC{\vdash\A}
    \DP
    $
    and it is parameterized by two relations: (i) the \emph{ancestor relation} $\iota$ which relates each formula of the conclusion to exactly one formula among the hypotheses and (ii) \emph{the multicut relation}, $\cutrel$, which links {\it cut-formulas} together. $\iota$ and $\cutrel$ are subject to a number of conditions detailed in \Cref{app:multicutdef}.

    
    \begin{exa}
    \label{ex:mcut}
Representing $\iota$ and $\cutrel$ in red and blue, the (\cut/\mcut) step is as follows:

{\footnotesize
        $
        \AIC{\vdash A\tikzmark{prePC1}, B}
        \AIC{\vdash B^\bot,C}
        \AIC{\vdash C^\bot,D}
        \RL{\cut}
        \BIC{\vdash B^\bot\tikzmark{prePC3}, D}
        \RL{\rmcutpar}
        \BIC{\vdash A\tikzmark{preCC},D}
        \DP
        ~\rightsquigarrow~
        \AIC{\vdash A\tikzmark{PC1}, B}
        \AIC{\vdash B^\bot\tikzmark{PC2},C}
        \AIC{\vdash C^\bot\tikzmark{PC3},D}
        \RL{\rmcutparprime}
        \TIC{\vdash A\tikzmark{CC},D}
        \DP
        \begin{tikzpicture}[overlay,remember picture,-,line cap=round,line width=0.1cm]
        \draw[rounded corners, smooth=2,red, opacity=.25] ($(pic cs:CC)+(-.2cm,.1cm)$)to ($(pic cs:PC1)+(-.2cm,-.1cm)$)to ($(pic cs:PC1)+(-.2cm,.1cm)$); 
        \draw[rounded corners, smooth=2,red, opacity=.25] ($(pic cs:CC)+(.3cm,.1cm)$)to ($(pic cs:PC3)+(.3cm,-.1cm)$)to ($(pic cs:PC3)+(.3cm,.1cm)$); 
        \draw[rounded corners, smooth=2,cyan, opacity=.25] ($(pic cs:PC1)+(.25cm,.1cm)$)to ($(pic cs:PC1)+(.25cm,-.1cm)$)to ($(pic cs:PC2)+(-.4cm,-.1cm)$) to ($(pic cs:PC2)+(-.4cm,.1cm)$);
            \draw[rounded corners, smooth=2,cyan, opacity=.25] ($(pic cs:PC2)+(.3cm,.1cm)$)to ($(pic cs:PC2)+(.3cm,-.1cm)$)to ($(pic cs:PC3)+(-.4cm,-.1cm)$) to ($(pic cs:PC3)+(-.4cm,.1cm)$);  
        
        \draw[rounded corners, smooth=2,red, opacity=.25] ($(pic cs:preCC)+(-.2cm,.1cm)$)to ($(pic cs:prePC1)+(-.2cm,-.1cm)$)to ($(pic cs:prePC1)+(-.2cm,.1cm)$); 
        \draw[rounded corners, smooth=2,red, opacity=.25] ($(pic cs:preCC)+(.3cm,.1cm)$)to ($(pic cs:prePC3)+(.3cm,-.1cm)$)to ($(pic cs:prePC3)+(.3cm,.1cm)$); 
        \draw[rounded corners, smooth=2,cyan, opacity=.25] ($(pic cs:prePC1)+(.25cm,.1cm)$)to ($(pic cs:prePC1)+(.25cm,-.1cm)$)to ($(pic cs:prePC3)+(-.4cm,-.1cm)$) to ($(pic cs:prePC3)+(-.4cm,.1cm)$);
        \end{tikzpicture}
        $
        }
        \end{exa}

To define the (\mcut) reduction step we need a last definition, that will be also useful when defining the reduction step of the super exponential system:
\begin{defi}[Restriction of a multicut context]\defref{Details in appendix~\ref{app:multicutRestriction}}
    \label{multicutRestriction}
    Let $
    \AIC{\C}
    \RL{\rmcutpar}
    \UIC{s}
    \DP
    $ be a multicut occurrence with $\C = s_1\dots s_n$ and $s_i$ be $\vdash F_1 \dots F_{k_i}$. 
    For $1\leq j \leq k_i$, $\C_{F_j}$ 
    is the restriction of $\C$ to the sequents
     hereditarily linked to $F_j$ with the $\cutrel$-relation.    
\end{defi}

The previous definition extends to contexts, writing $\C_{F_1\dots F_n}$.
For instance, writing $\C$ for the premises of the rightmost mcut in \Cref{ex:mcut}, 
$\C_{B^\perp} = \{\vdash A,B ; \vdash C^\perp, D\}$ while $\C_A=\emptyset$.

\smallskip
Cut-elimination for \muMALLinf{} and \muLLinf{} is proved syntactically with a rewriting system on proof with (\mcut), whose steps are given in appendix~\ref{app:mumallonestep}. As standard in sequent caluli, those (m)cut-reduction steps are divided in principal cases and (m)cut-commutation cases.

The cut elimination result is then stated as a strong normalization result for a class of infinitary reduction, initiated with proofs containing exactly one (\mcut) at the root of the proof. 
Indeed, strong normalization is trivially lost in such infinitary settings as one can always build infinite sequences that never activate some (\mcut), thus converging to a non cut-free proof.
\emph{Fair reductions} precisely prevent this situation by asking that \emph{no (\mcut) that can be activated remains forever inactive along the reduction sequence.} The following definition is borrowed from~\cite{LICS22,CSL16}, residuals corresponding to the usual notion of TRS~\cite{TERESE}:
\begin{defi}
    A reduction sequence $(\pi_i)_{i\in\omega}$ is \emph{fair}, if for each $\pi_i$ such that there is a reduction $\mathcal{R}$ to a proof $\pi'$, there exist a $j>i$ such that $\pi_j$ does not contain any residual of $\mathcal{R}$.
\end{defi}

This fairness condition allowed Baelde {\it et al.}~\cite{LICS22, CSL16} to obtain a (multi)cut-elimination result for \muMALLinf{} which, combined with the following encoding of exponential formulas using notations from \Cref{ex:prom-encoding}, $\trans{(\wn A)} = \trans{\wn}\trans{A}$ and $\trans{(\oc A)} = \trans{\oc}\trans{A}$ (extended to proof and cut-reduction steps),
 induces the following  \muLLinf{} multicut-elimination result~\cite{TABLEAUX23}:

%
%
\begin{thm}
    \label{thm:mullcutelim}
Every fair  \muLLinf{}  (\mcut)-reduction sequence converges to a cut-free proof.
\end{thm}

\section{Super exponentials}
\label{secsuperLL}

In this section, we define a family of parameterized logical systems, adapting the methodology of~\cite{TLLA21} and using the sequent formalism from the previous section. Consequently, the section lies in between background on the work by the first author and Laurent and new material since we propose an alternative system, with an alternative choice of formalization. We discuss briefly some of these differences here and shall come back to this comparison in the discussion of related works. 
Bauer and Laurent's super exponentials~\cite{TLLA21} only include \emph{functorial promotion} and rely on the so-called \emph{digging} rule to recover the usual \emph{Girard's promotion} rule.
On the other hand, we propose below another formalization of super exponentials, adapting the system to capture both functorial and Girard's promotions primitively while we discard the digging which is not needed nor well-suited for the extension we aim with fixed-points.

This means that the general philosophy of this section follows that of~\cite{TLLA21} and in particular we show how their proofs can be adapted to the present setting in~\ref{app:superllcutelim}. On the other hand, we will show in \Cref{section:cut-elimination} that our uniform cut-elimination theorem provides an alternative, completely new, proof of cut-elimination for the super exponential of the present section in the  sense that it does not rely on adapting the techniques and proof by the first author and Laurent.
The first parameters of these systems will allow us to define formulas:

\begin{defi}[Superexponential formulas]
Let $\Sig$ be a set. \defname{Formulas of $\superLL{}(\Sig)$} are the formulas of \MALL{} together with exponential connectives subscripted by an element $\e\in\Sig$:

\quad$F, G \hspace{0.5em}::= \hspace{0.5em} a\in\Atom \mid a^\perp \mid F \parr G \mid F \otimes G \mid \bot \mid 1 \mid F \oplus G \mid F \with G \mid 0 \mid \top \mid \wn_\e F \mid \oc_\e F.$

Elements of $\Sig$ are called \defname{exponential signatures}. The orthogonal $(-)^\perp$ is defined as the involution satisfying extending that of \Cref{negationDef} with: $ (\oc_\e A)^\perp = \wn_\e A^\perp$ for any $\e\in \Sig$.
\end{defi}

\begin{nota}[List of exponential signatures]
Let $\B=A_1 \dots A_n$ be a list of $n$ formulas and $\vec{\e}=\e_1 \dots \e_n$ a list of $n$ exponential signatures. 
The list of formulas $\wn_{\e_1}A_1 \dots \wn_{\e_n}A_n$ is written  $\wn_{\vec{\e}}\B$.
Moreover, given a binary relation $R$ on exponential signatures and two lists of exponential signatures $\vec{\e}=\e_1, \dots, \e_m$ and $\vec{\e'}=\e'_1,\dots, \e'_n$, we write $\vec{\e} \mathop{R} \vec{\e'}$ for $\bigwedge\limits_{\substack{1\leq i\leq m\\ 1\leq j \leq n}} \e_i \mathop{R} \e'_j$.
\end{nota}

While each element of $\e\in\Sig$ induces two exponential modalities, $\wn_\e,\oc_\e$, the inference rules will be described in two phases: first each $\e\in\Sig$ will be equipped with a set of rule names $\{\mpx{i}\mid i\in\mathbb{N}\}\cup \{\contr{i}\mid i\geq 2\}$ which can be used to introduce the connective $\wn_\e$.
Second, some binary relations over $\Sig$ will govern the available promotion rules, introducing $\oc_\e$.

\begin{defi}
The set of \defname{exponential rule names} is ${\Exp} = \{\mpx{i}\mid i\in\mathbb{N}\}\cup \{\contr{i}\mid i\geq 2\}$.

To each \defname{exponential signature} $\e\in \Sig$, one associates a subset of $\Exp$, $\Sigfun{\e}$.
\end{defi}

%

For the sake of clarity, given $\e\in \Sig$ we will write (when unambiguous) $\e$ instead of $\Sigfun{\e}$, omitting $\Sigfun{\cdot}$ throughout the paper. 
We shall also switch freely from viewing $\e$ (more precisely, $\Sigfun{\e}$) as a subset of $\Exp$ or as its boolean characteristic function, write, for instance, $\mpx{i}\in\e$ (resp. $\contr{i}\in\e$) when convenient, or considering $\e(\mpx{i})$ (resp. $\e(\contr{i})$) as a truth value.

\begin{defi}
For one set of signatures $\Sig$, we define many systems, parameterized by three binary relations on $\Sig$: $\leqg, \leqf$ and $\lequ$.
Rules for this system are the rules of \MALL{} from \Cref{fig:MALLrules} in combination with the super exponential rules of \Cref{fig:musuperllexprules}: multiplexing ($\mpx{i}$), contraction ($\contr{i}$) as well as functorial ($\ocf$), Girard ($\ocg$) and unary ($\ocu$) promotions. 

Each exponential rule comes with a side-condition written to the right of the premises.
\end{defi}

\begin{rem} Below, the side-condition for an exponential rule may also be written next to the rule label or simply omitted when it has been checked elsewhere. Those side-conditions are not part of the proof-object itself: all exponential inferences are {\it unary rules}.

Note that nullary multiplexing rule corresponds to usual weakening $(\wnwk)$ and unary multiplexing corresponds to dereliction $(\wnde)$.
\end{rem}

\begin{figure*}
    \centering
    \hspace*{-2.5cm}$
    {\scriptsize\AIC{\vdash \overbrace{\tikzmark{sllmpxf2}A, \dots, \tikzmark{sllmpxf3}A}^i, \tikzmark{sllmpxc2}\A ~~(\e(\mpx{i}))}
    \RL{\mpx{i}}
    \UIC{\vdash \tikzmark{sllmpxf1} \wn_\e A, \tikzmark{sllmpxc1}\A}
    \DP
    ~~
    \AIC{\vdash \overbrace{\tikzmark{sllcontrf2}\wn_\e A, \dots, \wn_\e A\tikzmark{sllcontrf3}}^i, \tikzmark{sllcontrc2}\A 
    ~~ (\e(\contr{i}))}
    \RL{\contr{i}}
    \UIC{\vdash \tikzmark{sllcontrf1}\wn_\e A, \tikzmark{sllcontrc1}\A}
    \DP
    ~~
    \AIC{\vdash \tikzmark{sllocgf2}A, \wn_{\vec{\e'}}\B\tikzmark{sllocgc3} ~~ (\e\leqg\vec{\e'})}
	\RL{\ocg}
    \UIC{\vdash\tikzmark{sllocgf1}\oc_\e A, \tikzmark{sllocgc1}\wn_{\vec{\e'}}\B}
    \DP
    ~~
    \AIC{\vdash A\tikzmark{sllocff2}, \tikzmark{sllocfc3}\B ~~ (\e\leqf\vec{\e'})}
	\RL{\ocf}
    \UIC{\vdash \oc_\e A\tikzmark{sllocff1}, \tikzmark{sllocfc1}\wn_{\vec{\e'}}\B}
    \DP
    ~~
    \AIC{\vdash A\tikzmark{sllocuf2}, \tikzmark{sllocuc2}B ~~ (\e_1\lequ \e_2)}
    \RL{\ocu}
    \UIC{\vdash  \oc_{\e_1} A\tikzmark{sllocuf1}, \tikzmark{sllocuc1}\wn_{\e_2} B}
    \DP}
    $
    \begin{tikzpicture}[overlay,remember picture,-,line cap=round,line width=0.1cm]
   \draw[rounded corners, smooth=2,green, opacity=.4] ([xshift=1mm] pic cs:sllmpxf1) to ([yshift=2mm] pic cs:sllmpxf2);
   \draw[rounded corners, smooth=2,green, opacity=.4] ([xshift=1mm] pic cs:sllmpxf1) to ([yshift=2mm] pic cs:sllmpxf3);
   \draw[rounded corners, smooth=2,red, opacity=.4] ([xshift=1mm] pic cs:sllmpxc1) to ([yshift=2mm] pic cs:sllmpxc2);
   \draw[rounded corners, smooth=2,green, opacity=.4] ([xshift=1mm] pic cs:sllcontrf1) to ([yshift=2mm] pic cs:sllcontrf2);
   \draw[rounded corners, smooth=2,green, opacity=.4] ([xshift=1mm] pic cs:sllcontrf1) to ([yshift=2mm] pic cs:sllcontrf3);
   \draw[rounded corners, smooth=2,red, opacity=.4] ([xshift=1mm] pic cs:sllcontrc1) to ([yshift=2mm] pic cs:sllcontrc2);
   \draw[rounded corners, smooth=2,red, opacity=.4] ([xshift=3mm] pic cs:sllocgc1) to ([xshift=-1mm, yshift=1mm] pic cs:sllocgc3);
   \draw[rounded corners, smooth=2,green, opacity=.4] ([xshift=2mm] pic cs:sllocgf1) to ([xshift=1mm,yshift=1mm] pic cs:sllocgf2);
   \draw[rounded corners, smooth=2,red, opacity=.4] ([xshift=2mm] pic cs:sllocfc1) to ([xshift=1mm,yshift=1mm] pic cs:sllocfc3);
   \draw[rounded corners, smooth=2,green, opacity=.4] ([xshift=-1mm] pic cs:sllocff1) to ([xshift=-1mm,yshift=1mm] pic cs:sllocff2);
   \draw[rounded corners, smooth=2,red, opacity=.4] ([xshift=1mm] pic cs:sllocuc1) to ([xshift=0mm,yshift=1mm] pic cs:sllocuc2);
   \draw[rounded corners, smooth=2,green, opacity=.4] ([xshift=-3mm] pic cs:sllocuf1) to ([xshift=-3mm,yshift=1mm] pic cs:sllocuf2);
	\end{tikzpicture}
    \caption{Exponential fragment of \musuperLLinf
    \label{fig:musuperllexprules}}
\end{figure*}

\begin{defi}[$\superLL(\Sig, \leqg, \leqf, \lequ)$]
$\superLL(\Sig, \leqg, \leqf, \lequ)$ proofs are the trees inductively generated by those inferences, satisfying the above side-conditions.
\end{defi}

There are instances of \superLL{} where cut-elimination fails: some conditions are required, so that cut inferences can indeed be eliminated.

The following two definitions aim at formulating these conditions in a suitable way:
\begin{defi}[Derivability closure] 
Given a signature $\e$, we define the derivability closure $\bar{\e}$ to be the signature inductively defined by:
$$  \AIC{\e(r)}
  \RL{}
  \UIC{\bar{\e}(r)}
  \DP
\qquad 
  \AIC{\bar{\e}(\contr{i})\quad \bar{\e}(\contr{j})}
  \RL{}
  \UIC{\bar{\e}(\contr{i+j-1})}
  \DP
\qquad
  \AIC{ \e(\contr{2}) \quad\bar{\e}(\mpx{i})\quad\bar{\e}(\mpx{j}) \quad i, j\neq 0}
  \RL{}
  \UIC{\bar{\e}(\mpx{i+j})}
  \DP
  \qquad
  \AIC{\e(\mpx{1})\quad\bar{\e}(\contr{i})}
  \RL{}
  \UIC{ \bar{\e}(\mpx{i})}
  \DP
$$

\end{defi}

Derivability closure comes with the following property, proved by induction on $\bar{\e}(r)$:
\begin{prop}
\label{prop:closureDerivabilityDerivationequivalence}
If $\bar{\e}(r)$ holds, then $(r)$ is derivable for 
 connective $\wn_\e$, using only inference rules $\mpx{i}$ and $\contr{i}$ on this connective.
\todo{Expliquer et donner un exemple.}
\end{prop}
\begin{nota}
We name $\contr{i}^{\bar{\e}}$ (resp. $\mpx{i}^{\bar{\e}}$), for $i\in\mathbb{N}$, any derivation using only $\contr{j}$ and $\mpx{j}$ rules and having the same conclusion and hypothesis as $\contr{i}$ (resp. $\mpx{i}$).
We write $\bar{\e}(\contr{0})$ for $\bar{\e}(\mpx{0})$ and set $\bar{\e}(\contr{1})$ to $\True$ for all $\e$ and $\contr{1}^{\bar{\e}}$ to be the empty derivation.
\end{nota}

To define a cut-reduction system, we consider cut-elimination axioms defined in \Cref{cutElimAxs}.
In \superLL{}-systems each axiom corresponds to one step of cut-elimination.
However, as our reduction system with fixed-points is based on the (\mcut)-rule,
some axioms will be used in several reduction cases.
\begin{table*}
\centering
$\begin{array}{|clcll|c|}
\hline
\e\leqg\e' & \Rightarrow & \e(\mpx{i})&\Rightarrow &\bar{\e'}(\contr{i}) \hfill i\geq 0 &  \axgmpx \label{gmpxAx}\\
\e\leq_s\e' & \Rightarrow & \e(\mpx{i})&\Rightarrow&\bar{\e'}(\mpx{i}) \hfill i\geq 0  \text{ and }s\neq g &  \axfumpx \label{fumpxAx}\\
\e\leq_s\e' & \Rightarrow &\e(\contr{i})&\Rightarrow&\bar{\e'}(\contr{i}) \hfill i\geq 2 &  \axcontr \label{contrAx}\\
\e\leq_s \e' & \Rightarrow& \e'\leq_s\e''& \Rightarrow &\e\leq_s\e''&  \axTrans \label{leqTrans}\\
\e\leqg\e' &\Rightarrow &\e'\leq_s\e'' &\Rightarrow& \e\leqg\e''& \axleqgs \label{leqgs}\\
\e\leqf\e' &\Rightarrow &\e'\lequ \e''&\Rightarrow& \e\leqf\e''& \axleqfu\label{leqfu}\\
\e\leqf\e' &\Rightarrow &\e'\leqg\e'' &\Rightarrow &\e\leqg\e'' 
\land (\e\leqf\e''' \Rightarrow (\e\leqg\e''' \land \e'''(\mpx{1}))) & \axleqfg\label{leqfg}\\
\e\lequ\e' & \Rightarrow& \e'\leq_s\e''&\Rightarrow &\e\leq_s\e''&  \axlequs\label{lequs}\\
\hline
\end{array}
$\\[2ex]
with $s\in\{g,f,u\}$, all the axioms are universally quantified.\\
For convenience, we use the notation $\contr{0}:=\mpx{0}$ and set $\bar{\e}(\contr{1})=\True$ for all $\e$.\\[2ex]
\caption{Cut-elimination axioms}\label{cutElimAxs}
\end{table*}
In Bauer and Laurent's system~\cite{TLLA21}, properties of {\it axiom expansion} and {\it cut-elimination} hold. 
We defer the former to  \Cref{app:axexp} and focus on  the latter:

\begin{thm}[Cut Elimination]\label{superllcutelim}\proofref{See a direct proof in appendix~\ref{app:superllcutelim}}
As soon as the 8 cut-elimination axioms of \Cref{cutElimAxs} are satisfied, cut elimination holds for $\superLL(\Sig,\leqg,\leqf,\lequ)$.
\end{thm}

This theorem will be proved, in \Cref{section:cut-elimination}, as a corollary of  \musuperLLinf{} cut-elimination theorem.
Many existing variants of \LL{} are instances of \superLL{}, {\it e.g.} let us consider \ELL~\cite{lll,djell}:
\begin{exa}
    \label{exa:ell}
\defname{Elementary Linear Logic} (\ELL) is a variant of \LL{} where ($\wnde$) and $(\ocprom)$ are replaced by  functorial promotion:
$
  \AIC{\vdash A,\A}
  \RL{\ocf}
  \UIC{\vdash \oc A,\wn\A}
  \DP
$.
This system is captured as the instance of $\superLL(\Sig, \leqg, \leqf, \lequ)$ system with $\Sig=\{\bullet\}$, defined by $\bullet(\contr{2})=\bullet(\mpx{0})=\True$ (and $(\bullet)(r)=\False$ otherwise), ${\leqg}={\lequ}=\emptyset$ and $\bullet\leqf\bullet$.
\defref{Details in appendix \ref{app:ELLexa}}

This $\superLL(\Sig, {\leqg}, {\leqf}, {\lequ})$ instance is \ELL{} and satisfies the axioms of cut-elimination.
\end{exa}

As argued in~\cite{TLLA21}, the \superLL{}-systems subsume many other existing variants of \LL{} such as \SLL~\cite{sll}, \LLL~\cite{lll}, \seLL~\cite{Nigam09}.
The last two are particularly interesting as they require more than one exponential signature to be formalized.
In the following section, we will look at some examples for the fixed-point version of \musuperLLinf{}.

    \section{Super exponentials with fixed-points}
\label{section:musuperll}

In this section, we define \musuperLLinf{} and give some interesting instances of it.

\subsection{Definition of \musuperLLinf{}}


Let $\Sig$ be an exponential name, the pre-formulas of $\musuperLLinf(\Sig)$ are $\superLL{}(\Sig)$ formulas extended with fixed-point variables and fixed-points constructs (with $a\in\Atom$, $ X\in\Var$, $\sigma\in \Sig$):

\noindent $F, G \hspace{0.1em}::= \hspace{0.1em} a \mid a^\perp \mid X \mid F \parr G \mid F \otimes G \mid \bot \mid 1 \mid F \oplus G \mid F \with G \mid 0 \mid \top \mid \wn_\e F \mid \oc_\e F\mid \mu X. F\mid \nu X. F.$

Formulas of $\musuperLLinf{}(\Sig)$ are the closed pre-formulas.
Negation is defined as the smallest involution on formulas satisfying the relations of \Cref{negationDef} as well as:\quad $ (\wn_\e F)^\perp := \oc_\e F^\perp. $

\medskip

Again, for one set of signatures $\Sig$ we define many systems, each parametrized with $\leqg, \leqf$ and $\lequ$. 
The inference rules for this system are the rules of $\superLL(\Sig, \leqg, \leqf, \lequ)$ together with the fixed-point fragment of \Cref{fig:fixFragment}.
As before, pre-proofs of $\musuperLLinf(\Sig, \leqg, \leqf, \lequ)$ are the trees coinductively generated by the rules of $\musuperLLinf(\Sig, \leqg, \leqf, \lequ)$ and validity is defined in the same way as for \muLLinf.

\subsection{Some instances of \musuperLLinf{}}

In this subsection, we give some interesting instances of \musuperLLinf{}.

\subsubsection{A linear modal $\mu$-calculus}

An application of super exponentials can be found in modelling the linear modal $\mu$-calculus introduced in~\cite{FOSSACS25} to prove a cut-elimination theorem for the modal $\mu$-calculus.
We show below how one can view a multi-modal $\mu$-calculus as \muLLmodinf{} as an instance of \musuperLLinf{}.

Let us consider a set of actions $\mathsf{ Act}$.
Formulas of \muLLmodinf{} are those of \muLLinf{} with the addition of a pair of modalities, $ \lozenge_\alpha F$ and $\Box_\alpha F$, for each $\alpha\in \mathsf{Act}$.
Rules of \muLLmodinf{} are the rules of \muLLinf{} where the promotion is extended with $\lozenge$-contexts. Rules on modalities are a functorial promotion (called the modal rule) and a contraction and a weakening on $\lozenge$-formulas:
$$
\AIC{\vdash F, \wn\A, \lozenge_{\alpha_1}G_1, \dots, \lozenge_{\alpha_n}G_n}
\RL{\ocpromloz}
\UIC{\vdash \oc F, \wn\A, \lozenge_{\alpha_1}G_1, \dots, \lozenge_{\alpha_n}G_n}
\DP\quad
\AIC{\vdash F, \A}
\RL{\boxprom}
\UIC{\vdash \Box_{\alpha} F, \lozenge_{\alpha}\A}
\DP\quad
\AIC{\vdash \lozenge_{\alpha} F, \lozenge_{\alpha} F, \A}
\RL{\diacontr}
\UIC{\vdash \lozenge_{\alpha} F, \A}
\DP\quad
\AIC{\vdash \A}
\RL{\diawk}
\UIC{\vdash \lozenge_{\alpha} F, \A}
\DP
$$
 (with $\alpha, \alpha_1, \dots \alpha_n\in \mathsf{Act}$) The system considered in \cite{FOSSACS25} corresponds to the case where $\mathsf{Act}$ is a singleton, that is a calculus with two exponential names, one of these names representing the $\mu$-calculus modality rather than a linear exponential.

\muLLmodinf{} can be modelled as the super exponential system $\musuperLLinf(\Sig, \leqg, \leqf, \lequ)$ with:
\begin{itemize}
\item $\Sig:=\{\bullet\}\cup \mathsf{Act}$.
\item $\contr{2}(\bullet) = \mpx{0}(\bullet)= \mpx{1}(\bullet) = \True$,
for any $\alpha\in \mathsf{Act}, \contr{2}(\alpha) = \mpx{0}(\alpha)=\True$, and all the other elements have value $\False$ for both signatures.
\item $\bullet\leqg\bullet$ ; $\bullet \leqg\alpha$; $\alpha\leqf\alpha$ for any $\alpha\in \mathsf{Act}$ and all other couples for the three relations $\leqg, \leqf$ and $\lequ$ are false.\defref{Details in \Cref{app:mullmodIsMusuperLL}}
\end{itemize}
This system is \muLLmodinf{} when taking: \quad$ \wn_{\bullet} := \wn,\quad \oc_{\bullet}:=\oc,\quad \wn_{\alpha}:= \lozenge_\alpha \quad\text{and}\quad \oc_{\alpha}:= \Box_\alpha$.
Moreover, the system satisfies cut-elimination axioms of \Cref{cutElimAxs}.

\subsubsection{\ELL{} with fixed points}

In~\cite{BAILLOT2015}, an affine version of second-order \ELL{} with recursive types, called \muEAL{}, is introduced. This system allows only finite proofs. \emph{Affine} means weakening applies to any formulas. Fixed points are added to a two-sided version with $\multimap$ and $(-)^\perp$ formulas, without any positivity condition on the fixed point variables, unlike what is enforced in our one-sided sequent version. The paper proves \muEAL{} cut-elimination and refines complexity bounds from \ELL{}.

Considering \muELLinf{}, an instance of \Cref{exa:ell} with fixed points, gives us a typing system which is close to \muEAL{}. Namely, consider $\musuperLLinf(\Sig, \leqg, \leqf, \lequ)$ with the same $\Sig, \leqg, \leqf$, and $\lequ$ as in \Cref{exa:ell}. Since the axioms in Table~\ref{cutElimAxs} only concern $\Sig, \leqg, \leqf$, and $\lequ$, they are also satisfied by this instance of \musuperLLinf{}.

Our systems differs in two ways from that of Baillot: (i) the extremal fixed-points instead of generic fixed-points and the condition of positivity on fixed-point variables, and (ii) the infinite nature of our proofs. Thus, our cut-elimination theorem may not apply due to (i), and even if it did, it might not ensure finite proofs because of (ii). However, Baillot~\cite{BAILLOT2015} uses only fixed-point variables in positive positions when proving complexity bounds, which addresses (i). Additionally, using only $\mu$-fixed-points to encode fixed points which ensures that cut-free proofs remain finite, resolving the incompatibility induced by (ii) by preventing infinite branches. (Moreover, the impact of weakening can be tamed by designing a translation making the system affine as well.)

\begin{rem}
    Note that there is no proof of the conclusion sequent of \Cref{exa:regProof} in \muELLinf{}.
\end{rem}

      \section{Cut-elimination}
    \label{section:cut-elimination}

In this section, we only consider instances of \musuperLLinf{} satisfying the axioms of \Cref{cutElimAxs}. Let us assume given such an instance, $\musuperLLinf(\Sig, \leqg, \leqf, \lequ)$, that we simply refer to as 
$\musuperLLinfSig$ in the following keeping the relations $ \leqg, \leqf$ and $\lequ$ implicit.

\subsection{(\mcut)-elimination steps}\label{mcutonestep}
Here, we define the (\mcut)-elimination steps of $\musuperLLinfSig$.
To do so, it is suitable to have a specific notation for the premisses containing only proofs concluded by a promotion. We use similar notations to those of \muLLi{} cut-elimination proof~\cite{TABLEAUX23}:
\begin{nota}[$(\oc)$-contexts]
$\ocProofs$ denotes a list of $\musuperLLinfSig$-proofs which are all concluded by some promotion rule ($\ocg, \ocf$ or $\ocu$).
Given $s\in\{g, f, u\}$, $\ocsProofs$ denotes a list of $\musuperLLinfSig$-proofs which are all concluded by an $(\oc_s)$-rule.
In both cases, $\C$ denotes the list of $\musuperLLinfSig$-proofs formed by gathering the immediate subproofs of the last promotion (being either $\C^\oc$, or $\C^{\oc_s}$).
\end{nota}

We now give a series of lemmas that will be used to justify the (\mcut)-reduction steps defined in 
\Cref{def:musuperLLcutred}.
We only give a proof sketch of \Cref{ocfgreadinesscondition}, and give complete proofs of each lemma in \Cref{app:mcutstepJustif}.
We start by the commutation cases of the different promotions.
The case~\hyperref[ocgcomm]{$({\text{comm}}_\ocg)$}  covers all the case where $(\ocg)$ commutes under the cut:
\begin{lem}[
Step~{\hyperref[ocgcomm]{$({\text{comm}}_\ocg)$}}]
\label{ocgreadinesscondition}
If\quad
\proofref{Details in \Cref{app:ocgreadinesscondition}}
$
\footnotesize{
\AIC{\pi}
\noLine
\UIC{\vdash A, \wn_{\vec{\f}}\B}
\RL{\ocg}
\UIC{\vdash \oc_{\e} A, \wn_{\vec{\f}}\B}
\AIC{\ocProofs}
\RL{\rmcutpar}
\BIC{\vdash \oc_{\e} A, \wn_{\vec{\g}}\A}
\DP
}
$
is a $\musuperLLinfSig$-proof then\quad
$
\footnotesize{
\AIC{\pi}
\noLine
\UIC{\vdash A, \wn_{\vec{\f}}\B}
\AIC{\ocProofs}
\RL{\rmcutpar}
\BIC{\vdash A, \wn_{\vec{\g}}\A}
\RL{\ocg}
\UIC{\vdash \oc_{\e} A, \wn_{\vec{\g}}\A}
\DP
}
$
is also a $\musuperLLinfSig$-proof.
\end{lem}


The case~\hyperref[ocfcommocgEmpty]{$({\text{comm}}^1_\ocf)$} covers the case of commutation of an $(\ocf)$-promotion but where only $(\ocg)$-rules with empty contexts appear in the hypotheses of the multi-cut.
Note that an $(\ocg)$ occurrence with empty context could be seen as an $(\ocf)$ occurrence (with empty context).
\begin{lem}[
Step~{\hyperref[ocfcommocgEmpty]{$({\text{comm}}^1_\ocf)$}}]
\label{ocfreadinesscondition}
If each sequent in $\ocProofs$ concluded by an $(\ocg)$ has an empty context and
\proofref{Details in \Cref{app:ocfreadinesscondition}}
$
\footnotesize{
\AIC{\pi}
\noLine
\UIC{\vdash A, \B}
\RL{\ocf}
\UIC{\vdash \oc_{\e} A, \wn_{\vec{\f}}\B}
\AIC{\ocProofs}
\RL{\rmcutpar}
\BIC{\vdash \oc_{\e} A, \wn_{\vec{\g}}\A}
\DP
}
$
is a $\musuperLLinfSig$-proof, then\quad
$
\footnotesize{
\AIC{\pi}
\noLine
\UIC{\vdash A, \B}
\AIC{\C}
\RL{\rmcutpar}
\BIC{\vdash A, \A}
\RL{\ocf}
\UIC{\vdash \oc_{\e} A, \wn_{\vec{\g}}\A}
\DP
}
$
is a $\musuperLLinfSig$-proof.
\end{lem}

   

We then have the following case where we commute an $(\ocf)$-rule, but where there is at least one $(\ocg)$-promotion with a non-empty context in the premisses of the multicut rule:
\begin{lem}[
Step~{\hyperref[ocfcommocgNonEmpty]{$({\text{comm}}^2_\ocf)$}}]
\label{ocfgreadinesscondition}
If  some $(\ocg)$-rule in $\ocgProofs$ has at least one formula in the context and
$
\footnotesize{
\AIC{\pi}
\noLine
\UIC{\vdash A, \B}
\RL{\ocf}
\UIC{\vdash \oc_{\e} A, \wn_{\vec{\f}}\B}
\AIC{\ocgProofs}
\RL{\rmcutpar}
\BIC{\vdash \oc_{\e} A, \wn_{\vec{\g}}\A}
\DP
}
$
is a $\musuperLLinfSig$-proof, then
$
\footnotesize{
\AIC{\pi}
\noLine
\UIC{\vdash A, \B}
\doubleLine
\RL{\mpx{1}}
\UIC{\vdash A, \wn_{\vec{\f}}\B}
\AIC{\ocgProofs}
\RL{\rmcutpar}
\BIC{\vdash A, \wn_{\vec{\g}}\A}
\RL{\ocg}
\UIC{\vdash \oc_{\e} A, \wn_{\vec{\g}}\A}
\DP
}
$
is also a $\musuperLLinfSig$-proof.
\end{lem}
\begin{proof}[Proof sketch]
\proofref{Details in \Cref{app:ocfgreadinesscondition}.}
First notice that, by hypothesis, $\e\leqf\vec{\f}$.   The proof is done in two steps:
   \begin{enumerate}
      \item From $\vdash \oc_\e A, \wn_{\vec{\f}}\B$ we follow mcut-connected sequents until reaching one $\vdash \oc_{\e'} A', \wn_{\vec{\f'}}\B'$ conclusion of an $(\ocg)$-rule with $\wn_{\vec{\f'}}\B'$ non empty, for each signature $\e'$ in these sequents, we prove that $\e\leqf\e'$ using axiom~\refleqTrans{} or \refleqfu. Then we use axiom~\refleqfg{} to prove that $\vec{\f}(\mpx{1})$ holds and $\e\leqg\vec{\f}$. Since $\vec{\f}(\mpx{1})$ holds, application of $(\mpx{1})$ is allowed.
      \item We run through all the sequents and using axiom~\refleqgs, we prove that $\e\leqg\e''$ for each signature $\e''$ we encounter.
   \end{enumerate}
   We therefore have $\e\leqg\vec{\g}$ as signatures from $\vec{\g}$ are contained on hypotheses of the \mcut: the application of $(\ocg)$ is therefore legal. 
\end{proof}

We then cover the cases where we commute an $(\ocu)$-rule with the multi-cut.
The first case is where there are only a list of $(\ocu)$-rules in the hypotheses of the multi-cut:
\begin{lem}[
Step~{\hyperref[ocucommOnlyocu]{$({\text{comm}}^1_\ocu)$}}]
\label{ocureadinesscondition}
If
\proofref{Details in \Cref{app:ocureadinesscondition}}
$
\footnotesize{
\AIC{\pi}
\noLine
\UIC{\vdash A, C}
\RL{\ocu}
\UIC{\vdash \oc_{\e} A, \wn_{\f} C}
\AIC{\ocuProofs}
\RL{\rmcutpar}
\BIC{\vdash \oc_{\e} A, \wn_{\g} B}
\DP
}
$
is a $\musuperLLinfSig$-proof, then
$
\footnotesize{
\AIC{\pi}
\noLine
\UIC{\vdash A, C}
\AIC{\C}
\RL{\rmcutpar}
\BIC{\vdash A, B}
\RL{\ocu}
\UIC{\vdash \oc_{\e} A, \wn_{\g} B}
\DP
}
$
is a $\musuperLLinfSig$-proof.
\end{lem}

The second case of $(\ocu)$-commutation is where we have an $(\ocf)$-rule and where the hypotheses concluded by an $(\ocg)$-rule have empty contexts.
\begin{lem}[
Step~{\hyperref[ocucommocgEmpty]{$({\text{comm}}^2_\ocu)$}}]
\label{ocufreadinesscondition}
If $\ocProofs$ contains at least one $(\ocf)$, if each $(\ocg)$ has empty context and if
\proofref{Details in \Cref{app:ocufreadinesscondition}}
$
\footnotesize{
\AIC{\pi}
\noLine
\UIC{\vdash A, B}
\RL{\ocu}
\UIC{\vdash \oc_{\e} A, \wn_{\f} B}
\AIC{\ocProofs}
\RL{\rmcutpar}
\BIC{\vdash \oc_{\e} A, \wn_{\vec{\g}}\A}
\DP
}
$
is a $\musuperLLinfSig$-proof, then
$
\footnotesize{
\AIC{\pi}
\noLine
\UIC{\vdash A, B}
\AIC{\C}
\RL\rmcutpar
\BIC{\vdash A, \A}
\RL{\ocf}
\UIC{\vdash \oc_{\e} A, \wn_{\vec{\g}}\A}
\DP
}
$
is also a $\musuperLLinfSig$-proof.
\end{lem}

The following lemma deals with the case where there are sequents concluded by an $(\ocg)$-rule with non-empty context and where the first rule encountered is an $\ocf$-rule.
\begin{lem}[
Step~{\hyperref[ocucommocgNonEmptyocfFirst]{$({\text{comm}}^3_\ocu)$}}]
\proofref{Details in \Cref{app:ocugreadinessconditionfcase}}
\label{ocugreadinessconditionfcase}
Let
 $\ocProofs_2$ contain a $(\ocg)$ with non-empty context, $\C:=\{\vdash \oc_{\e} A, \wn_{\f} B\} \cup \ocuProofs_1 \cup \{\vdash \oc_{\e'} C, \wn_{\vec{\f'}}\B\}$ is cut-connected and $\C':=\{\vdash \oc_{\e'} C, \wn_{\vec{\f'}}\B\} \cup \ocProofs_2$ as well. If
$
{\scriptsize
\AIC{\pi_1}
\noLine
\UIC{\vdash A, B}
\RL{\ocu}
\UIC{\vdash \oc_{\e} A, \wn_{\f} B}
\AIC{\hspace{-.5cm}\ocuProofs_1~~ \ocProofs_2\hspace{-.5cm}}
\AIC{\pi_2}
\noLine
\UIC{\vdash C, \B}
\RL{\ocf}
\UIC{\vdash \oc_{\e'} C, \wn_{\vec{\f'}}\B}
\RL{\rmcutpar}
\TIC{\vdash \oc_{\e} A, \wn_{\vec{\g}}\A}
\DP}
$
is
 a $\musuperLLinfSig$-proof
then
$
{\scriptsize\AIC{\pi_1}
\noLine
\UIC{\vdash A, B}
\AIC{\hspace{-.5cm}\mathcal{C}_1 ~~ \ocProofs_2\hspace{-.5cm}}
\AIC{\pi_2}
\noLine
\UIC{\vdash C, \B}
\doubleLine
\RL{\mpx{1}}
\UIC{\vdash C, \wn_{\vec{\f'}}\B}
\RL{\rmcutpar}
\TIC{\vdash A, \wn_{\vec{\g}}\A}
\RL{\ocg}
\UIC{\vdash \oc_{\e} A, \wn_{\vec{\g}}\A}
\DP}
$
is also a $\musuperLLinfSig$-proof.
\end{lem}

The last lemma of promotion commutation is about the case where we commute an $(\ocu)$-promotion but when first meeting an $(\ocg)$-promotion.
\begin{lem}[
Step~{\hyperref[ocucommocgNonEmptyocgFirst]{$({\text{comm}}^4_\ocu)$}}]
\label{ocugreadinessconditiongcase}
Let $\C:=\{\vdash \oc_{\e} A, \wn_{\f} B\} \cup \ocuProofs_1 \cup \{\vdash \oc_{\e'} C, \wn_{\vec{\f'}}\B\}$ be cut-connected and $\C':=\{\vdash \oc_{\e'} C, \wn_{\vec{\f'}}\B\} \cup \ocProofs_2$ as well. If
\proofref{Details in \Cref{app:ocugreadinessconditiongcase}}
$
{\scriptsize\AIC{\pi_1}
\noLine
\UIC{\vdash A, B}
\RL{\ocu}
\UIC{\vdash \oc_{\e} A, \wn_{\f} B}
\AIC{\hspace{-.4cm}\ocuProofs_1
~~ \ocProofs_2\hspace{-.4cm}}
\AIC{\pi_2}
\noLine
\UIC{\vdash C, \wn_{\vec{\f'}}\B}
\RL{\ocg}
\UIC{\vdash \oc_{\e'} C, \wn_{\vec{\f'}}\B}
\RL{\rmcutpar}
\TIC{\vdash \oc_{\e} A, \wn_{\vec{\g}}\A}
\DP}
$
is a $\musuperLLinfSig$-proof then
$
{\scriptsize\AIC{\pi_1}
\noLine
\UIC{\vdash A, B}
\AIC{\hspace{-.4cm}\C_1~~\ocProofs_2\hspace{-.4cm}}
\AIC{\pi_2}
\noLine
\UIC{\vdash C, \wn_{\vec{\f'}}\B}
\RL{\rmcutpar}
\TIC{\vdash A, \wn_{\vec{\g}}\A}
\RL{\ocg}
\UIC{\vdash \oc_{\e} A, \wn_{\vec{\g}}\A}
\DP}
$
is also a $\musuperLLinfSig$-proof.
\end{lem}

The principal cases start with the contraction:

\begin{lem}[
Step~{\hyperref[contrPrincip]{$(\text{principal}_{\contr{}})$}}]
\label{contrPrincipLemma}
If
\proofref{Details in \Cref{app:contrPrincipLemma}}
$
\footnotesize{
\AIC{\tikzmark{contrPrincipLemma12}\C_\B}
\AIC{\pi}
\noLine
\UIC{\vdash \overbrace{\wn_\e A, \dots, \wn_\e A}^i, \B}
\RL{\contr{i}}
\UIC{\vdash \wn_\e A, \tikzmark{contrPrincipLemma12p}\B}
\AIC{\tikzmark{contrPrincipLemma22}\ocProofs_{\wn_\e A}}
\RL{\rmcutpar}
\TIC{\vdash \tikzmark{contrPrincipLemma11}\A, \wn_{\vec{\g}}\A'\tikzmark{contrPrincipLemma21}}
\DP
}
$
\begin{tikzpicture}[overlay,remember picture,-,line cap=round,line width=0.1cm]
   \draw[rounded corners, smooth=2,green, opacity=.4] ([xshift=-2mm] pic cs:contrPrincipLemma21) to ([xshift=2mm, yshift=2mm] pic cs:contrPrincipLemma22);
   \draw[rounded corners, smooth=2,red, opacity=.4] (pic cs:contrPrincipLemma11) to ([xshift=2mm, yshift=2mm] pic cs:contrPrincipLemma12);
   \draw[rounded corners, smooth=2,red, opacity=.4] (pic cs:contrPrincipLemma11) to ([xshift=2mm, yshift=2mm] pic cs:contrPrincipLemma12p);
\end{tikzpicture}
is a $\musuperLLinfSig$-proof, then
$
{\footnotesize\AIC{\tikzmark{contrPrincipLemmapostred12}\C_\B}
\AIC{\pi}
\noLine
\UIC{\vdash \overbrace{\wn_\e A, \dots, \wn_\e A}^i, \B\tikzmark{contrPrincipLemmapostred12p}}
\AIC{\tikzmark{contrPrincipLemmapostred22}\overbrace{\ocProofs_{\wn_\e A} \dots \ocProofs_{\wn_\e A}}^i\tikzmark{contrPrincipLemmapostred32}}
\RL{\mathsf{mcut(\iota', \perp\!\!\!\perp')}}
\TIC{\tikzmark{contrPrincipLemmapostred11}\A, \tikzmark{contrPrincipLemmapostred21}\wn_{\vec{\g}}\A', \dots, \wn_{\vec{\g}}\A'\tikzmark{contrPrincipLemmapostred31}}
\doubleLine
\RL{\contr{i}^{\bar{\vec{\g}}}}
\UIC{\vdash \A, \wn_{\vec{\g}}\A'}
\DP}
$
\begin{tikzpicture}[overlay,remember picture,-,line cap=round,line width=0.1cm]
   \draw[rounded corners, smooth=2,green, opacity=.4] ([xshift=0mm] pic cs:contrPrincipLemmapostred21) to ([xshift=2mm, yshift=2mm] pic cs:contrPrincipLemmapostred22);
   \draw[rounded corners, smooth=2,green, opacity=.4] ([xshift=-2mm] pic cs:contrPrincipLemmapostred31) to ([xshift=-2mm, yshift=2mm] pic cs:contrPrincipLemmapostred32);
   \draw[rounded corners, smooth=2,red, opacity=.4] (pic cs:contrPrincipLemmapostred11) to ([xshift=2mm, yshift=2mm] pic cs:contrPrincipLemmapostred12);
   \draw[rounded corners, smooth=2,red, opacity=.4] (pic cs:contrPrincipLemmapostred11) to ([xshift=-1mm, yshift=1mm] pic cs:contrPrincipLemmapostred12p);
\end{tikzpicture}
is so.
\end{lem}
Before giving the principal case for the multiplexing, we need to define $\opmpx_{S^\oc}(\ocProofs)$ contexts. The intuition is that when a multiplexing rule reduces (i) with a Girard's promotion, they simply cancel each other while when it  interacts (ii) with a $(\ocf)$ or $(\ocu)$, not only those two rules cancel, but also the other promotions hereditarily $\cutrel$-connected to the first $(\ocf)$ or $(\ocu)$ rule, \emph{until some Girard's promotion is reached}, in which case this propagation stops:

\begin{defi}[$\opmpx_{S^\oc}(\ocProofs)$ contexts]
   \label{defi:mpxPrincipompx}
\defref{A graphical representation of this definition is given in \Cref{app:mpxPrincipLemma}, \Cref{app:mpxPrincipompx}.}
Let $\pi$ be some  $\musuperLLinfSig$-proof concluded in a $\mcut(\iota,\cutrel)$ inference,
$\ocProofs$ a context of the multicut which is a tree with respect to a cut-relation $\cutrel$ and
$S^\oc$ be a sequent of $\ocProofs$ that we shall consider as the root of the tree. 

We define a $\musuperLLinfSig$-context $\opmpx_{S^\oc}(\ocProofs)$ altogether with two sets of sequents, $\mathcal{S}^{\mpx{}}_{\C^\oc, S^\oc}$ and $\mathcal{S}^{\contr{}}_{\C^\oc, S^\oc}$, by induction on the tree ordering on $\ocProofs$:

Let $\ocProofs_1, \dots, \ocProofs_n$ be the sons of $S^\oc$, such that $\ocProofs = (S^\oc, (\ocProofs_1, \dots, \ocProofs_n))$, we have two cases:
\begin{itemize}
\item $S^\oc = S^\ocg$, then we define $\opmpx_{S^\oc}(\ocProofs):= (S, (\ocProofs_1, \dots, \ocProofs_n))$ ; $ \mathcal{S}^{\mpx{}}_{\ocProofs, S^\oc}:=\emptyset$ ; $ \mathcal{S}^{\contr{}}_{\ocProofs, S^\oc} := \ocProofs$.

\item $S^\oc = S^\ocf$ or $S^\oc = S^\ocu$, then let the root of $\ocProofs_i$ be $S_i^\oc$, we define $\opmpx_{S^\oc}(\ocProofs)$ as\\ $(S, \opmpx_{S_1^\oc}(\ocProofs_1), \dots, \opmpx_{S_n^\oc}(\ocProofs_n))$ ; $\mathcal{S}^{\mpx{}}_{\ocProofs, S^\oc} := \{S^\oc\}\cup \bigcup \mathcal{S}^{\mpx{}}_{\ocProofs_i, S_i^\oc}$ ; $\mathcal{S}^{\contr{}}_{\ocProofs, S^\oc} := \bigcup \mathcal{S}^{\contr{}}_{\ocProofs_i, S_i^\oc}$.
\end{itemize}
\end{defi}

We can now state the multiplexing principal case:
\begin{lem}[
Step~{\hyperref[mpxcomm]{$(\text{principal}_{\mpx{}})$}}]
   \label{mpxPrincipLemma}
If \proofref{Details in \Cref{app:mpxPrincipLemma}}
$
\footnotesize{
\AIC{\C_\B}
\AIC{\vdash 
\overbrace{A, \dots, A}^i, 
\B}
\RL{\mpx{i}}
\UIC{\vdash \wn_\e A, \B}
\AIC{\ocProofs_{\wn_\e A}}
\RL{\rmcutpar}
\TIC{\vdash \A, \wn_{\g'}\A', \wn_{\g''}\A''}
\DP
}
$
is a $\musuperLLinfSig$-proof with $\A$ sent on $\C_{\B}\cup{\B}$ by $\iota$ ; $\wn_{\vec{\g''}} \A''$ sent on sequents of $\mathcal{S}^{\mpx{}}_{\ocProofs, S^\oc}$ ; and $\wn_{\vec{\g'}}\A'$  sent on $\mathcal{S}^{\contr{}}_{\ocProofs, S^\oc}$, where $S^\oc :=\oc_\e A, \wn_{\vec{\f'}}\B'$ is the sequent cut-connected to $\vdash \wn_\e A, \B$ on the formula $\wn_\e A$, then
$
{\footnotesize\AIC{\C_\B}
\AIC{\vdash 
\overbrace{A, \dots, A}^i, 
\B}
\AIC{
\overbrace{\opmpx_{S^\oc}(\ocProofs_{\wn_\e A}) \dots \opmpx_{S^\oc}(\ocProofs_{\wn_\e A})}^i
}
\RL{\mathsf{mcut(\iota', \perp\!\!\!\perp')}}
\TIC{\vdash\A,
\A', \dots,\A',
\wn_{\vec{\g''}}\A'', \dots, \wn_{\vec{\g''}}\A''
}
\doubleLine
\RL{\mpx{i}^{\bar{\vec{\g'}}}}
\UIC{\vdash \A, \wn_{\vec{\g'}}\A', 
\wn_{\vec{\g''}}\A'', \dots, \wn_{\vec{\g''}}\A''
}
\doubleLine
\RL{\contr{i}^{\bar{\vec{\g''}}}}
\UIC{\vdash \A, \wn_{\vec{\g'}}\A', \wn_{\vec{\g''}}\A''}
\DP}
$
is also a $\musuperLLinfSig$-proof.

\end{lem}

\begin{defi}
\label{def:musuperLLcutred}
\Cref{fig:musuperllexpcommcutstep,fig:musuperllexpcommcutstep2,fig:musuperllexpprincipcutstep} (with the applicability conditions stated in the corresponding lemmas) induce the  (\mcut)-reduction relation over $\musuperLLinfSig$ proofs. 
\end{defi}

\begin{rem}
No justification lemma is stated for $(\text{comm}_{\mpx{}})$ nor $(\text{comm}_{\contr{}})$ as applicability of $(\mpx{})$ and $(\contr{})$ only depends on the connective and not on the context.

Even though some reduction rules presented in \Cref{fig:musuperllexpcommcutstep} may seem to overlap, note that the applicability conditions of the Lemmas ensure that it is not the case.
\end{rem}

\begin{figure}[t]
$$
\hspace{-1cm}\begin{array}{c|c|c}
\text{Reduction} & \text{Name} & \text{Lemma}\\
\footnotesize{
\AIC{\pi}
\noLine
\UIC{\vdash A, \wn_{\vec{\f}}\B}
\RL{\ocg}
\UIC{\vdash \oc_{\e} A, \wn_{\vec{\f}}\B}
\AIC{\ocProofs}
\RL{\rmcutpar}
\BIC{\vdash \oc_{\e} A, \wn_{\vec{\g}}\A}
\DP\quad\rightsquigarrow\quad
\AIC{\pi}
\noLine
\UIC{\vdash A, \wn_{\vec{\f}}\B}
\AIC{\ocProofs}
\RL{\rmcutpar}
\BIC{\vdash A, \wn_{\vec{\g}}\A}
\RL{\ocg}
\UIC{\vdash \oc_{\e} A, \wn_{\vec{\g}}\A}
\DP}
& ({\text{comm}}_\ocg)\label{ocgcomm}& \ref{ocgreadinesscondition}
\\[2ex]
\footnotesize{
\AIC{\pi}
\noLine
\UIC{\vdash A, \B}
\RL{\ocf}
\UIC{\vdash \oc_{\e} A, \wn_{\vec{\f}}\B}
\AIC{\ocProofs}
\RL{\rmcutpar}
\BIC{\vdash \oc_{\e} A, \wn_{\vec{\g}}\A}
\DP\quad\rightsquigarrow\quad\nolinebreak
\AIC{\pi}
\noLine
\UIC{\vdash A, \B}
\AIC{\C}
\RL{\rmcutpar}
\BIC{\vdash A, \A}
\RL{\ocf}
\UIC{\vdash \oc_{\e} A, \wn_{\vec{\g}}\A}
\DP
}
& ({\text{comm}}^1_\ocf)\label{ocfcommocgEmpty}
& \ref{ocfreadinesscondition}
\\[2ex]
{
\footnotesize
\AIC{\pi}
\noLine
\UIC{\vdash A, \B}
\RL{\ocf}
\UIC{\vdash \oc_{\e} A, \wn_{\vec{\f}}\B}
\AIC{\ocProofs}
\RL{\rmcutpar}
\BIC{\vdash \oc_{\e} A, \wn_{\vec{\g}}\A}
\DP
\quad\rightsquigarrow\quad
\AIC{\pi}
\noLine
\UIC{\vdash A, \B}
\doubleLine
\RL{\mpx{1}}
\UIC{\vdash A, \wn_{\vec{\f}}\B}
\AIC{\ocProofs}
\RL{\rmcutpar}
\BIC{\vdash A, \wn_{\vec{\g}}\A}
\RL{\ocg}
\UIC{\vdash \oc_{\e} A, \wn_{\vec{\g}}\A}
\DP} & ({\text{comm}}^2_\ocf) \label{ocfcommocgNonEmpty}
& \ref{ocfgreadinesscondition}
\\[2ex]
\footnotesize{
\AIC{\pi}
\noLine
\UIC{\vdash A, C}
\RL{\ocu}
\UIC{\vdash \oc_{\e} A, \wn_{\f} C}
\AIC{\ocuProofs}
\RL{\rmcutpar}
\BIC{\vdash \oc_{\e} A, \wn_{\g} B}
\DP
\quad\rightsquigarrow\quad
\AIC{\pi}
\noLine
\UIC{\vdash A, C}
\AIC{\C}
\RL{\rmcutpar}
\BIC{\vdash A, B}
\RL{\ocu}
\UIC{\vdash \oc_{\e} A, \wn_{\g} B}
\DP 
}
& ({\text{comm}}^1_\ocu)\label{ocucommOnlyocu}
& \ref{ocureadinesscondition}
\\[2ex]
\footnotesize{
\AIC{\pi}
\noLine
\UIC{\vdash A, B}
\RL{\ocu}
\UIC{\vdash \oc_{\e} A, \wn_{\f} B}
\AIC{\ocProofs}
\RL{\rmcutpar}
\BIC{\vdash \oc_{\e} A, \wn_{\vec{\g}}\A}
\DP
\quad\rightsquigarrow\quad
\AIC{\pi}
\noLine
\UIC{\vdash A, B}
\AIC{\C}
\RL{\rmcutpar}
\BIC{\vdash A, \A}
\RL{\ocf}
\UIC{\vdash \oc_{\e} A, \wn_{\vec{\g}}\A}
\DP
}
& ({\text{comm}}^2_\ocu)\label{ocucommocgEmpty}
& \ref{ocufreadinesscondition}
\\[2ex]
{\footnotesize \AIC{\pi_1}
\noLine
\UIC{\vdash A, B}
\RL{\ocu}
\UIC{\vdash \oc_{\e} A, \wn_{\f} B}
\AIC{\hspace{-.4cm}\ocuProofs_1 ~~ \ocProofs_2\hspace{-.4cm}}
\AIC{\pi_2}
\noLine
\UIC{\vdash C, \B}
\RL{\ocf}
\UIC{\vdash \oc_{\e'} C, \wn_{\vec{\f'}}\B}
\RL{\rmcutpar}
\TIC{\vdash \oc_{\e} A, \wn_{\vec{\g}}\A}
\DP
\rightsquigarrow\quad
\AIC{\pi_1}
\noLine
\UIC{\vdash A, B}
\AIC{\hspace{-.4cm}\C_1 ~~ \ocProofs_2\hspace{-.4cm}}
\AIC{\pi_2}
\noLine
\UIC{\vdash C, \B}
\doubleLine
\RL{\mpx{1}}
\UIC{\vdash C, \wn_{\vec{\f'}}\B}
\RL{\rmcutpar}
\TIC{\vdash A, \wn_{\vec{\g}}\A}
\RL{\ocg}
\UIC{\vdash \oc_{\e} A, \wn_{\vec{\g}}\A}
\DP
} & ({\text{comm}}^3_\ocu)\label{ocucommocgNonEmptyocfFirst}
& \ref{ocugreadinessconditionfcase}
\\[2ex]
{\footnotesize
\AIC{\pi_1}
\noLine
\UIC{\vdash A, B}
\RL{\ocu}
\UIC{\vdash \oc_{\e} A, \wn_{\f} B}
\AIC{\hspace{-.4cm}\ocuProofs_1 ~~\ocProofs_2\hspace{-.4cm}}
\AIC{\pi_2}
\noLine
\UIC{\vdash C, \wn_{\vec{\f'}}\B}
\RL{\ocg}
\UIC{\vdash \oc_{\e'} C, \wn_{\vec{\f'}}\B}
\RL{\rmcutpar}
\TIC{\vdash \oc_{\e} A, \wn_{\vec{\g}}\A}
\DP
\rightsquigarrow
\AIC{\pi_1}
\noLine
\UIC{\vdash A, B}
\AIC{\hspace{-.4cm}\C_1 ~~ \ocProofs_2\hspace{-.4cm}}
\AIC{\pi_2}
\noLine
\UIC{\vdash C, \wn_{\vec{\f'}}\B}
\RL{\rmcutpar}
\TIC{\vdash A, \wn_{\vec{\g}}\A}
\RL{\ocg\quad}
\UIC{\vdash \oc_{\e} A, \wn_{\vec{\g}}\A}
\DP} & ({\text{comm}}^4_\ocu) \label{ocucommocgNonEmptyocgFirst}
& \ref{ocugreadinessconditiongcase}
\end{array}$$
    \caption{Commutative cut-reduction steps of the  \musuperLLinf{} promotion rules
    \label{fig:musuperllexpcommcutstep}}
\end{figure}

\begin{figure}[t]
$$
\hspace{-1cm}\begin{array}{c|c|c}
{\footnotesize
\AIC{\pi}
\noLine
\UIC{\vdash \overbrace{A,\dots, A}^i, \B}
\RL{\mpx{i}}
\UIC{\vdash \wn_\e A, \B}
\AIC{\mathcal{C}}
\RL{\rmcutpar}
\BIC{\vdash \wn_\e A, \A}
\DP\rightsquigarrow
\AIC{\pi}
\noLine
\UIC{\vdash \overbrace{A,\dots, A}^i, \B}
\AIC{\mathcal{C}}
\RL{\mathsf{mcut(\iota', \perp\!\!\!\perp')}}
\BIC{\vdash A,\dots, A, \A}
\RL{\mpx{i}}
\UIC{\vdash \wn_\e A, \A}
\DP}&(\text{comm}_{\mpx{}})\label{mpxcomm}
& 
\\[2ex]
{\footnotesize
\AIC{\pi}
\noLine
\UIC{\vdash \overbrace{\wn_{\e} A, \dots \wn_{\e} A}^i, \B}
\RL{\contr{i}}
\UIC{\vdash \wn_\e A, \B}
\AIC{\mathcal{C}}
\RL{\rmcutpar}
\BIC{\vdash \wn_\e A, \A}
\DP\rightsquigarrow
\AIC{\pi}
\noLine
\UIC{\vdash \overbrace{\wn_{\e} A, \dots, \wn_{\e} A}^i, \B}
\AIC{\mathcal{C}}
\RL{\mathsf{mcut(\iota', \perp\!\!\!\perp')}}
\BIC{\vdash \wn_{\e} A,\dots \wn_{\e} A, \A}
\RL{\contr{i}}
\UIC{\vdash \wn_\e A, \A}
\DP}&(\text{comm}_{\contr{}})\label{contrcomm}
& 
\end{array}$$
    \caption{Commutative cut-reduction steps for \musuperLLinf{} contraction and multiplexing rules
    \label{fig:musuperllexpcommcutstep2}}
\end{figure}

\begin{figure*}[t]
$
\hspace{-1cm}
\tiny{
\begin{array}{c|c}
\AIC{\tikzmark{contrPrincipRed12}\C_\B}
\AIC{\pi}
\noLine
\UIC{\vdash \overbrace{\wn_\e A, \dots, \wn_\e A}^i, \B}
\RL{\contr{i}}
\UIC{\vdash \wn_\e A, \B\tikzmark{contrPrincipRed12p}}
\AIC{\tikzmark{contrPrincipRed22}\ocProofs_{\wn_\e A}}
\RL{\rmcutpar}
\TIC{\vdash \tikzmark{contrPrincipRed11}\A, \wn_{\vec{\g}}\A'\tikzmark{contrPrincipRed21}}
\DP\rightsquigarrow\quad
\AIC{\tikzmark{contrPrincipRedpostred12}\C_\B}
\AIC{\pi}
\noLine
\UIC{\vdash \overbrace{\wn_\e A, \dots, \wn_\e A}^i, \B\tikzmark{contrPrincipRedpostred12p}}
\AIC{\tikzmark{contrPrincipRedpostred22}\overbrace{\ocProofs_{\wn_\e A}\dots \ocProofs_{\wn_\e A}}^i\tikzmark{contrPrincipRedpostred32}}
\RL{\mathsf{mcut(\iota', \perp\!\!\!\perp')}}
\TIC{\tikzmark{contrPrincipRedpostred11}\A, \tikzmark{contrPrincipRedpostred21}\wn_{\vec{\g}}\A', \dots, \wn_{\vec{\g}}\A'\tikzmark{contrPrincipRedpostred31}}
\doubleLine
\RL{\contr{i}^{\bar{\vec{\g}}}}
\UIC{\vdash \A, \wn_{\vec{\g}}\A'}
\DP
& 
\begin{array}{c}
(\text{principal}_{\contr{}})\label{contrPrincip}\\
\Cref{contrPrincipLemma}
\end{array}
\begin{tikzpicture}[overlay,remember picture,-,line cap=round,line width=0.1cm]
   \draw[rounded corners, smooth=2,green, opacity=.4] ([xshift=-2mm] pic cs:contrPrincipRed21) to ([xshift=2mm, yshift=2mm] pic cs:contrPrincipRed22);
   \draw[rounded corners, smooth=2,red, opacity=.4] (pic cs:contrPrincipRed11) to ([xshift=2mm, yshift=2mm] pic cs:contrPrincipRed12);
   \draw[rounded corners, smooth=2,red, opacity=.4] (pic cs:contrPrincipRed11) to ([xshift=-1mm, yshift=1mm] pic cs:contrPrincipRed12p);
   \draw[rounded corners, smooth=2,green, opacity=.4] ([xshift=0mm] pic cs:contrPrincipRedpostred21) to ([xshift=2mm, yshift=2mm] pic cs:contrPrincipRedpostred22);
   \draw[rounded corners, smooth=2,green, opacity=.4] ([xshift=-2mm] pic cs:contrPrincipRedpostred31) to ([xshift=-2mm, yshift=2mm] pic cs:contrPrincipRedpostred32);
   \draw[rounded corners, smooth=2,red, opacity=.4] (pic cs:contrPrincipRedpostred11) to ([xshift=2mm, yshift=2mm] pic cs:contrPrincipRedpostred12);
   \draw[rounded corners, smooth=2,red, opacity=.4] (pic cs:contrPrincipRedpostred11) to ([xshift=-1mm, yshift=1mm] pic cs:contrPrincipRedpostred12p);
\end{tikzpicture}
\\[2ex]
\AIC{\C_\B}
\AIC{\vdash \overbrace{A, \dots, A}^i, \B}
\RL{\mpx{i}}
\UIC{\vdash \wn_\e A, \B}
\AIC{\ocProofs_{\wn_\e A}}
\RL{\rmcutpar}
\TIC{\vdash \A, \wn_{\g'}\A', \wn_{\g''}\A''}
\DP\rightsquigarrow
\AIC{\C_\B}
\AIC{\vdash \overbrace{A, \dots, A}^i, \B}
\AIC{\overbrace{\opmpx_S(\ocProofs_{\wn_\e A})  \dots  \opmpx_S(\ocProofs_{\wn_\e A})}^i}
\RL{\mathsf{mcut(\iota', \perp\!\!\!\perp')}}
\TIC{\vdash\A,\overbrace{\A', \dots, \A'}^i, \overbrace{\wn_{\vec{\g''}}\A'', \dots, \wn_{\vec{\g''}}\A''}^i}
\doubleLine
\RL{\mpx{i}^{\bar{\vec{\g'}}}}
\UIC{\vdash \A, \wn_{\vec{\g'}}\A', \overbrace{\wn_{\vec{\g''}}\A'', \dots, \wn_{\vec{\g''}}\A''}^i}
\doubleLine
\RL{\contr{i}^{\bar{\vec{\g''}}}}
\UIC{\vdash \A, \wn_{\vec{\g'}}\A', \wn_{\vec{\g''}}\A''}
\DP
&
\begin{array}{c}
(\text{principal}_{\mpx{}})\label{mpxPrincip}\\
\Cref{mpxPrincipLemma}
\end{array}
\\[2ex]
\end{array}
}$
with $S$ being the sequent cut-connected to $\wn_\e A, \B$ on the formula $\wn_\e A$.
\caption{Principal cut-reduction steps of the exponential fragment of \musuperLLinf{}}\label{fig:musuperllexpprincipcutstep}
\end{figure*}

    \subsection{Translating \musuperLLinf{} into \muLLinf}\label{musuperlltollTranslation}
We now give a translation of $\musuperLLinf(\Sig, \leqg, \leqf, \lequ)$ into \muLLinf{} using directly the results of~\cite{TABLEAUX23} to deduce $\musuperLLinf(\Sig, \leqg, \leqf, \lequ)$ cut-elimination in a more modular way:
\begin{defi}[$(-)^\circ$-translation]
We define $(-)^\circ$ by induction on formulas ($c$ is any non-exponential connective):
$
c(F_1,\dots, F_n)^\circ  := c(F_1^\circ, \dots, F_n^\circ); \quad
X^\circ  := X; \quad
\forall \e, (\wn_\e A)^{\circ}  := \wn A^{\circ}; \quad
a^\circ  := a;\quad
(\oc_\e A)^{\circ}  := \oc A^{\circ}.
$
We define translations for exponential rules of $\musuperLLinf(\Sig, \leqg, \leqf, \lequ)$ in \Cref{fig:superlltollRules}. Other rules have their translations equal to themselves.
Proof translation $\pi^\circ$ of $\pi$ is the proof coinductively defined on $\pi$ from rule translations.
\end{defi}

\begin{figure*}
    \centering
    \hspace*{-.5cm}$
    {\small\AIC{\vdash \overbrace{A, \dots, A}^i, \A}
    \AIC{\hspace{-.4cm}i\neq 0}
    \noLine
    \UIC{\hspace{-.4cm}\e(\mpx{i})}
    \RL{\mpx{i}}
    \BIC{\vdash \wn_\e A, \A}
    \DP
    ~\rightsquigarrow~
    	\AIC{\vdash \overbrace{A^\circ, \dots, A^\circ}^i, \A^\circ}
	\doubleLine
	\RL{\wnde\times i}
	\UIC{\vdash \wn A^\circ, \dots, \wn A^\circ, \A^\circ}
	\RL{\wncontr\times (i-1)}
	\UIC{\vdash \wn A^\circ, \A^\circ}
    \DP
\quad
    \AIC{\vdash \A}
    \AIC{\hspace{-.4cm}\e(\mpx{0})}
    \RL{\mpx{0}}
    \BIC{\vdash \wn_\e A, \A}
    \DP
    ~\rightsquigarrow~
        \AIC{\vdash \A^\circ}
    \RL{\wnwk}
    \UIC{\vdash \wn A^\circ, \A^\circ}
    \DP}
    $\\[0ex]
    \hspace*{-1.5cm}
    $
    {\scriptsize
    \AIC{\vdash \overbrace{\wn_\e A, \dots, \wn_\e A}^i, \A}
    \AIC{\hspace{-.5cm}\e(\contr{i})}
    \RL{\contr{i}}
    \BIC{\vdash \wn_\e A, \A}
    \DP
        ~\rightsquigarrow~
	\AIC{\vdash \overbrace{\wn A^\circ, \dots, \wn A^\circ}^i, \A^\circ}
	\doubleLine
	\RL{\wncontr\times i}
	\UIC{\vdash \wn A^\circ, \A^\circ}
    \DP\quad
    \AIC{\vdash A, \wn_{\e_1} A_1, \dots, \wn_{\e_n} A_n}
    \AIC{\hspace{-.5cm}i\in\llbracket 1, n\rrbracket}
	\noLine    
    \UIC{\hspace{-.5cm}\e\leqg\e_i}
	\RL{\ocg}    
    \BIC{\vdash \oc_\e A, \wn_{\e_1} A_1, \dots, \wn_{\e_n} A_n}
    \DP
        ~\rightsquigarrow~
    \AIC{\vdash  A^\circ, \wn A_1^\circ, \dots, \wn A_n^\circ}
    \RL{\ocprom}
    \UIC{\vdash \oc A^\circ, \wn A_1^\circ, \dots, \wn A_n^\circ}
    \DP}
    $\\[1ex]
    \hspace*{-1cm}$
    {\small\AIC{\vdash A, A_1, \dots, A_n}
    \AIC{\hspace{-.4cm}i\in\llbracket 1, n\rrbracket}
	\noLine    
    \UIC{\hspace{-.4cm}\e\leqf\e_i}
	\RL{\ocf}
    \BIC{\vdash \oc_\e A, \wn_{\e_1} A_1, \dots, \wn_{\e_n} A_n}
    \DP
        ~\rightsquigarrow~
            \AIC{\vdash A^\circ, A_1^\circ, \dots, A_n^\circ}
    \doubleLine
    \RL{\wnde}
    \UIC{\vdash A^\circ, \wn A_1^\circ, \dots, \wn A_n^\circ}
    \RL{\ocprom}
    \UIC{\vdash \oc A^\circ, \wn A_1^\circ, \dots, \wn A_n^\circ}
    \DP
    \qquad
    \AIC{\vdash A, B}
    \AIC{\hspace{-.4cm}\e_1\lequ \e_2}
    \RL{\ocu}
    \BIC{\vdash \oc_{\e_1} A, \wn_{\e_2} B}
    \DP
    ~\rightsquigarrow~
    \AIC{\vdash A^\circ, B^\circ}
    \RL{\wnde}
    \UIC{\vdash A^\circ, \wn B^\circ, }
    \RL{\ocprom}
    \UIC{\vdash \oc A^\circ, \wn B^\circ}
    \DP}
    $
    \caption{Exponential rule translations from $\musuperLLinf(\Sig, \leqg, \leqf, \lequ)$ into \muLLinf}\label{fig:superlltollRules}
\end{figure*}

Since fixed-points are not affected by the translation, we have the following lemma:
\begin{lem}[$(-)^\circ$ preserves validity]\label{circValidityRobustness}
$\pi$ is a valid proof if and only if $\pi^\circ$ is a valid proof.
\end{lem}

The goal of this section is to prove that each
fair reductions sequence converges to a cut-free proof.
We have to make sure $(\mcut)$-reduction sequences are robust under this translation. In our proof of the final theorem, we also need one-step reduction-rules to be simulated by a finite number of reduction steps in the translation, which is the objective of the following lemma. We only give a proof sketch here, full proof can be found in appendix~\ref{app:detailsredSeqTranslationFiniteness}.
\begin{lem}
\label{redSeqTranslationFiniteness}
Let $\pi_0$ be a $\musuperLLinf(\Sig, \leqg, \leqf, \lequ)$ proof and let $\pi_0\rightsquigarrow \pi_1$ be a $\musuperLLinf(\Sig, \leqg, \leqf, \lequ)$ step of reduction. There exist a finite number of \muLLinf{} proofs $\theta_0, \dots, \theta_n$ such that $\theta_0\redseq\dots\redseq\theta_n,\quad
\pi_0^\circ = \theta_0$ and $\theta_n = \pi_1^{\circ}$ up to a finite number of rule permutations, done only on rules that just permuted down the $(\mcut)$.
\end{lem}
\begin{proof}[Proof sketch] \proofref{Details in \Cref{app:detailsredSeqTranslationFiniteness}.}
    Non exponential cases and commutations of multiplexing or contraction are immediate.
    Promotion commutations translate to commutation 
    rules and promotion key-cases.
    We must ensure that there exists a sequence of reductions commuting the translation of each promotion.
    Key-cases are trickier as they do not send the rules in the correct order: we need rule permutations to recover the translation of the target proof of the step.
\end{proof}
    
    Now that we know that a step of (\mcut)-reduction in $\musuperLLinf(\Sig, \leqg, \leqf, \lequ)$ translates to some steps of (\mcut)-reduction \muLLinf, the following lemma allows us to control the fairness:
    \begin{lem}[Completeness of the (\mcut)-reduction system]
        \label{mcutonestepCompleteness}
    If there is a \muLLinf{}-redex $\mathcal{R}$ sending $\pi^\circ$ to ${\pi'}^\circ$ then there exists a $\musuperLLinf(\Sig, \leqg, \leqf, \lequ)$-redex $\mathcal{R}'$ sending $\pi$ to a proof $\pi''$, such that in the translation of $\mathcal{R}'$, $\mathcal{R}$ is applied.
        \proofref{See proof in appendix, lemma~\ref{app:mcutonestepCompleteness}}
    \end{lem}
   
We define rule permutation with precision in appendix~\ref{app:rulePerm}. Here we show that validity is preserved if each rule is permuted a finite number of times:
\begin{prop}
\proofref{See proof in Appendix, \Cref{app:validFinitePerm}}
    \label{prop:validFinitePerm}
    If $\pi$ is a \muLLinf{} pre-proof sent to a pre-proof $\pi'$, via a permutation for which the permutation of one particular rule is finite, then $\pi$ is valid if and only if $\pi'$ is.
    \end{prop}

\begin{coro}
\label{FairredSeqTranslationFiniteness}
For every fair $\musuperLLinf(\Sig, \leqg, \leqf, \lequ)$ reduction sequences $(\pi_i)_{i\in\omega}$, there exists:
    \begin{itemize}
        \item a fair \muLLinf{} reduction sequence $(\theta_i)_{i\in\omega}$;
        \item a sequence of strictly increasing $(\varphi(i))_{i\in\omega}$ natural numbers;
        \item for each $i$, an integer $k_i$ and a finite sequence of rule permutations $(p_i^k)_{k\in\llbracket 0, k_i-1\rrbracket}$ starting from $\pi_i^\circ$ and ending $\theta_{\varphi(i)}$. For convenience in the proof, let's denote by $(\pi_i^k)_{k\in\llbracket 0, k_i\rrbracket}$ be the sequence of proofs associated to the permutation;
        \item for all $i>i'$, $p_i^k>p_i^{k'}$ if $k'\in\llbracket 0, k_{i'}-1$ and $k\geq k_{i'}$;
        \item for all $i,k$, $p_i^k$ are positions lower than the multicuts in $\pi_i^\circ$.
        \item for each $i'\geq i$ and for each $k\in\llbracket 0, k_i-1\rrbracket, p^k_{i'}=p^k_{i}$
    \end{itemize}
\proofref{See details on proof in \Cref{app:FairredSeqTranslationFiniteness}}
\end{coro}
\begin{proof}[Proof sketch]
We construct the sequence by induction on the steps of reductions of $(\pi_i)_{i\in\omega}$, starting with $\theta_0=\pi_0^\circ$, $\varphi(0)=0$ and $k_0=0$ and then applying \Cref{redSeqTranslationFiniteness} for each of the following steps.
We get fairness of $(\theta_i)_{i\in\omega}$ from \Cref{mcutonestepCompleteness}.
\end{proof}

Finally, we have our main result, proving cut-elimination of $\musuperLLinf(\Sig, \leqg, \leqf, \lequ)$:
\begin{thm}
\label{musuperLLinfCutElim}
\label{musuperLLmodinfCutElim}
If the axioms of \Cref{cutElimAxs} are satisfied, then every fair (\mcut)-reduction sequence of $\musuperLLinf(\Sig, \leqg, \leqf, \lequ)$ converges to a $\musuperLLinf(\Sig, \leqg, \leqf, \lequ)$ cut-free proof.
\end{thm}
\begin{proof}[Proof sketch, see full proof in appendix, \Cref{app:musuperLLmodinfCutElim}]
Consider $(\pi_i)_{i\in 1+\lambda}$, $\lambda\in\omega+1$,
a fair $\musuperLLinf(\Sig, \leqg, \leqf, \lequ)$ cut-reduction sequence. If the sequence is finite, we use  \Cref{redSeqTranslationFiniteness} and we are done.
If the sequence is infinite, using \Cref{FairredSeqTranslationFiniteness} we get a fair infinite \muLLinf{} reduction sequence $(\theta_i)_{i\in\omega}$.
By \Cref{thm:mullcutelim}, we know that $(\theta_i)_{i\in\omega}$ converges to a cut-free proof $\theta$ of \muLLinf{}.
We prove that $(\pi_i)_{i\in\omega}$ converges to a $\musuperLLinf(\Sig, \leqg, \leqf, \lequ)$ pre-proof using the fact that $(\theta_i)_i$ is the translation of $(\pi_i)_i$ and that it is productive.

Validity of the limit $\pi$ of $(\pi_i)_i$ follows from the translation of $\pi$ being equal to $\theta$ up to rule-permutation (each particular rule permutes finitely). From \Cref{circValidityRobustness} and \Cref{prop:validFinitePerm}, these two operations preserve validity, therefore $\pi$ is valid which concludes the proof.
\end{proof}

An important remark is that the above proof does not rely on \Cref{superllcutelim} in any way. As a consequence, cut-elimination for \superLL{} is in fact a direct corollary of \Cref{musuperLLinfCutElim}:
\proofref{See proof in appendix, \Cref{app:superllcutelim2}} 
\begin{coro}[Cut Elimination for \superLL, that is, \Cref{superllcutelim}]
    \label{coro:superllcutelim2}
Cut elimination holds for $\superLL(\Sig,\leqg,\leqf,\lequ)$ as soon as the 8 cut-elimination axioms of definition~\ref{cutElimAxs} are satisfied.
\end{coro}

\begin{rem}
This result not only gives another way of proving cut-elimination for \superLL{}-systems but the reductions we build in it are generally different from the ones that are built in~\cite{TLLA21}. Indeed, we are eliminating cuts from the bottom of the proof using the multicut rule whereas in~\cite{TLLA21} the deepest cuts in the proof are eliminated first.
\end{rem}

Since  $\muLLmodinf$ and $\muELLinf$ are instances of \musuperLLinf{} satisfying the cut-elimination axioms, we have the following results as immediate corollaries of \Cref{musuperLLinfCutElim}: 
 
\begin{coro}[Cut Elimination for \muLLmodinf]
    \label{coro:muLLmodinf}
Cut elimination holds for \muLLmodinf{}.
\end{coro}

\begin{coro}[Cut Elimination for \muELLinf]
    \label{coro:muELLinf}
Cut elimination holds for $\muELLinf$.
\end{coro}

    \section{Conclusion}

We introduced a family of logical systems, \musuperLLinf{}, and proved a syntactic cut-elimination theorem for them. Our systems features various exponential modalities with least and greatest fixed-points in the setting of circular and non-wellfounded  proofs. Our aim in doing so is to develop a methodology to make cut-elimination proofs more uniform and reusable. A key feature of our development is to combine proof-theoretical methods for establishing cut-elimination properties using translation and simulation results with axiomatization of sufficient conditions for cut-elimination.

While our initial motivation was to make more systematic a key step in our recent proof of cut-elimination for the modal $\mu$-calculus~\cite{FOSSACS25}, this allowed us to generalize our previous result (capturing directly the multi-modal $\mu$-calculus without needing for a proof) but also to capture various extensions of light logics with induction and coinuctions, notably a calculus close to Baillot \muEAL.
Our system therefore  encompasses various fixed-point extensions of existing linear logic systems, including well-known light logics extended with least and greatest fixed-points and a non-well-founded proof system. We provide a relatively simple and uniform proof of cut-elimination for these extensions. 
Quite interestingly, the addition of fixed-points provide a new cut-elimination proof for the fixed-point free setting.
\medskip

The \musuperLLinf{} system, as defined in this paper, does not include the digging rule. 
We plan to work on this question in future work, at least for restrictions of the digging. 
Indeed digging is a very challenging rule wrt to its possible modelling using fixed-points as it would contradict the finiteness of the Fisher-Ladner closure, a basic property of fixed-point systems. On the other hand, incorporating digging would enable us to cover all of the super exponential version from~\cite{TLLA21} while our current system is incomparable with that of~\cite{TLLA21}. It could also be relevant for modal calculus, as the digging rule for modal formulas is equivalent to Axiom 4 of modal logic. Other modal logic axioms, such as Axiom T and co-dereliction rules from differential linear logic, can be viewed as rules in linear logic.

Another natural future work would be to explore linear translations of affine linear logic and/or intuitionistic/classical translations of these systems, facilitating the study of proof theory closer to~\cite{BAILLOT2015}.

Finally, while we started with non-wellfounded proofs, studying how these results can be adapted to finitary version of \musuperLLinf{} is another interesting open question.

 \bibliography{biblio}

\clearpage

\appendix

\section{Details on \Cref{section:background}}

\subsection{Details on the multicut rule (\Cref{multicutdef})}\label{app:multicutdef}

We recall the conditions on the multi-cut rule~\cite{CSL16,aminaphd,TABLEAUX23}.
The multi-cut rule is a rule with an arbitrary number of hypotheses:
$$
\AIC{\vdash \B_1}
\AIC{\dots}
\AIC{\vdash \B_n}
\RL{\rmcutpar}
\TIC{\vdash\B}
\DP$$
Let $C:=
\{(i, j)\mid i\in\llbracket 1, n\rrbracket, j\in\llbracket 1, \#\B_i\rrbracket\}$, 
$\iota$ is a map from $\llbracket 1, \# \B\rrbracket$ to $C$ and $\cutrel$ is binary a relation on $C$:
\begin{itemize}
\item The map $\iota$ is injective;
\item The relation $\cutrel$ is defined for $C\setminus \iota$, and is total for this set;
\item The relation $\cutrel$ is symmetric;
\item Each index can be related at most once to another one;
\item If $(i,j) \cutrel (i',j')$, then the $\B_i[j] = (\B_{i'}[j'])^\perp$;
\item The relation on premisses sequents defined as: $\{(i,i') \mid \exists j,j', (i,j) \cutrel (i',j')\}$ is acyclic and connected.
\end{itemize}


\subsection{Details on the restriction of a multicut context (\Cref{multicutRestriction})}\label{app:multicutRestriction}
\begin{defi}[Restriction of a multicut context]
Let $
\AIC{\C}
\RL{\rmcutpar}
\UIC{s}
\DP
$ be a multicut occurrence such that $\C = s_1\quad\dots\quad s_n$ and let $s_i :=\vdash G_1, \dots, G_{r_i} $, we define $\C_{G_j}$ with $G_j\in s_i$ to be the least sub-context of $\C$ such that:
\begin{itemize}
\item The sequent $s_i$ is in $\C_{G_j}$;
\item If there exists $l$ such that $(i,j)\cutrel(k,l)$ then $s_k\in\C_{G_j}$;
\item For any $k\neq i$, if there exists $l$ such that $(k,l)\cutrel(k',l')$ and that $s_k\in\C_{G_j}$ then $s_{k'}\in\C_{G_j}$.
\end{itemize}
We then extend the notation to contexts, setting $\C_\emptyset := \emptyset$ and $\C_{F, \A} := \C_F\cup\C_\A$.
\end{defi}

\subsection{One-step multicut-elimination for \muMALLinf{}}\label{app:mumallonestep}
Commutative one-step reductions for \muMALLinf{} are given in \Cref{fig:muMALLonestepComm} whereas principal reductions in \Cref{fig:muMALLonestepPrincip}.
\begin{figure*}
\centering
$
\AIC{}
\RL{\ax}
\UIC{\vdash F, F^\perp}
\RL{\rmcutpar}
\UIC{\vdash F, F^\perp}
\DP
\rightsquigarrow
\AIC{}
\RL{\ax}
\UIC{\vdash F, F^\perp}
\DP
$\\[2ex]
$
\AIC{\C_{\A'}}
\AIC{\C_{\B'}}
\AIC{\vdash F, \A'}
\AIC{\vdash G, \B'}
\RL{\otimes}
\BIC{\vdash F\otimes G, \A', \B'}
\RL{\rmcutpar}
\TIC{\vdash \A, \B}
\DP\rightsquigarrow
\AIC{\C_{\A'}}
\AIC{\vdash F, \A'}
\RL{\rmcutparprime}
\BIC{\vdash F, \A}
\AIC{\C_{\B'}}
\AIC{\vdash G, \B'}
\RL{\rmcutpardouble}
\BIC{\vdash G, \B}
\RL{\otimes}
\BIC{\vdash F\otimes G, \A, \B}
\DP
$\\[2ex]
$
\AIC{\C}
\AIC{\vdash F, G, \A'}
\RL{\parr}
\UIC{\vdash F\parr G, \A'}
\RL\rmcutpar
\BIC{\vdash F\parr G, \A}
\DP
\rightsquigarrow
\AIC\C
\AIC{\vdash F, G, \A'}
\RL\rmcutparprime
\BIC{\vdash F, G, \A}
\RL\parr
\UIC{\vdash F\parr G, \A}
\DP
$\\[2ex]
$
\AIC{\C}
\AIC{\vdash F_i, \A'}
\RL{\oplus^i}
\UIC{\vdash F_1\oplus F_2, \A'}
\RL\rmcutpar
\BIC{\vdash F_1\oplus F_2, \A}
\DP
\rightsquigarrow
\AIC\C
\AIC{\vdash F_i, \A'}
\RL\rmcutpar
\BIC{\vdash F_i, \A}
\RL{\oplus^i}
\UIC{\vdash F_1\oplus F_2, \A}
\DP
$\\[2ex]
$
\AIC{\C}
\AIC{\vdash F, \A'}
\AIC{\vdash G, \A'}
\RL{\with}
\BIC{\vdash F\with G, \A'}
\RL\rmcutpar
\BIC{\vdash F\with G, \A}
\DP
\rightsquigarrow
\AIC\C
\AIC{\vdash F, \A'}
\RL\rmcutpar
\BIC{\vdash F, \A}
\AIC\C
\AIC{\vdash G, \A'}
\RL\rmcutpar
\BIC{\vdash G, \A}
\RL\with
\BIC{\vdash F\with G, \A}
\DP
$\\[2ex]
$
\AIC{\C}
\AIC{\vdash F[\delta X. F/X], \A'}
\RL{\delta}
\UIC{\vdash \delta X. F, \A'}
\RL\rmcutpar
\BIC{\vdash \delta X. F, \A}
\DP
\rightsquigarrow
\AIC{\C}
\AIC{\vdash F[\delta X. F/X], \A'}
\RL\rmcutpar
\BIC{\vdash F[\delta X. F/X], \A'}
\RL{\delta}
\UIC{\vdash \delta X. F, \A}
\DP\quad\text{with $\delta\in\{\mu, \nu\}$}
$\\[2ex]
$
{\small\AIC\C
\AIC{}
\RL{\top}
\UIC{\vdash \top, \A'}
\RL{\rmcutpar}
\BIC{\vdash \top, \A}
\DP
\rightsquigarrow
\AIC{}
\RL{\top}
\UIC{\vdash \top, \A}
\DP\qquad
\AIC{}
\RL{1}
\UIC{\vdash 1}
\RL{\rmcutpar}
\UIC{\vdash 1}
\DP
\rightsquigarrow
\AIC{}
\RL{1}
\UIC{\vdash 1}
\DP
}
$\\[2ex]
$
{\small
\AIC{\C}
\AIC{\vdash \A'}
\RL\bot
\UIC{\vdash \bot, \A'}
\RL\rmcutpar
\BIC{\vdash \bot, \A}
\DP
\rightsquigarrow
\AIC{\C}
\AIC{\vdash \A'}
\RL\rmcutparprime
\BIC{\vdash \A}
\RL{\bot}
\UIC{\vdash \bot, \A}
\DP}
$
\caption{Commutative one-step reduction rules for \muMALLinf{}}\label{fig:muMALLonestepComm}
\end{figure*}

\begin{figure*}
\centering
$
\AIC{\C}
\AIC{}
\RL\ax
\UIC{\vdash F, F^\perp}
\RL{\rmcutpar}
\BIC{\A}
\DP
\rightsquigarrow
\AIC\C
\RL{\rmcutparprime}
\UIC{\A}
\DP
$\\[2ex]
$
\AIC{\C}
\AIC{\vdash F, \A'}
\AIC{\vdash F^\perp, \B}
\RL{\cut}
\BIC{\vdash \A', \B}
\RL\rmcutpar
\BIC{\vdash \A}
\DP
\rightsquigarrow
\AIC{\C}
\AIC{\vdash F, \A'}
\AIC{\vdash F^\perp, \B}
\RL\rmcutparprime
\TIC{\vdash \A}
\DP
$\\[2ex]
$
\AIC{\C}
\AIC{\vdash F, G, \B}
\RL{\parr}
\UIC{\vdash F\parr G, \B}
\AIC{\vdash F^\perp, \A_1}
\AIC{\vdash G^\perp, \A_2}
\RL{\otimes}
\BIC{\vdash F^\perp \otimes G^\perp, \A_1, \A_2}
\RL{\rmcutpar}
\TIC{\vdash \A}
\DP
\rightsquigarrow
\AIC{\C}
\AIC{\vdash F, G, \B}
\AIC{\vdash F^\perp, \A_1}
\AIC{\vdash G^\perp, \A_2}
\RL{\rmcutparprime}
\QIC{\vdash \A}
\DP
$\\[2ex]
$
\AIC{\C}
\AIC{\vdash F_i, \B}
\RL{\oplus_i}
\UIC{\vdash F_1\oplus F_2, \B}
\AIC{\vdash F_1^\perp, \A'}
\AIC{\vdash F_2^\perp, \A'}
\RL{\with}
\BIC{\vdash F_1 \with F_2^\perp, \A'}
\RL{\rmcutpar}
\TIC{\vdash\A}
\DP
\rightsquigarrow
\AIC{\C}
\AIC{\vdash F_i, \B}
\AIC{\vdash F_i^\perp, \A'}
\RL\rmcutparprime
\TIC{\vdash \A}
\DP
$\\[2ex]
$
\AIC{\C}
\AIC{\vdash F[X:=\mu X. F], \B}
\RL{\mu}
\UIC{\vdash \mu X. F, \B}
\AIC{\vdash F[X:=\nu X. F], \B'}
\RL{\nu}
\UIC{\vdash \nu X. F, \B'}
\RL{\rmcutpar}
\TIC{\vdash \A}
\DP\rightsquigarrow
\AIC{\C}
\AIC{\vdash F[X:=\mu X. F], \B}
\AIC{\vdash F[X:=\nu X. F], \B'}
\RL{\rmcutpar}
\TIC{\vdash \A}
\DP
$\\[2ex]
$
\AIC{\C}
\AIC{}
\RL{1}
\UIC{\vdash 1}
\AIC{\vdash \A'}
\RL{\bot}
\UIC{\vdash \bot, \A'}
\RL\rmcutpar
\TIC{\vdash \A}
\DP
\rightsquigarrow
\AIC{\C}
\AIC{\vdash \A'}
\RL\rmcutparprime
\BIC{\vdash \A}
\DP
$
\caption{Principal one-step reduction rules for \muMALLinf{}}\label{fig:muMALLonestepPrincip}
\end{figure*}

\subsection{One-step multicut-elimination for \muLLinf{}}\label{app:mullonestep}
Commutative one-step reductions for \muLLinf{} are steps from \muMALLinf{} together with the reduction of the exponential fragment given in \Cref{fig:muLLonestep}.
\begin{figure*}
    $
    \AIC{\pi}
    \noLine
    \UIC{\vdash A, \wn\B}
    \RL{\ocprom}
    \UIC{\vdash \oc A, \wn\B}
    \AIC{\ocProofs}
    \RL{\rmcutpar}
    \BIC{\vdash \oc A, \wn\A}
    \DP\quad\rightsquigarrow\quad
    \AIC{\pi}
    \noLine
    \UIC{\vdash A, \wn\B}
    \AIC{\ocProofs}
    \RL{\rmcutpar}
    \BIC{\vdash A, \wn\A}
    \RL{\ocprom}
    \UIC{\vdash \oc A, \wn\A}
    \DP
    $\\[2ex]
    $
{\footnotesize\AIC{\pi}
\noLine
\UIC{\vdash \B}
\RL{\wnwk}
\UIC{\vdash \wn A, \B}
\AIC{\mathcal{C}}
\RL{\rmcutpar}
\BIC{\vdash \wn A, \A}
\DP\rightsquigarrow
\AIC{\pi}
\noLine
\UIC{\vdash \B}
\AIC{\mathcal{C}}
\RL{\mathsf{mcut(\iota', \perp\!\!\!\perp')}}
\BIC{\vdash \A}
\RL{\wnwk}
\UIC{\vdash \wn A, \A}
\DP}
$
\\[2ex]
$
{\footnotesize\AIC{\pi}
\noLine
\UIC{\vdash A, \B}
\RL{\wnde}
\UIC{\vdash \wn A, \B}
\AIC{\mathcal{C}}
\RL{\rmcutpar}
\BIC{\vdash \wn A, \A}
\DP\rightsquigarrow
\AIC{\pi}
\noLine
\UIC{\vdash A, \B}
\AIC{\mathcal{C}}
\RL{\mathsf{mcut(\iota', \perp\!\!\!\perp')}}
\BIC{\vdash A, \A}
\RL{\wnde}
\UIC{\vdash \wn A, \A}
\DP}
$
\\[2ex]
$
{\footnotesize\AIC{\pi}
\noLine
\UIC{\vdash \wn A, \wn A, \B}
\RL{\wncontr}
\UIC{\vdash \wn A, \B}
\AIC{\mathcal{C}}
\RL{\rmcutpar}
\BIC{\vdash \wn A, \A}
\DP\rightsquigarrow
\AIC{\pi}
\noLine
\UIC{\vdash \wn A, \wn A, \B}
\AIC{\mathcal{C}}
\RL{\mathsf{mcut(\iota', \perp\!\!\!\perp')}}
\BIC{\vdash \wn A, \wn A, \A}
\RL{\wncontr}
\UIC{\vdash \wn A, \A}
\DP}
$
$
\AIC{\tikzmark{prewnprincip5}\C_\B}
\AIC{\vdash \B}
\RL{\wnwk}
\UIC{\vdash \wn A, \tikzmark{prewnprincip4}\B}
\AIC{\tikzmark{prewnprincip2}\ocProofs_{\wn A}}
\RL{\rmcutpar}
\TIC{\vdash \tikzmark{prewnprincip3}\A, \tikzmark{prewnprincip1}\wn\A'}
\DP\rightsquigarrow
\AIC{\C_\B}
\AIC{\vdash \B}
\RL{\rmcutparprime}
\BIC{\vdash \A}
\doubleLine
\RL{\wnwk}
\UIC{\vdash \A, \wn\A'}
\DP
$\\[2ex]
\begin{tikzpicture}[overlay,remember picture,-,line cap=round,line width=0.1cm]
    \draw[rounded corners, smooth=2,green, opacity=.4] ([xshift=-2mm] pic cs:prewnprincip1) to ([xshift=2mm, yshift=2mm] pic cs:prewnprincip2);
    \draw[rounded corners, smooth=2,red, opacity=.4] (pic cs:prewnprincip3) to ([xshift=2mm, yshift=2mm] pic cs:prewnprincip4);
    \draw[rounded corners, smooth=2,red, opacity=.4] (pic cs:prewnprincip3) to ([xshift=2mm, yshift=2mm] pic cs:prewnprincip5);
\end{tikzpicture}
$
\AIC{\vdash A, \B}
\RL{\wnde}
\UIC{\vdash \wn A, \B}
\AIC{\vdash A^\perp, \wn\B'}
\RL{\ocprom}
\UIC{\vdash \oc A^\perp, \wn\B'}
\AIC{\C}
\RL{\rmcutpar}
\TIC{\vdash \A}
\DP\rightsquigarrow
\AIC{\vdash A, \B}
\AIC{\vdash A^\perp, \wn\B'}
\AIC{\C}
\RL{\rmcutparprime}
\TIC{\vdash \A}
\DP
$\\[2ex]
$
\AIC{\tikzmark{precontrprincip5}\C_\B}
\AIC{\vdash \wn A, \wn A, \B}
\RL{\wncontr}
\UIC{\vdash \wn A, \tikzmark{precontrprincip4}\B}
\AIC{\tikzmark{precontrprincip2}\ocProofs_{\wn A}}
\RL{\rmcutpar}
\TIC{\vdash \tikzmark{precontrprincip3}\A, \tikzmark{precontrprincip1}\wn\A'}
\DP\rightsquigarrow
\AIC{\C_\B}
\AIC{\vdash \wn A, \wn A,  \B}
\AIC{\ocProofs_{\wn A}}
\AIC{\ocProofs_{\wn A}}
\RL{\rmcutparprime}
\QIC{\vdash \wn\A', \wn\A', \A}
\doubleLine
\RL{\wncontr}
\UIC{\vdash \A, \wn\A'}
\DP
$
\begin{tikzpicture}[overlay,remember picture,-,line cap=round,line width=0.1cm]
    \draw[rounded corners, smooth=2,green, opacity=.4] ([xshift=-2mm] pic cs:precontrprincip1) to ([xshift=2mm, yshift=2mm] pic cs:precontrprincip2);
    \draw[rounded corners, smooth=2,red, opacity=.4] (pic cs:precontrprincip3) to ([xshift=2mm, yshift=2mm] pic cs:precontrprincip4);
    \draw[rounded corners, smooth=2,red, opacity=.4] (pic cs:precontrprincip3) to ([xshift=2mm, yshift=2mm] pic cs:precontrprincip5);
\end{tikzpicture}
        \caption{Multicut-elimination steps of the exponential fragment of \musuperLLinf
        \label{fig:muLLonestep}}
    \end{figure*}

\section{Details on \Cref{secsuperLL}}
\subsection{Proof of Axiom Expansion property 
}
\label{app:axexp}

\begin{lem}[Axiom Expansion]
One-step axiom expansion holds for formulas $\wn_\e A$ and $\oc_\e A$ in $\superLL(\Sig, \leqg, \leqf, \lequ)$ if $\e$ satisfies the following \emph{expansion axiom}:
$$ \e\lequ\e\quad \lor\quad \e\leqf\e \quad\lor\quad (\e\leqg\e \land \e(\mpx{1})).$$

The axiom expansion holds in $\superLL(\Sig, \leqg, \leqf, \lequ)$ if all $\e$ satisfy the expansion axiom.
\end{lem}
\begin{proof}
We start by proving the first part of the theorem. We distinguish three cases depending on which branch of the disjunction holds for $\e$:
\begin{itemize}
	\item If $\e\lequ\e$ is true, then we have:
	$$
	\AIC{\vdash A^\perp,A}
	\ZIC{\e\lequ\e}
	\RL{\ocu}
	\BIC{\vdash\oc_\e A^\perp,\wn_\e A}
    \DP
$$
        \item If $\e\leqf\e$ is true, it is similar to the previous case:
	$$
	\AIC{\vdash A^\perp,A}
	\ZIC{\e\leqf\e}
	\RL{\ocf}
	\BIC{\vdash\oc_\e A^\perp,\wn_\e A}
\DP     $$
	\item And if $\e\leqg\e$ and $\overline{ (\e)}(\mpx{1})$:
	$$
	\AIC{\vdash A^\perp,A}
	\ZIC{\overline{ (\e)}(\mpx{1})}
	\RL{\overline{\mpx{1}}}
	\BIC{\vdash A^\perp,\wn_\e A}
	\ZIC{\e\leqg\e}
	\RL{\ocg}
	\BIC{\vdash\oc_\e A^\perp,\wn_\e A}
	\DP
	$$
\end{itemize}
The second part of the theorem is proved by induction on the size of the formula, 
using the first part of the theorem.
\end{proof}

\subsection{Proof of cut-elimination of \superLL{} (\Cref{superllcutelim})}\label{app:superllcutelim}

We first need three lemmas called the substitution lemmas:
\begin{lem}[Girard Substitution Lemma]
\label{subs1}
Let $\e_1$ be a signature and $\vec{\e_2}$ a list of signatures such that $\e_1\leqg\vec{\e_2}$.
Let $A$ be a formula, and let $\B$ be a context, such that for all $\A$, if $\vdash A,\A$ is provable without using any cut then $\vdash\wn_{\vec{\e_2}}\B,\A$ is provable without using any cut.
Then we have that for all $\A$, if $\vdash\overbrace{\wn_{\e_1} A,\dotsc,\wn_{\e_1}A}^n,\A$ is provable without using any cut then $\vdash\overbrace{\wn_{\vec{\e_2}}\B,\dotsc,\wn_{\vec{\e_2}}\B}^n,\A$.
\end{lem}

\begin{proof}
First we can notice that for any $\A$ the following rule:
$$
\AIC{\vdash A,\dotsc,A,\A}
\dashedLine
\RL{S_g}
\UIC{\vdash\wn_{\vec{\e_2}}\B,\dotsc,\wn_{\vec{\e_2}}\B,\A}
\DP
$$
is admissible in the system without cuts (by an easy induction on the number of $A$).

Now we show the lemma by induction on the proof of\linebreak$\vdash\wn_{\e_1} A,\dotsc,\wn_{\e_1}A,\A$. We distinguish cases according to the last rule:

\begin{itemize}
	\item If it is a rule on a formula of $\A$ which is not a promotion: 
\begin{equation*}
	\AIC{\pi}
	\noLine
	\RL{}
	\UIC{\vdash\wn_{\e_1} A,\dotsc,\wn_{\e_1} A,\A'}
	\RL{r}
	\UIC{\vdash\wn_{\e_1} A,\dotsc,\wn_{\e_1} A,\A}
	\DP
	\qquad\rightsquigarrow\qquad
	\AIC{IH(\pi)}
	\noLine
	\RL{}
	\UIC{\vdash\wn_{\vec{\e_2}}\B,\dotsc,\wn_{\vec{\e_2}} \B,\A'}
	\RL{r}
	\UIC{\vdash\wn_{\vec{\e_2}}\B,\dotsc,\wn_{\vec{\e_2}} \B,\A}
	\DP
\end{equation*}
	\item If it is a Girard's style promotion, thanks to the axiom~\refleqTrans{}, we have:
$$
	\AIC{\pi}
	\noLine	
	\RL{}
	\UIC{\vdash B,\wn_{\vec{\e_3}}\A',\wn_{\e_1} A,\dotsc,\wn_{\e_1}A}
	\AIC{\e_0\leqg\vec{\e_3}}
    \AIC{\e_0\leqg\e_1}
	\RL{\ocg}
	\TIC{\vdash\oc_{\e_0}B,\wn_{\vec{\e_3}}\A',\wn_{\e_1} A,\dotsc,\wn_{\e_1} A}
	\DP
	\qquad\rightsquigarrow
	$$
$$
	{\footnotesize\AIC{IH(\pi)}
	\noLine	
	\RL{}
	\UIC{\vdash B,\wn_{\vec{\e_3}}\A',\wn_{\vec{\e_2}}\B,\dotsc,\wn_{\vec{\e_2}}\B}
	\AIC{\e_0\leqg\vec{\e_3}}
	\AIC{\e_0\leqg\e_1}
	\ZIC{\e_1\leqg\vec{\e_2}}
	\RL{\text{\refleqTrans}}
	\BIC{\e_0\leqg\vec{\e_2}}
	\RL{\ocg}
	\TIC{\vdash\oc_{\e_0} B,\wn_{\vec{\e_3}}\A',\wn_{\vec{\e_2}}\B,\dotsc,\wn_{\vec{\e_2}}\B}
	\DP}
$$
    \item If it is a unary promotion, we use axiom~\reflequs{}:
	$$
	\AIC{\pi}
	\noLine
	\UIC{\vdash B,A}
	\AIC{\e_0\lequ\e_1}
	\RL{\ocu}
	\BIC{\vdash\oc_{\e_0}B,\wn_{\e_1}A}
	\DP
	\rightsquigarrow
	$$
	$$
	\AIC{\pi}
	\noLine
	\UIC{\vdash B,A}
	\dashedLine
	\RL{S_g}
	\UIC{\vdash B,\wn_{\vec{\e_2}}\B}
	\AIC{\e_0\lequ\e_1}
	\ZIC{\e_1\leqg\vec{\e_2}}
	\RL{\text{\reflequs}}
	\BIC{\e_0\leqg\vec{\e_2}}
	\RL{\ocg}
	\BIC{\vdash\oc_{\e_0}B,\wn_{\vec{\e_2}}\B}
	\DP
	$$
	\item If it is a functorial promotion:
	$$
			\hspace{-0.5cm}\AIC{\pi}
			\noLine
			\UIC{\vdash B,\A',\overbrace{A,\dotsc,A}^n}
			\AIC{\e_0\leqf\e_1}
			\AIC{\e_0\leqf\vec{\e_3}}
			\RL{\ocf}
			\TIC{\vdash\oc_{\e_0} B,\wn_{\vec{\e_3}}\A',\wn_{\e_1}A,\dotsc,\wn_{\e_1}A}
			\DP\quad\rightsquigarrow
	$$
	$$\hspace*{-3.7cm}
{\scriptsize			\AIC{IH(\pi)}
			\noLine
			\UIC{\vdash B,\A',\overbrace{A,\dotsc,A}^n}
			\dashedLine
			\RL{S_g}
			\UIC{\vdash B,\A',\wn_{\vec{\e_2}}\B,\dotsc,\wn_{\vec{\e_2}}\B}
			\AIC{\e_0\leqf\e_1}
			\ZIC{\e_1\leqg\vec{\e_2}}
			\AIC{\hspace{-0.5cm}e_0\leqf\vec{e_3}}
			\RL{\text{\refleqfg}}
			\TIC{\overline{ (\vec{\e_3})}(\mpx{1})}
			\doubleLine
			\RL{\overline{\mpx{1}}}
			\BIC{\vdash B,\wn_{\vec{\e_3}}\A',\wn_{\vec{\e_2}}\B,\dotsc,\wn_{\vec{\e_2}}\B}
		\AIC{\hspace{-0.5cm}\e_0\leqf\e_1}
		\ZIC{\hspace{-0.5cm}\e_1\leqg\vec{\e_2}}
		\AIC{\hspace{-0.5cm}\e_0\leqf\vec{e_3}}
		\RL{\text{\refleqfg}}
		\TIC{\hspace{-0.5cm}\e_0\leqg\vec{\e_3}}
		\AIC{\e_0\leqf\e_1}
		\ZIC{\e_1\leqg\vec{\e_2}}
		\RL{\text{\refleqfg}}
		\BIC{\e_0\leqg\vec{\e_2}}
		\RL{\ocg}
		\TIC{\vdash\oc_{\e_0}B,\wn_{\vec{\e_3}}\A',\wn_{\vec{\e_2}}\B,\dotsc,\wn_{\vec{\e_2}}\B}
		\DP}
		$$
	\item If it is a contraction ($\contr{i}$) on a $\wn_{\e_1}A$, we use axiom~\refcontrAx{}:
	$$
		\hspace{-0.5cm}
		\AIC{\pi}
		\noLine
		\UIC{\vdash\overbrace{\wn_{\e_1}A,\dotsc,\wn_{\e_1}A}^{i+n-1},\A}
		\AIC{ (\e_1)(\contr{i})}
		\RL{\contr{i}}
		\BIC{\vdash\wn_{\e_1}A,\dotsc,\wn_{\e_1}A,\A}
		\DP
	$$
	$$
		\hspace{2cm}\AIC{IH(\pi)}
		\noLine
		\UIC{\vdash\overbrace{\wn_{\vec{\e_2}}\B,\dotsc,\wn_{\vec{\e_2}}\B}^{n-1+i},\A}
		\AIC{ (\e_1)(\contr{i})}
		\ZIC{\e_1\leqg\vec{\e_2}}
		\RL{\text{\refcontrAx}}
		\BIC{\overline{ (\vec{\e_2})}(\contr{i})}
		\dashedLine
		\doubleLine
		\RL{\overline{\contr{i}}}
		\BIC{\vdash\wn_{\vec{\e_2}}\B,\dotsc,\wn_{\vec{\e_2}}\B,\A}
		\DP
	$$
	\item If it is a multiplexing ($\mpx{i}$) on a $\wn_{\e_1}A$, we use axiom~\refgmpxAx{}:
	$$
		\AIC{\pi}
		\noLine
		\UIC{\vdash\wn_{\e_1}A,\dotsc,\wn_{\e_1}A,\overbrace{A, \dots, A}^i,\wn_{\e_1}A,\dotsc,\wn_{\e_1}A,\A}
		\AIC{ (\e_1)(\mpx{i})}
		\RL{\mpx{i}}
		\BIC{\vdash\wn_{\e_1}A,\dotsc,\wn_{\e_1}A,\A}
		\DP
	\qquad\rightsquigarrow$$
$$		
\AIC{IH(\pi)}
		\noLine
		\UIC{\vdash\wn_{\vec{\e_2}}\B,\dotsc,\wn_{\vec{\e_2}}\B,\overbrace{A,\dotsc,A}^i,\wn_{\vec{\e_2}}\B,\dotsc,\wn_{\vec{\e_2}}\B,\A}
		\dashedLine
		\RL{S_g}
		\UIC{\vdash\wn_{\vec{\e_2}}\B,\dotsc,\wn_{\vec{\e_2}}\B,\A}
		\AIC{ (\e_1)(\mpx{i})}
		\ZIC{\e_1\leqg\vec{e_2}}
		\RL{\text{\refgmpxAx}}
		\BIC{\overline{ (\vec{\e_2})}(\contr{i})}
		\dashedLine		
		\doubleLine
		\RL{\overline{\contr{i}}}
		\BIC{\vdash\wn_{\vec{\e_2}}\B,\dotsc,\wn_{\vec{\e_2}}\B,\A}
		\DP
$$
	\item If it is an ($\ax$) rule on $\wn_{\e_1}A$. Then $\A=\oc_{\e_1}A^\perp$ and we have:
	$$
		\RL{\ax}
		\ZIC{\vdash A^\perp,A}
		\dashedLine
		\RL{S_g}
		\UIC{\vdash A^\perp,\wn_{\vec{\e_2}}\B}
        \ZIC{\e_1\leqg\vec{e_2}}
		\RL{\ocg}
		\BIC{\vdash\oc_{\e_1} A^\perp,\wn_{\vec{\e_2}}\B}
\DP
$$
\end{itemize}
\end{proof}

\begin{lem}[Functorial Substitution Lemma]
\label{subs2}
Let $\e_1$ be a signature and $\vec{\e_2}$ a list of signatures such that $\e_1\leqf\vec{\e_2}$.
Let $A$ be a formula, and let $\B$ be a context, such that for all $\A$, if $\vdash A,\A$ is provable without using any cut then $\vdash\B,\A$ is provable without using any cut.
Then we have that for all $\A$, if $\vdash\overbrace{\wn_{\e_1}A,\dotsc,\wn_{\e_1}A}^n,\A$ is provable without using any cut then $\vdash\overbrace{\wn_{\vec{\e_2}}\B,\dotsc,\wn_{\vec{\e_2}}\B}^n,\A$ as well.
\end{lem}
\begin{proof}
First we can notice that for any $\A$ the following rule:
$$
\AIC{\vdash A,\dotsc,A,\A}
\dashedLine
\RL{S_f}
\UIC{\vdash\B,\dotsc,\B,\A}
\DP
$$
is admissible in the system without cuts (by an easy induction on the number of $A$).
Now we show the lemma by induction on the proof of $\vdash\wn_{\e_1}A,\dotsc,\wn_{\e_1}A,\A$. We distinguish cases according to the last applied rule :
\begin{itemize}
	\item If it is a rule on a formula of $\A$ which is not a promotion: 
\begin{equation*}
	\AIC{\pi}
	\noLine
	\UIC{\vdash\wn_{\e_1}A,\dotsc,\wn_{\e_1} A,\A'}
	\RL{r}
	\UIC{\vdash\wn_{\e_1}A,\dotsc,\wn_{\e_1}A,\A}
	\DP
	\qquad\rightsquigarrow\qquad
	\AIC{IH(\pi)}
	\noLine
	\UIC{\vdash\wn_{\vec{\e_2}}\B,\dotsc,\wn_{\vec{\e_2}}\B,\A'}
	\RL{r}
	\UIC{\vdash\wn_{\vec{\e_2}}\B,\dotsc,\wn_{\vec{\e_2}}\B,\A}
	\DP
\end{equation*}
	\item If it is a Girard's style promotion. Thanks to the axiom~\refleqgs{}, we have:
	$$
	\AIC{\pi}
	\noLine	
	\UIC{\vdash B,\wn_{\vec{\e_3}}\A',\wn_{\e_1}A,\dotsc,\wn_{\e_1}A}
	\AIC{\e_0\leqg\vec{\e_3}}
	\AIC{\e_0\leqg\e_1}
	\RL{\ocg}
	\TIC{\vdash\oc_{\e_0}B,\wn_{\vec{\e_3}}\A',\wn_{\e_1}A,\dotsc,\wn_{\e_1}A}
	\DP
	\qquad\rightsquigarrow
	$$
	$$
	\AIC{IH(\pi)}
	\noLine
	\UIC{\vdash B,\wn_{\vec{\e_3}}\A',\wn_{\vec{\e_2}}\B,\dotsc,\wn_{\vec{\e_2}}\B}
	\AIC{\e_0\leqg\vec{\e_3}}
	\AIC{\e_0\leqg\e_1}
	\ZIC{\e_1\leqf\vec{\e_2}}
	\RL{\text{\refleqgs}}
	\BIC{\e_0\leqg\vec{\e_2}}
	\RL{\ocg}
	\TIC{\vdash\oc_{\e_0}B,\wn_{\vec{\e_3}}\A',\wn_{\vec{\e_2}}\B,\dotsc,\wn_{\vec{\e_2}}\B}
	\DP
	$$
	\item If it is a unary promotion, we use axiom~\reflequs{}:
	$$
	\hspace{-2cm}\AIC{\pi}	
	\noLine
	\UIC{\vdash B,A}
	\AIC{\e_0\lequ\e_1}
	\RL{\ocu}
	\BIC{\vdash\oc_{\e_0}B,\wn_{\e_1}A}
	\DP
$$	

$$\AIC{\pi}
	\noLine
	\UIC{\vdash B,A}
	\dashedLine
	\RL{S_f}
	\UIC{\vdash  B,\B}
	\AIC{\e_0\lequ\e_1}
	\ZIC{\e_1\leqf\vec{\e_2}}
	\RL{\text{\reflequs}}
	\BIC{\e_0\leqf\vec{\e_2}}
	\RL{\ocf}
	\BIC{\vdash\oc_{\e_0}B,\wn_{\vec{\e_2}}\B}
	\DP
$$
	\item If it is a functorial promotion, thanks to the axiom~\refleqTrans{} we have:
	$$
	\AIC{\pi}
	\noLine
	\UIC{\vdash B,\A', A,\dotsc,A}
	\AIC{\e_0\leqf\vec{e_3}}
	\AIC{\e_0\leqf\e_1}
	\RL{\ocf}
	\TIC{\vdash\oc_{\e_0}B,\wn_{\vec{e_3}}\A',\wn_{\e_1}A,\dotsc,\wn_{\e_1}A}
	\DP
	\qquad\rightsquigarrow
	$$
	$$
	\AIC{IH(\pi)}
	\noLine
	\UIC{\vdash B,\A', A,\dotsc,A}
	\dashedLine
	\RL{S_f}
	\UIC{\vdash B,\A',\wn_{\vec{\e_2}}\B,\dotsc,\wn_{\vec{\e_2}}\B}
	\AIC{\e_0\leqf\vec{e_3}}
	\AIC{\e_0\leqf\e_1}
	\AIC{\e_1\leqf\vec{\e_2}}
	\RL{\text{\refleqTrans}}
	\BIC{\e_0\leqf\vec{\e_2}}
	\RL{\ocf}
	\TIC{\vdash\oc_{\e_0} B,\wn_{\vec{e_3}}\A',\wn_{\vec{\e_2}}\B,\dotsc,\wn_{\vec{\e_2}}\B}
	\DP
	$$
	\item If it is a contraction ($\contr{i}$) on $\wn_{\e_1}A$, we use axiom~\reffumpxAx{}:
	$$
		\AIC{\pi}
		\noLine
		\UIC{\vdash,\overbrace{\wn_{\e_1}A,\dotsc,\wn_{\e_1}A}^{n+i-1},\A}
		\AIC{(\e_1)(\contr{i})}
		\RL{\contr{i}}
		\BIC{\vdash\wn_{\e_1} A, \dots, \wn_{\e_1} A,\A}
		\DP
	\qquad\rightsquigarrow
	$$
	$$
		\AIC{IH(\pi)}
		\noLine
		\UIC{\vdash\overbrace{\wn_{\vec{\e_2}}\B,\dotsc,\wn_{\vec{\e_2}}\B}^{n+i-1},\A}
		\AIC{ (\e_1)(\contr{i})}
		\ZIC{\e_1\leqf\vec{e_2}}
		\RL{\text{\reffumpxAx}}
		\BIC{\overline{ (\vec{\e_2})}(\contr{i})}
		\dashedLine
		\doubleLine
		\RL{\overline{\contr{i}}}
		\BIC{\vdash\wn_{\vec{\e_2}}\B,\dotsc,\wn_{\vec{\e_2}}\B,\A}
		\DP
	$$
	\item If it is a multiplexing ($\mpx{i}$) on $\wn_{\e_1}A$, we use axiom~\reffumpxAx:
	$$
		\AIC{\pi}
		\noLine
		\UIC{\vdash\wn_{\e_1}A,\dotsc,\wn_{\e_1}A,\overbrace{A,\dotsc,A}^i,\wn_{\e_1}A,\dotsc,\wn_{\e_1}A,\A}
		\AIC{ (\e_1)(\mpx{i})}
		\RL{\mpx{i}}
		\BIC{\vdash\wn_{\e_1}A,\dotsc,\wn_{\e_1}A,\A}
		\DP\qquad
	\rightsquigarrow
	$$
	$$
		\AIC{IH(\pi)}
		\noLine
		\UIC{\vdash\wn_{\vec{\e_2}}\B,\dotsc,\wn_{\vec{\e_2}}\B,\overbrace{A,\dotsc,A}^i,\wn_{\vec{\e_2}}\B,\dotsc,\wn_{\vec{\e_2}}\B,\A}
		\dashedLine
		\RL{S_f}
		\UIC{\vdash\wn_{\vec{\e_2}}\B,\dotsc,\wn_{\vec{\e_2}}\B,\B,\dotsc,\B,\wn_{\vec{\e_2}}\B,\dotsc,\wn_{\vec{\e_2}}\B,\A}
		\AIC{(\e_1)(\mpx{i})}
		\ZIC{\e_1 \leqf \vec{\e_2}}
		\RL{\text{\reffumpxAx{}}}
		\BIC{\overline{(\vec{\e_2})}(\mpx{i})}
		\dashedLine
		\doubleLine
		\RL{\overline{\mpx{i}}}
		\BIC{\vdash\wn_{\vec{\e_2}}\B,\dotsc,\wn_{\vec{\e_2}}\B,\A}
		\DP
	$$
	\item If it is an ($\ax$) rule on $\wn_{\e_1} A$. Then $\A=\oc_{\e_1} A^\perp$ and we have:
	$$
		\RL{\ax}
		\ZIC{\vdash A^\perp,A}
		\dashedLine
		\RL{S_f}
		\UIC{\vdash A^\perp,\B}
                \ZIC{\e_1\leqf\vec{e_2}}
		\RL{\ocf}
		\BIC{\vdash\oc_{\e_1} A^\perp,\wn_{\vec{\e_2}} \B}
		\DP
	$$
\end{itemize}
\end{proof}

\begin{lem}[Unary Functorial Substitution Lemma]
\label{subs3}
Let $\e_1$ and $\e_2$ be two exponential signatures such that $\e_1\lequ\e_2$.
Let $A$ and $B$ be formulas, such that for all $\A$, if $\vdash A,\A$ is provable without using any cut then $\vdash B,\A$ is provable without using any cut.
Then we have that for all $\A$, if $\vdash\overbrace{\wn_{\e_1}A,\dotsc,\wn_{\e_1}A}^n,\A$ is provable without using any cut then $\vdash\overbrace{\wn_{\e_2}B,\dotsc,\wn_{\e_2}B}^n,\A$ as well, with $k_i$ positive integers.
\end{lem}
\begin{proof}
This lemma is proven the same way as Lemma~\ref{subs2}.
\end{proof}

Finally we prove cut-elimination theorem~\ref{superllcutelim}:
\begin{thm}[Cut Elimination]
Cut elimination holds for\linebreak $\superLL(\Sig,\leqg,\leqf,\lequ)$ as soon as the 8 cut-elimination axioms of \Cref{cutElimAxs} are satisfied.
\end{thm}
\begin{proof}
We prove the result by induction on the couple $(t, s)$ with lexicographic order, where $t$ is the size of the cut formula and $s$ is the sum of the sizes of the premises of the cut. We distinguish cases depending on the last rules of the premises of the cut:
	\begin{itemize}
	  \item If one of the premises does not end with a rule acting on the cut formula, we apply the induction hypothesis with the premise(s) of this rule.
          \item If both last rules act on the cut formula which does not start with an exponential connective, we apply the standard reduction steps for non-exponential cuts leading to cuts involving strictly smaller cut formulas. We conclude by applying the induction hypothesis.
          \item If we have an exponential cut for which the cut formula $\oc_{\e_1}A^\perp$ is not the conclusion of a promotion rule introducing $\oc_{\e_1}$, the rule above $\oc_{\e_1}A^\perp$ cannot be a promotion rule, and we apply the induction hypothesis to its premise(s).
          \item If we have an exponential cut for which the cut formula $\oc_{\e_1}A^\perp$ is the conclusion of an ($\ocg$)-rule. We can apply:
          $$
			\AIC{\vdash A^\perp,\wn_{\vec{\e_2}}\B}
			\AIC{\e_1\leqg\vec{\e_2}}
			\RL{\ocg}
			\BIC{\vdash\oc_{\e_1}A^\perp,\wn_{\vec{\e_2}}\B}
			\AIC{\vdash\wn_{\e_1} A,\A}
			\RL{\cut}
			\BIC{\vdash\wn_{\vec{\e_2}}\B,\A}
			\DP
			\quad\rightsquigarrow\quad
			\AIC{\vdash\wn_{\e_1}A,\A}
			\AIC{\e_1\leqg\vec{\e_2}}
			\dashedLine			
			\RL{\text{Lem.}~\ref{subs1}}
			\BIC{\vdash\wn_{\vec{\e_2}}\B,\A}
			\DP
		$$
		We have that $A$ and $\B$ are such that for every $\A$ such that $\vdash A,\A$ is provable without cuts, $\vdash\wn_{\vec{\e_2}}\B,\A$ too. Indeed, $A$ and $\B$ are such that $\vdash A^\perp,\wn_{\vec{\e_2}}\B$ is provable without cuts and we can apply the induction hypothesis ($\size{A}<\size{\wn_{\e_1} A}$). Therefore, we can apply Lemma~\ref{subs1} on $\vdash\wn_{\e_1}A,\A$ and obtain that $\vdash\wn_{\vec{\e_2}}\B,\A$ is provable without cut.
		\item If we have an exponential cut for which the cut formula $\oc_{\e_1}A^\perp$ is the conclusion of an ($\ocf$)-rule. We can apply:
		$$
			\AIC{\vdash A^\perp,\B}
			\AIC{\e_1\leqf\vec{\e_2}}
			\RL{\ocf}
			\BIC{\vdash\oc_{\e_1}A^\perp,\wn_{\vec{\e_2}}\B}
			\AIC{\vdash\wn_{\e_1}A,\A}
			\RL{\cut}
			\BIC{\vdash\wn_{\vec{\e_2}}\B,\A}
			\DP
			\quad\rightsquigarrow\quad
			\AIC{\vdash\wn_{\e_1}A,\A}
			\AIC{\e_1\leqf\vec{\e_2}}
			\dashedLine			
			\RL{\text{Lem.}~\ref{subs2}}
			\BIC{\vdash\wn_{\vec{\e_2}}\B,\A}
			\DP
		$$
		We have that $A$ and $\B$ are such that for every $\A$ such that $\vdash A,\A$ is provable without cuts, $\vdash\B,\A$ too. Indeed, $A$ and $\B$ are such that $\vdash A^\perp,\B$ is provable without cuts and we can apply the induction hypothesis. Therefore, we can apply \Cref{subs2} on $\vdash\wn_{\e_1}A,\A$ and obtain that $\vdash\wn_{\vec{\e_2}}\B,\A$ is provable without cut.
		\item If we have an exponential cut for which the cut formula $\oc_{\e_1}A^\perp$ is the conclusion of an ($\ocu$)-rule, this case is treated in the exact same way as ($\ocf$), using \Cref{subs3}.
	\end{itemize}
\end{proof}

\subsection{Details on \ELL{} as instance of \superLL}\label{app:ELLexa}

\paragraph{Elementary Linear Logic.} Elementary Linear Logic (\ELL)~\cite{lll,djell} is a variant of \LL{} where we remove ($\wnde$) and $(\ocg)$ and add the functorial promotion:
\begin{equation*}
  \AIC{\vdash A,\A}
  \RL{\ocf}
  \UIC{\vdash \oc A,\wn\A}
  \DP
\end{equation*}
It is the $\superLL(\Sig, \leqg, \leqf, \lequ)$ system with $\Sig=\{\bullet\}$, defined by $\bullet(\contr{2})=\bullet(\mpx{0})=\True$ (and $(\bullet)(r)=\False$ otherwise), ${\leqg}={\lequ}=\emptyset$ and $\bullet\leqf\bullet$.
This $\superLL(\Sig, , {\leqg}, {\leqf}, {\lequ})$ instance is \ELL{} and satisfies the cut-elimination axioms and the expansion axiom:
\begin{itemize}
\item The rule ($\mpx{0}$) is the weakening rule ($\wnwk$), ($\contr{2}$) is the contraction rule ($\contr{}$), and we can always apply promotion ($\ocf$) as ${\leqf}$ is the plain relation on $\Sig$:
\begin{equation*}
\AIC{\vdash A,\A}
\ZIC{\bullet\leqf\bullet}
\RL{\ocf}
\BIC{\vdash\oc_{\bullet}A,\wn_{\bullet}\A}
\DP
\qquad\leftrightsquigarrow\qquad
\AIC{\vdash A,\A}
\RL{\oc_f}
\UIC{\vdash\oc A,\wn A}
\DP
\end{equation*}
We have that ($\ocg$) is a restriction of ($\ocf$) in \ELL{} and ($ \ocu $) is non-existent.

\item Moreover, the cut-elimination axioms are satisfied.
As $\Sig$ is a singleton, axioms \refgmpxAx, \reffumpxAx, \refcontrAx, \refleqTrans, \refleqgs, \refleqfu, \reflequs{} hold. Axiom \refleqfg{} is vacuously satisfied.

\item The expansion axiom is satisfied since $\leqf$ is reflexive.
\end{itemize}

\section{Details on \Cref{section:musuperll}}
\label{app:musuperll}

\subsection{Details on \muLLmodinf{} as an instance of \musuperLLinf{}}
\label{app:mullmodIsMusuperLL}

We show here in details how the system \muLLmodinf{} is an instance of super exponentials. 

\muLLmodinf{} coincides with  the system $\musuperLLinf(\Sig, \leqg, \leqf, \lequ)$ such that:
\begin{itemize}
\item The set of signatures contains two elements $\Sig:=\{\bullet, \star\}$.
\item $\contr{2}(\bullet) = \contr{2}(\star) = \True$
\item $\mpx{1}(\bullet)=\True$, 
\item $\mpx{0}(\bullet)=\mpx{0}(\star)=\True$, 
\item all the other elements have value $\False$ for both signatures.
\item $\bullet\leqg\bullet$ ; $\bullet \leqg\star$, $\star\leqf\star$, and all other couples for the three relations $\leqg, \leqf$ and $\lequ$ being false.
\end{itemize}

\bigskip
This system is \muLLmodinf{} when taking:
$$ \wn_{\bullet} := \wn,\quad \oc_{\bullet}:=\oc,\quad \wn_{\star}:= \lozenge \quad\text{and}\quad \oc_{\star}:= \Box. $$

\bigskip
We can indeed check that the system satisfies the cut-elimination axioms of \Cref{cutElimAxs}:
\begin{itemize}
    \item Hypotheses of axiom \refcontrAx{} are ony true for $i=2$ in two cases: for $\e=\e'=\bullet$, in that case $\bar{\e}(\contr{2}$ is true because $\e(\contr{2})$ is; or for $\e=\bullet$ and $\e'=\star$, in that case the axiom is satisfied as $\e'(\contr{2})$ is true.
    \item Hypotheses of axiom \refgmpx{} are true for $i=0$ when $\e=\e'=\bullet$, or for $\e=\bullet$ and $\e'=\star$, in both cases we have that $\bar{\e'}(\contr{0})$ is true because $\e'(\mpx{0})$ is true.
    \item Axiom \refgmpx{} is always true for $i=1$
    \item Hypotheses of axiom \refgmpx{} are not satisfied for $i>1$.
    \item Hypotheses of axiom \reffumpx{} are satisfied only for $\e=\e'=\star$ and so easily satisfied.
    \item Axiom \refleqTrans{} is satisfied as $\leqg$ and $\leqf$ are transitive.
    \item Hypotheses of axiom \refleqgs are only satisfied for $\e=\bullet$ and $\e'=\e''=\star$, and in this case the conclusion is one of the hypothesis.
    \item Hypotheses of the other axioms are never fully satisfied.
\end{itemize}

\section{Details on \Cref{section:cut-elimination}}

\subsection{Details on the justification of (\mcut)-steps}
\label{app:mcutstepJustif}

In the following, we shall prove the lemmas justifying the mcut-reduction steps.
The following statement are identical to those found in the body of the paper but for the fact that we make explicit the side conditions on the exponential rules: in the hypotheses of the lemmas, such side-conditions are assumptions we can use in our proof while in the conclusion derivation these side-conditions are goals to be proved in order to establish that the derivation is indeed a proof in the considered $\musuperLLinf(\Sig, \leqg, \leqf, \lequ)$ system.

\subsubsection{Justification for step~{\hyperref[ocgcomm]{$({\text{comm}}_\ocg)$}}: proof of \Cref{ocgreadinesscondition}}
\label{app:ocgreadinesscondition}

The case~\hyperref[ocgcomm]{$({\text{comm}}_\ocg)$}  covers all the case where $(\ocg)$ commute under the cut:
\begin{lem}[Justification for step~{\hyperref[ocgcomm]{$({\text{comm}}_\ocg)$}}]
If\quad
$$
\AIC{\pi}
\noLine
\UIC{\vdash A, \wn_{\vec{\f}}\B}
\AIC{\e\leqg\vec{\f}}
\RL{\ocg}
\BIC{\vdash \oc_{\e} A, \wn_{\vec{\f}}\B}
\AIC{\ocProofs}
\RL{\rmcutpar}
\BIC{\vdash \oc_{\e} A, \wn_{\vec{\g}}\A}
\DP
$$
is a $\musuperLLinf(\Sig, \leqg, \leqf, \lequ)$-proof then\quad
$$
\AIC{\pi}
\noLine
\UIC{\vdash A, \wn_{\vec{\f}}\B}
\AIC{\ocProofs}
\RL{\rmcutpar}
\BIC{\vdash A, \wn_{\vec{\g}}\A}
\AIC{\e\leqg\vec{\g}}
\RL{\ocg}
\BIC{\vdash \oc_{\e} A, \wn_{\vec{\g}}\A}
\DP
$$
is also a $\musuperLLinf(\Sig, \leqg, \leqf, \lequ)$-proof.
\end{lem}
\begin{proof}
We prove that for each sequent $\vdash \oc_{\e'} A', \wn_{\vec{\f'}}\B'$ of $\C':=\ocProofs \cup \{\vdash \oc_{\e} A, \wn_{\vec{\f}}\B\}$, we have that $\e\leqg\vec{\f'}$.

The $\cutrel$-relation extended to sequent defines a tree on $\C'$. Taking $\vdash \oc_{\e} A, \wn_{\vec{\f}}\B$ as the root, the ancestor relation of this tree is a well-founded relation. We can therefore do a proof by induction:
\begin{itemize}
\item The base case is given by the condition of application of $(\ocg)$ in the proof.
\item For heredity,  we have that there is a sequent $\vdash \oc_{\e''} A'', \wn_{\vec{\f''}}\B'', \wn_{\e'} ({A'}^\perp)$ of $\C'$, connected on $\oc_{\e'} A'$ to our sequent.
By induction hypothesis, we have that $\e\leqg\e'$.
The rule on top of $\vdash \oc_{\e'} A', \wn_{\vec{\f'}}\B'$ is a promotion. We have two cases:
\begin{itemize}
\item If it's a $(\ocg)$-promotion, we can use axiom~\refleqTrans{} with the application condition of the promotion, to get $\e\leqg\vec{\f'}$.
\item If it's an $(\ocf)$-promotion or an $(\ocu)$-promotion, we can use axiom~\refleqgs{} with the application condition of the promotion, to get $\e\leqg\vec{\f'}$.
\end{itemize}
\end{itemize}
We conclude by induction and use the inequalities to prove that $\e\leqg\vec{\g}$.
\end{proof}

\subsubsection{Justification for step~{\hyperref[ocfcommocgEmpty]{$({\text{comm}}^1_\ocf)$}}: proof of \Cref{ocfreadinesscondition}}
\label{app:ocfreadinesscondition}

The case~\hyperref[ocfcommocgEmpty]{$({\text{comm}}^1_\ocf)$} covers the case of commutation of an $(\ocf)$-promotion but where only $(\ocg)$-rules with empty contexts appears in the hypotheses of the multi-cut.
Note that an $(\ocg)$ occurrence with empty context could be seen as an $(\ocf)$ occurrence (with empty context).
\begin{lem}[Justification for step~{\hyperref[ocfcommocgEmpty]{$({\text{comm}}^1_\ocf)$}}]
If
$$
\AIC{\pi}
\noLine
\UIC{\vdash A, \B}
\AIC{\e\leqf\vec{\f}}
\RL{\ocf}
\BIC{\vdash \oc_{\e} A, \wn_{\vec{\f}}\B}
\AIC{\ocProofs}
\RL{\rmcutpar}
\BIC{\vdash \oc_{\e} A, \wn_{\vec{\g}}\A}
\DP
$$
is a $\musuperLLinf(\Sig, \leqg, \leqf, \lequ)$-proof with $\ocProofs$ such that each sequents concluded by an $(\ocg)$ have an empty context, then\quad
$$
\AIC{\pi}
\noLine
\UIC{\vdash A, \B}
\AIC{\C}
\RL{\rmcutpar}
\BIC{\vdash A, \A}
\AIC{\e\leqf\vec{\g}}
\RL{\ocf}
\BIC{\vdash \oc_{\e} A, \wn_{\vec{\g}}\A}
\DP
$$
is a $\musuperLLinf(\Sig, \leqg, \leqf, \lequ)$-proof.
\end{lem}
\begin{proof}
We prove that for each sequent $\vdash \oc_{\e'} A', \wn_{\vec{\f'}}\B'$ of $\C':=\ocgProofs \cup \{\vdash \oc_{\e} A, \wn_{\vec{\f}}\B\}$, $\e\leqf\vec{\f'}$.

The $\cutrel$-relation extended to sequent defines a tree on $\C'$. Taking $\vdash \oc_{\e} A, \wn_{\vec{\f}}\B$ as the root, the ancestor relation of this tree is a well-founded relation. We can therefore do an induction proof:
\begin{itemize}
\item The base case is given by the condition of application of $(\ocf)$ in the proof.
\item For heredity,  we have that there is a sequent $\vdash \oc_{\e''} A'', \wn_{\vec{\f''}}\B'', \wn_{\e'} ({A'}^\perp)$ of $C'$, connected on $\oc_{\e'} A'$ to our sequent.
By induction hypothesis, we have that $\e\leqf\e'$.
The rule on top of $\vdash \oc_{\e'} A', \wn_{\vec{\f'}}\B'$ is a promotion. We have three cases:
\begin{itemize}
\item If it's an $(\ocg)$-promotion, then the context is empty and the proof is easily satisfied.
\item If it's an $(\ocf)$-promotion, we can use axiom~\refleqTrans{} with the application condition of the promotion to get $\e\leqf\vec{\f'}$.
\item If it's an $(\ocu)$-promotion, we can use axiom \refleqfu{} with the application condition of the promotion to get $\e\leqf\vec{\f'}$.
\end{itemize}
\end{itemize}
We conclude by induction and use the inequalities to prove that $\e\leqf\vec{\g}$.
\end{proof}

\subsubsection{Justification for step~{\hyperref[ocfcommocgNonEmpty]{$({\text{comm}}^2_\ocf)$}}: proof of \Cref{ocfgreadinesscondition}}
\label{app:ocfgreadinesscondition}

We then have the following case where we commute an $(\ocf)$-rule, but where there is one (at least) $(\ocg)$-promotion with a non-empty context in the premisses of the multicut rule:
\begin{lem}[Justification for step~{\hyperref[ocfcommocgNonEmpty]{$({\text{comm}}^2_\ocf)$}}]
If
$$\AIC{\pi}
\noLine
\UIC{\vdash A, \B}
\AIC{\e\leqf\vec{\f}}
\RL{\ocf}
\BIC{\vdash \oc_{\e} A, \wn_{\vec{\f}}\B}
\AIC{\ocProofs}
\RL{\rmcutpar}
\BIC{\vdash \oc_{\e} A, \wn_{\vec{\g}}\A}
\DP$$
is a $\musuperLLinf(\Sig, \leqg, \leqf, \lequ)$-proof with $\ocgProofs$ containing a sequent conclusion of an $(\ocg)$-rule with at least one formula in the context, then
$$
\AIC{\pi}
\noLine
\UIC{\vdash A, \B}
\AIC{\vec{\f}(\mpx{1})}
\doubleLine
\RL{\mpx{1}}
\BIC{\vdash A, \wn_{\vec{\f}}\B}
\AIC{\ocProofs}
\BIC{\vdash A, \wn_{\vec{\g}}\A}
\AIC{\e\leqg\vec{\g}}
\RL{\ocg}
\BIC{\vdash \oc_{\e} A, \wn_{\vec{\g}}\A}
\DP$$
is also a $\musuperLLinf(\Sig, \leqg, \leqf, \lequ)$-proof.
\end{lem}
\begin{proof}
We prove that for each sequent $\vdash \oc_{\e'} A', \wn_{\vec{\f'}}\B'$ of $\ocProofs:=\ocgProofs_1\cup \ocfProofs_2\cup \ocuProofs_3 \cup \{\vdash \oc_{\e} A, \wn_{\vec{\f}}\B\}$, we have that $\e\leqg\vec{\f'}$. Moreover, we prove that $\vec{\f}(\mpx{1})$.
We prove that in two steps:
\begin{enumerate}
\item There is a sequent $\vdash \oc_{\e'} A',  \wn_{\vec{\f'}}\B'$, with $\B'$ being non-empty, which is conclusion of an $(\ocg)$-rule. Let's suppose without loss of generality, that this sequent is the closest such sequent to $S:= \vdash \oc_{\e} A, \wn_{\vec{\f}}\B$.
The $\cutrel$-relation extended to sequents defines a tree with the hypotheses of the multi-cut rule, therefore there is a path from the sequent $S$ to the sequent $S':= \vdash \oc_{\e'} A', \wn_{\vec{\f'}}\B'$, of sequents $\vdash \oc_{\e''} A'', \wn_{\vec{\f''}}\B''$. We prove by induction on this path, starting from $S$ and stopping one sequent before $S'$ that $\e\leqf\f''$:
\begin{itemize}
\item The initialisation comes from the condition of application of $\ocf$ on $S$.

\item For the heredity, we have that the sequent $\vdash \oc_{\e''} A'', \wn_{\vec{\f''}}\B''$ is cut-connected to a $\vdash \oc_{\e^{(3)}} A^{(3)}, \wn_{\vec{\f^{(3)}}}\B^{(3)}, \wn_{\e''}({A''}^\bot)$ on $\oc_{\e''} A''$, therefore $\e\leqf\e''$. We have two cases: either this sequent is the conclusion of an $(\ocu)$-rule and we apply axiom~\refleqfu, either of an $(\ocf)$-rule and we apply axiom~\refleqTrans. In each case, we have that $\e\leqf\vec{\f''}$.
\end{itemize}
We conclude by induction and get a sequent $S'':=\vdash \oc_{\e''} A'', \wn_{\vec{\f''}}\B''$ cut-connected to $S'$ on the formula $\oc_{\e'}A'$ with $\e\leqf\vec{\f''}$. From that we get that $\e\leqf\e'$. Moreover, we have that $\e'\leqg\vec{\f'}$. As $\B'$ is non-empty, there is a signature $\g'\in\vec{\f'}$ such that $\e'\leqg\g'$. We can therefore apply axiom~\refleqfg. We get that for each signatures $\e^{(3)}$ such that $\e\leqf\e^{(3)}$, $\e\leqg\e^{(3)}$ and $\e^{(3)}(\mpx{1})$, which we can apply to $\e$ and $\vec{\f}$ to get that $\e\leqg\vec{\f}$ and $\vec{\f}(\mpx{1})$.

\item Then, we prove by induction on the tree defined with the $\cutrel$-relation and rooted by $S$ that for each sequents $\vdash \oc_{\e''} A'', \wn_{\vec{\f''}}\B''$, $\e\leqg\vec{\f''}$:
\begin{itemize}
\item The initialisation is done with the first step.

\item For heredity, we have that there is a sequent \linebreak$\vdash \oc_{\e^{(3)}} A^{(3)}, \wn_{\vec{\f^{(3)}}}\B^{(3)}, \wn_{\e''} ({A''}^\bot)$ cut-connected to $\vdash \oc_{\e''} A'', \wn_{\vec{\f''}}\B''$ on $\oc_{\e''} A''$, meaning that $\e\leqg\e''$, as the sequent is the conclusion of a promotion, we have that $\e''\leq_s\f''$ for a $s\in\{g,f,u\}$, we conclude using axiom~\refleqgs.
\end{itemize}
\end{enumerate}
We conclude by induction and we use the inequalities from it to prove that $\e\leqg\vec{\g}$.

\end{proof}

\subsubsection{Justification for step~{\hyperref[ocucommOnlyocu]{$({\text{comm}}^1_\ocu)$}}: proof of \Cref{ocureadinesscondition}}
\label{app:ocureadinesscondition}

We then cover the cases where we commute an $(\ocu)$-rule with the multi-cut.
The first case is where there are only a list of $(\ocu)$-rules in the hypotheses of the multi-cut:
\begin{lem}[Justification for step~{\hyperref[ocucommOnlyocu]{$({\text{comm}}^1_\ocu)$}}]
If
$$
\AIC{\pi}
\noLine
\UIC{\vdash A, C}
\AIC{\e\lequ\f}
\RL{\ocu}
\BIC{\vdash \oc_{\e} A, \wn_{\f} C}
\AIC{\ocuProofs}
\RL{\rmcutpar}
\BIC{\vdash \oc_{\e} A, \wn_{\g} B}
\DP
$$
is a $\musuperLLinf(\Sig, \leqg, \leqf, \lequ)$-proof, then
$$
\AIC{\pi}
\noLine
\UIC{\vdash A, C}
\AIC{\C}
\RL{\rmcutpar}
\BIC{\vdash A, B}
\AIC{\e\lequ\g}
\RL{\ocu}
\BIC{\vdash \oc_{\e} A, \wn_{\g} B}
\DP
$$
is a $\musuperLLinf(\Sig, \leqg, \leqf, \lequ)$-proof.
\end{lem}
\begin{proof}
We prove that for each sequent $\vdash \oc_{\e'} A', \wn_{\f'}B'$ of $\C' := \ocuProofs \cup \{\vdash \oc_{\e} A, \wn_{\f} B\}$, we have that $\e\lequ\f'$.

The $\cutrel$-relation extended to sequent defines a tree on $\C'$. Taking $\vdash \oc_{\e} A, \wn_\f B$ as the root, the ancestor relation of this tree is a well-founded relation. We can therefore do an induction proof:
\begin{itemize}
\item The base case is given by the condition of application of $(\ocu)$ in the proof.
\item For heredity,  we have that there is a sequent \linebreak$\vdash \oc_{\e''} A'', \wn_{\f''} B'', \wn_{\e'} ({A'}^\bot)$ of $C'$, connected on $\oc_{\e'} A'$ to our sequent.
By induction hypothesis, we have that $\e\lequ\e'$.
The rule on top of $\vdash \oc_{\e'} A', \wn_{\f'} B'$ is an $(\ocu)$-promotion, we can use axiom~\refleqTrans{} and with the application condition of the promotion, we get that $\e\lequ f'$.
\end{itemize}
We conclude by induction and get that $\e\lequ\g$.
\end{proof}

\subsubsection{Justification for step~{\hyperref[ocucommocgEmpty]{$({\text{comm}}^2_\ocu)$}}: proof of \Cref{ocufreadinesscondition}}
\label{app:ocufreadinesscondition}

The second case of $(\ocu)$-commutation is where we have an $(\ocf)$-rule and where the hypotheses concluded by an $(\ocg)$-rule have empty contexts.
\begin{lem}[Justification for step~{\hyperref[ocucommocgEmpty]{$({\text{comm}}^2_\ocu)$}}]
Let
$$
\AIC{\pi}
\noLine
\UIC{\vdash A, B}
\AIC{\e\lequ\f}
\RL{\ocu}
\BIC{\vdash \oc_{\e} A, \wn_{\f} B}
\AIC{\ocProofs}
\RL{\rmcutpar}
\BIC{\vdash \oc_{\e} A, \wn_{\vec{\g}}\A}
\DP
$$
be a $\musuperLLinf(\Sig, \leqg, \leqf, \lequ)$-proof with $\C$ containing at least one proof concluded by an $(\ocf)$-promotion ; and such that for each sequent conclusion of an $(\ocg)$-promotion has empty context. We have that
$$
\AIC{\pi}
\noLine
\UIC{\vdash A, B}
\AIC{\C}
\RL\rmcutpar
\BIC{\vdash A, \A}
\AIC{\e\leqf\vec{\g}}
\RL{\ocf}
\BIC{\vdash \oc_{\e} A, \wn_{\vec{\g}}\A}
\DP
$$
is also a $\musuperLLinf(\Sig, \leqg, \leqf, \lequ)$-proof.
\end{lem}
\begin{proof}
If one $(\ocf)$-rule has empty contexts, there is only one $(\ocf)$,$\wn_{\vec{\g}}\A$ is empty and therefore $\e\leqf\vec{\g}$ is easily satisfied.
If not, we do our proof in two steps:
\begin{enumerate}
\item As always, we notice that the $\cutrel$-relation extended to sequent defines a tree on $\C'$, meaning that there is a path in this tree, from $S := \vdash \oc_\e A, \wn_\f B$ to a sequent $S' := \vdash \oc_{\e'} A', \wn_{\vec{\f'}}\B$ being the conclusion of an $\ocf$-rule and with $\B$ being non-empty. Without loss of generality, we ask for $S'$ to be the closest such sequent (with respect to the $\cutrel$-relation).
We prove by induction on this path, starting from $S$ and stopping one sequent before $S'$, that for each sequent $\vdash\oc_{\e''}A'', \wn_{\f''} B''$, that $\e\lequ\f''$:
\begin{itemize}
\item The initialization comes from the condition of application of $(\ocu)$ on $S$.

\item The heredity comes from the condition of application of $\ocu$ on the sequent $\vdash\oc_{\e''}A'', \wn_{\f''} B''$ and from lemma~\refleqTrans.
\end{itemize}
Finally, as $S'$ is linked by the cut-formula $\oc_{\e'} A'$ to one of these sequents, we get that $\e\lequ\e'$. By the condition of application of $(\ocf)$ on $S'$, we get that $\e'\leqf\vec{\f'}$, and from lemma~\reflequs, we have that $\e\leqf\vec{\f'}$.

\item We prove, for the remaining tree (which is rooted in $S'$), that for each sequents $\vdash \oc_{\e''} A'', \wn_{\vec{\f''}} \B''$, that $\e\leqf\f''$. We prove it by induction.
\begin{itemize}
\item Initialization was done at last point.

\item For heredity, if the sequent $\vdash \oc_{\e''} A'', \wn_{\vec{\f''}} \B''$ is the conclusion of an $(\ocu)$-rule, by induction hypothesis, we get that $\e\leqf\e''$, and by $(\ocu)$ application condition we get that $\e''\lequ\vec{\f''}$, we get $\e\leqf\vec{\f''}$ with axiom~\refleqfu.

\item For heredity, if the sequent $\vdash \oc_{\e''} A'', \wn_{\vec{\f''}} \B''$ is the conclusion of an $(\ocf)$-rule, by induction hypothesis, we get that $\e\leqf\e''$, and by $(\ocf)$ application condition we get that $\e''\leqf\vec{\f''}$, we get $\e\leqf\vec{\f''}$ with axiom~\refleqTrans.

\item For heredity, if the sequent $\vdash \oc_{\e''} A'', \wn_{\vec{\f''}} \B''$ is the conclusion of an $(\ocg)$-rule, then $\B''$ is empty and the proposition is easily satisfied.
\end{itemize}
\end{enumerate}
We conclude by induction and we use the inequalities from it to prove that $\e\leqf\vec{\g}$.
\end{proof}

\subsubsection{Justification for step~{\hyperref[ocucommocgNonEmptyocfFirst]{$({\text{comm}}^3_\ocu)$}}: proof of \Cref{ocugreadinessconditionfcase}}
\label{app:ocugreadinessconditionfcase}

The following lemma deals with the case where there are sequents concluded by an $(\ocg)$-rule with non-empty context and where the first rule encountered is an $\ocf$-rule.
\begin{lem}[Justification for step~{\hyperref[ocucommocgNonEmptyocfFirst]{$({\text{comm}}^3_\ocu)$}}]
Let
$$
{\scriptsize
\AIC{\pi_1}
\noLine
\UIC{\vdash A, B}
\AIC{\e\lequ\f}
\RL{\ocu}
\BIC{\vdash \oc_{\e} A, \wn_{\f} B}
\AIC{\ocuProofs_1}
\AIC{\pi_2}
\noLine
\UIC{\vdash C, \B}
\AIC{\e'\leqf\vec{\f'}}
\RL{\ocf}
\BIC{\vdash \oc_{\e'} C, \wn_{\vec{\f'}}\B}
\AIC{\ocProofs_2}
\RL{\rmcutpar}
\QIC{\vdash \oc_{\e} A, \wn_{\vec{\g}}\A}
\DP}
$$
be a $\musuperLLinf(\Sig, \leqg, \leqf, \lequ)$-proof, such that $\ocProofs_2$ contains a sequent conclusion of an $(\ocg)$ rule with non-empty context ; $\C:=\{\vdash \oc_{\e} A, \wn_{\f} B\} \cup \ocuProofs_1 \cup \{\vdash \oc_{\e'} C, \wn_{\vec{\f'}}\B\}$ are a cut-connected subset of sequents ; and $\C':=\{\vdash \oc_{\e'} C, \wn_{\vec{\f'}}\B\} \cup \ocProofs_2$ another one.
We have that
$$
{\scriptsize\AIC{\pi_1}
\noLine
\UIC{\vdash A, B}
\AIC{\mathcal{C}_1}
\AIC{\pi_2}
\noLine
\UIC{\vdash C, \B}
\AIC{\vec{\f'}(\mpx{1})}
\doubleLine
\RL{\mpx{1}}
\BIC{\vdash C, \wn_{\vec{\f'}}\B}
\AIC{\ocProofs_2}
\RL{\rmcutpar}
\QIC{\vdash A, \wn_{\vec{\g}}\A}
\AIC{\e\leqg\vec{\g}}
\RL{\ocg}
\BIC{\vdash \oc_{\e} A, \wn_{\vec{\g}}\A}
\DP}
$$
is also a $\musuperLLinf(\Sig, \leqg, \leqf, \lequ)$-proof.
\end{lem}
\begin{proof}
We do our proof in three steps:
\begin{enumerate}
\item There is a sequent $S'' := \vdash \oc_{\e''} A'',  \wn_{\vec{\f''}}\B''$, with $\B''$ being non-empty, which is conclusion of an $(\ocg)$-rule.
The $\cutrel$-relation extended to sequents defines a tree on $\C'$, therefore there is a path from the sequent $S':= \vdash \oc_{\e'} C, \wn_{\vec{\f'}}\B$ to the sequent $S''$, of sequents $\vdash \oc_{\e^{(3)}} A^{(3)}, \wn_{\vec{\f^{(3)}}}\B^{(3)}$.
Let's suppose without loss of generality, that this sequent is the closest such sequent to $S'$.
We prove by induction on this path, starting from $S'$ and stopping one sequent before $S''$ that $\e'\leqf\f^{(3)}$:
\begin{itemize}
\item The initialisation comes from the condition of application of $\ocf$ on $S'$.

\item For the heredity, we have that the sequent $\vdash \oc_{\e^{(3)}} A^{(3)}, \wn_{\vec{\f^{(3)}}}\B^{(3)}$ is cut-connected to a $\vdash \oc_{\e^{(4)}} A^{(4)}, \wn_{\vec{\f^{(4)}}}\B^{(4)}, \wn_{\e^{(3)}} ({A^{(3)}}^\perp)$ on $\oc_{\e^{(3)}} A^{(3)}$, therefore $\e'\leqf\e^{(3)}$. We have two cases: either this sequent is the conclusion of an $(\ocu)$-rule and we apply axiom~\refleqfu, either of an $(\ocf)$-rule and we apply axiom~\refleqTrans. In each case, we have that $\e'\leqf\vec{\f^{(3)}}$.
\end{itemize}
We conclude by induction and get a sequent $S^{(3)}:=\vdash \oc_{\e^{(3)}} A^{(3)}, \wn_{\vec{\f^{(3)}}}\B^{(3)}$ cut-connected to $S''$ on the formula $\oc_{\e''}A''$ with $\e'\leqf\vec{\f^{(3)}}$. From that we get that $\e'\leqf\e''$. Moreover, we have that $\e''\leqg\vec{\f''}$. As $\B''$ is non-empty, there is a signature $\g''\in\vec{\f''}$ such that $\e''\leqg\g''$. We can therefore apply axiom~\refleqfg. We get that for each signatures $\e^{(4)}$ such that $\e'\leqf\e^{(4)}$, $\e'\leqg\e^{(4)}$ and $\e^{(4)}(\mpx{1})$, which we can apply to $\e'$ and $\vec{\f'}$ to get that $\e'\leqg\vec{\f'}$ and $\vec{\f'}(\mpx{1})$.

\item Again, we notice that the $\cutrel$-relation extended to sequent defines a tree on $\C$, meaning that there is a path in this tree, from $S := \vdash \oc_\e A, \wn_\f B$ to $S'$.
We prove by induction on this path, starting from $S$ and stopping one sequent before $S'$, that for each sequent $\vdash\oc_{\e^{(3)}}A^{(3)}, \wn_{\f^{(3)}} B^{(3)}$, that $\e\lequ\f^{(3)}$:
\begin{itemize}
\item The initialization comes from the condition of application of $(\ocu)$ on $S$.

\item The heredity comes from the condition of application of $\ocu$ on the sequent $\vdash\oc_{\e^{(3)}}A^{(3)}, \wn_{\f^{(3)}} B^{(3)}$ and from lemma~\refleqTrans.
\end{itemize}
Finally, as $S'$ is linked by the cut-formula $\oc_{\e'} A'$ to one of these sequents, we get that $\e\lequ\e'$.

\item Finally, we prove that for each sequents $\vdash \oc_{\e^{(3)}}A^{(3)}, \wn_{\f^{(3)}}\B^{(3)}$ of $\C'$, $\e\leqg\f^{(3)}$. We prove it by induction as $\C'$ is a tree with the $\cutrel$-relation.
\begin{itemize}
\item Initialization comes from the face that $\e\lequ\e'$, $\e'\leqg\vec{\f'}$ and axiom~\reflequs.

\item For heredity,  we have that there is a sequent $\vdash \oc_{\e^{(4)}} A^{(4)}, \wn_{\vec{\f^{(4)}}}\B^{(4)}, \wn_{\e^{(3)}} (A^{(3)})^\bot$ of $\C'$, connected on $\oc_{\e^{(3)}} A^{(3)}$ to our sequent.
By induction hypothesis, we have that $\e\leqg\e^{(3)}$.
The rule on top of $\vdash \oc_{\e^{(3)}} A^{(3)}, \wn_{\vec{\f^{(3)}}}\B^{(3)}$ is a promotion. We have two cases:
\begin{itemize}
\item If it's a $(\ocg)$-promotion, we can use axiom~\refleqTrans{} and with the application condition of the promotion, we get that $\e\leqg\vec{\f^{(3)}}$.
\item If it's an $(\ocf)$-promotion or an $(\ocu)$-promotion, we can use axiom~\refleqgs and with the application condition of the promotion, we get that $\e\leqg\vec{\f^{(3)}}$.
\end{itemize}
\end{itemize}
We conclude by induction.
\end{enumerate}
We got two important properties:
\begin{enumerate}
\item For each sequent $\vdash \oc_{\e^{(3)}} A^{(3)}, \wn_{\vec{\f^{(3)}}}\B^{(3)}$ of $\C'$, we have that $\e\leqg\vec{\f^{(3)}}$.
\item We have $\vec{\f'}(\mpx{1})$. 
\end{enumerate}
We conclude using inequalities of the first property to find that $\e\leqg\vec{\g}$. And we use the second property for the $(\mpx{1})$-rule.
\end{proof}

\subsubsection{Justification for step~{\hyperref[ocucommocgNonEmptyocgFirst]{$({\text{comm}}^4_\ocu)$}}: proof of \Cref{ocugreadinessconditiongcase}}
\label{app:ocugreadinessconditiongcase}

The last lemma of promotion commutation is about the case where we commute an $(\ocu)$-promotion but when first meeting an $(\ocg)$-promotion.
\begin{lem}[Justification for step~{\hyperref[ocucommocgNonEmptyocgFirst]{$({\text{comm}}^4_\ocu)$}}]
Let
$$
\hspace{-1cm}
{\scriptsize\AIC{\pi_1}
\noLine
\UIC{\vdash A, B}
\AIC{\e\lequ\f}
\RL{\ocu}
\BIC{\vdash \oc_{\e} A, \wn_{\f} B}
\AIC{\ocuProofs_1}
\AIC{\pi_2}
\noLine
\UIC{\vdash C, \wn_{\vec{\f'}}\B}
\AIC{\e'\leqg\vec{\f'}}
\RL{\ocg}
\BIC{\vdash \oc_{\e'} C, \wn_{\vec{\f'}}\B}
\AIC{\ocProofs_2}
\RL{\rmcutpar}
\QIC{\vdash \oc_{\e} A, \wn_{\vec{\g}}\A}
\DP}
$$
be a $\musuperLLinf(\Sig, \leqg, \leqf, \lequ)$-proof such that $\C:=\{\vdash \oc_{\e} A, \wn_{\f} B\} \cup \ocuProofs_1 \cup \{\vdash \oc_{\e'} C, \wn_{\vec{\f'}}\B\}$ are a cut-connected subset of sequents ; and $\C':=\{\vdash \oc_{\e'} C, \wn_{\vec{\f'}}\B\} \cup \ocProofs_2$ another one. Then,
$$
\AIC{\pi_1}
\noLine
\UIC{\vdash A, B}
\AIC{\C_1}
\AIC{\pi_2}
\noLine
\UIC{\vdash C, \wn_{\vec{\f'}}\B}
\AIC{\ocProofs_2}
\RL{\rmcutpar}
\QIC{\vdash A, \wn_{\vec{\g}}\A}
\AIC{\e\leqg\vec{\g}}
\RL{\ocg}
\BIC{\vdash \oc_{\e} A, \wn_{\vec{\g}}\A}
\DP
$$
is also a $\musuperLLinf(\Sig, \leqg, \leqf, \lequ)$-proof.
\end{lem}
\begin{proof}
We do our proof in two steps:
\begin{enumerate}
\item First, we prove that for each sequents $\vdash \oc_{\e''} A, \wn_{\f''}B$ of $\C\setminus\{\vdash\oc_{\e'}C, \wn_{\vec{\f'}}\B\}$ that $\e\lequ\f''$. We prove it by induction on this list starting with the sequent $S := \vdash\oc_\e A, \wn_{\vec{\f}} B$ (it is a list with the $\cutrel$-relation):
\begin{itemize}
\item Initialization comes from the condition of application of $(\ocu)$ on $S$.
\item Heredity comes from the condition of application of $(\ocu)$ on the concerned sequent, from induction hypothesis and from axiom~\refleqTrans.
\end{itemize}
We conclude by induction and deduce from the obtained property that $\e\lequ\e'$.

\item We then prove that for each sequents $\vdash \oc_{\e''} A, \wn_{\f''}\B$ of $\C'$, $\e\leqg\vec{\f''}$. We prove it by induction on $\C'$ as the $\cutrel$-relation defines a tree on it, for which we take $S' := \oc_{\e'}C, \wn_{\vec{\f'}}\B$ as the root.
\begin{itemize}
\item The initialization comes from $\e\lequ\e'$ that we showed for first step, from $\e'\leqg\vec{\f'}$ which is the condition of application of $(\ocg)$ on $S'$ and from axiom~\reflequs.

\item For heredity,  we have that there is a sequent\linebreak $\vdash \oc_{\e^{(3)}} A^{(3)}, \wn_{\vec{\f^{(3)}}}\B^{(3)}, \wn_{\e''} ({A''}^\perp)$ of $\C'$, connected on $\oc_{\e''} A''$ to our sequent.
By induction hypothesis, we have that $\e\leqg\e''$.
The rule on top of $\vdash \oc_{\e''} A'', \wn_{\vec{\f''}}\B''$ is a promotion. We have two cases:
\begin{itemize}
\item If it's a $(\ocg)$-promotion, we can use axiom~\refleqTrans{} and with the application condition of the promotion, we get that $\e\leqg\vec{\f''}$.
\item If it's an $(\ocf)$-promotion or an $(\ocu)$-promotion, we can use axiom~\refleqgs{} and with the application condition of the promotion, we get that $\e\leqg\vec{\f''}$.
\end{itemize}
\end{itemize}
We conclude by induction
\end{enumerate}
From the inequalities that we get from induction, we can easily prove that $\e\leqg\vec{\g}$.
\end{proof}

\subsubsection{Justification for step~{\hyperref[contrPrincip]{$(\text{principal}_{\contr{}})$}}: proof of \Cref{contrPrincipLemma}}
\label{app:contrPrincipLemma}

Then we have the principal cases, starting with the contraction:

\begin{lem}[Justification for step~{\hyperref[contrPrincip]{$(\text{principal}_{\contr{}})$}}]
If
$$
\AIC{\tikzmark{app:contrPrincipLemma12}\C_\B}
\AIC{\pi}
\noLine
\UIC{\vdash \overbrace{\wn_\e A, \dots, \wn_\e A}^i, \B}
\AIC{\e(\contr{i})}
\RL{\contr{i}}
\BIC{\vdash \wn_\e A, \tikzmark{app:contrPrincipLemma12p}\B}
\AIC{\tikzmark{app:contrPrincipLemma22}\ocProofs_{\wn_\e A}}
\RL{\rmcutpar}
\TIC{\vdash \tikzmark{app:contrPrincipLemma11}\A, \wn_{\vec{\g}}\A'\tikzmark{app:contrPrincipLemma21}}
\DP
$$
\begin{tikzpicture}[overlay,remember picture,-,line cap=round,line width=0.1cm]
   \draw[rounded corners, smooth=2,green, opacity=.4] ([xshift=-2mm] pic cs:app:contrPrincipLemma21) to ([xshift=2mm, yshift=2mm] pic cs:app:contrPrincipLemma22);
   \draw[rounded corners, smooth=2,red, opacity=.4] (pic cs:app:contrPrincipLemma11) to ([xshift=2mm, yshift=2mm] pic cs:app:contrPrincipLemma12);
   \draw[rounded corners, smooth=2,red, opacity=.4] (pic cs:app:contrPrincipLemma11) to ([xshift=2mm, yshift=2mm] pic cs:app:contrPrincipLemma12p);
\end{tikzpicture}
is a $\musuperLLinf(\Sig, \leqg, \leqf, \lequ)$-proof, then
$$
{\scriptsize\AIC{\tikzmark{app:contrPrincipLemmapostred12}\C_\B}
\AIC{\pi}
\noLine
\UIC{\vdash \overbrace{\wn_\e A, \dots, \wn_\e A}^i, \B\tikzmark{app:contrPrincipLemmapostred12p}}
\AIC{\tikzmark{app:contrPrincipLemmapostred22}\overbrace{\ocProofs_{\wn_\e A} \quad\dots\quad \ocProofs_{\wn_\e A}}^i\tikzmark{app:contrPrincipLemmapostred32}}
\RL{\mathsf{mcut(\iota', \perp\!\!\!\perp')}}
\TIC{\tikzmark{app:contrPrincipLemmapostred11}\A, \tikzmark{app:contrPrincipLemmapostred21}\wn_{\vec{\g}}\A', \dots, \wn_{\vec{\g}}\A'\tikzmark{app:contrPrincipLemmapostred31}}
\AIC{\bar{\vec{\g}}(\contr{i})}
\doubleLine
\RL{\contr{i}^{\bar{\vec{\g}}}}
\BIC{\vdash \A, \wn_{\vec{\g}}\A'}
\DP}
$$
\begin{tikzpicture}[overlay,remember picture,-,line cap=round,line width=0.1cm]
   \draw[rounded corners, smooth=2,green, opacity=.4] ([xshift=0mm] pic cs:app:contrPrincipLemmapostred21) to ([xshift=2mm, yshift=2mm] pic cs:app:contrPrincipLemmapostred22);
   \draw[rounded corners, smooth=2,green, opacity=.4] ([xshift=-2mm] pic cs:app:contrPrincipLemmapostred31) to ([xshift=-2mm, yshift=2mm] pic cs:app:contrPrincipLemmapostred32);
   \draw[rounded corners, smooth=2,red, opacity=.4] (pic cs:app:contrPrincipLemmapostred11) to ([xshift=2mm, yshift=2mm] pic cs:app:contrPrincipLemmapostred12);
   \draw[rounded corners, smooth=2,red, opacity=.4] (pic cs:app:contrPrincipLemmapostred11) to ([xshift=-1mm, yshift=1mm] pic cs:app:contrPrincipLemmapostred12p);
\end{tikzpicture}
is also a $\musuperLLinf(\Sig, \leqg, \leqf, \lequ)$-proof.
\end{lem}
\begin{proof}
We prove for each sequent $\vdash \oc_{\e''} A'', \wn_{\vec{\f''}}\B''\in\C_{\wn_\e A}^\oc$, we have that $\e\leq_s \vec{\f''}$ (for one $s\in\{g, f, u\}$. As the relation $\cutrel$ defines a tree on $\C':\ocProofs_{\wn_\e A}$ (rooted on the sequent $S := \vdash\oc_\e A, \wn_{\vec{\f'}}\B'$ which is the sequent connected to $\vdash\wn_\e A, \B$ on $\wn_\e A$), we do a proof by induction on this tree:
\begin{itemize}
\item Initialization comes from the application condition of the promotion.
\item For heredity, we get from induction hypothesis that $\e\leq_s\e''$ for a $s\in\{g, f, u\}$, from the condition of application of the promotion, we get that $\e''\leq_{s'}\vec{\f''}$ (again for a $s'\in\{g, f, u\}$), depending on the cases, from axioms \refleqTrans, \refleqgs, \refleqfu, \refleqfg, \reflequs, we get that $\e\leq_{s''}\vec{\f''}$ for a $s''\in\{g, f, u\}$.
\end{itemize}
We conclude by induction, we get using the obtained property, the fact that $\e(\contr{i})$ and from axiom~\refcontrAx, that for each sequent $\vdash \oc_{\e''} A'', \wn_{\vec{\f''}}\B''\in\C_{\wn_\e A}^\oc$, $\bar{\vec{\f''}}(\contr{i})$. We use property~\ref{prop:closureDerivabilityDerivationequivalence} to get that $\bar{\vec{\g}}(\contr{i})$ is true, making the derivation valid in the proof of the statement.
\end{proof}

\subsubsection{Justification for step~{\hyperref[mpxcomm]{$(\text{comm}_{\mpx{}})$}}: proof of \Cref{mpxPrincipLemma}}
    \label{app:mpxPrincipLemma}

Before justifying the case for the multiplexing principal reduction, we recall \Cref{defi:mpxPrincipompx} together 
with a graphical representation to make it more understandable:

\begin{defi}[$\opmpx_{S^\oc}(\ocProofs)$ contexts]\label{app:mpxPrincipompx}
Let $\pi$ be some  $\musuperLLinf(\Sig, \leqg, \leqf, \lequ)$-proof concluded in a $\mcut(\iota,\cutrel)$ inference,
$\ocProofs$ a context of the multicut which is a tree with respect to a cut-relation $\cutrel$ and
$S^\oc$ be a sequent of $\ocProofs$ that we shall consider as the root of the tree. 

We define a $\musuperLLinf(\Sig, \leqg, \leqf, \lequ)$-context $\opmpx_{S^\oc}(\ocProofs)$ altogether with two sets of sequents, $\mathcal{S}^{\mpx{}}_{\C^\oc, S^\oc}$ and $\mathcal{S}^{\contr{}}_{\C^\oc, S^\oc}$, by induction on the tree ordering on $\ocProofs$:

Let $\ocProofs_1, \dots, \ocProofs_n$ be the sons of $S^\oc$, such that $\ocProofs = (S^\oc, (\ocProofs_1, \dots, \ocProofs_n))$, we have two cases:
\begin{itemize}
\item $S^\oc = S^\ocg$, then we define $\opmpx_{S^\oc}(\ocProofs):= (S, (\ocProofs_1, \dots, \ocProofs_n))$ ; $ \mathcal{S}^{\mpx{}}_{\ocProofs, S^\oc}:=\emptyset$ ; $ \mathcal{S}^{\contr{}}_{\ocProofs, S^\oc} := \ocProofs$.

\item $S^\oc = S^\ocf$ or $S^\oc = S^\ocu$, then let the root of $\ocProofs_i$ be $S_i^\oc$, we define $\opmpx_{S^\oc}(\ocProofs)$ as\\ $(S, \opmpx_{S_1^\oc}(\ocProofs_1), \dots, \opmpx_{S_n^\oc}(\ocProofs_n))$, $\mathcal{S}^{\mpx{}}_{\ocProofs, S^\oc}$ as $\{S^\oc\}\cup \bigcup \mathcal{S}^{\mpx{}}_{\ocProofs_i, S_i^\oc}$ and $\mathcal{S}^{\contr{}}_{\ocProofs, S^\oc}$ as $\bigcup \mathcal{S}^{\contr{}}_{\ocProofs_i, S_i^\oc}$.
\end{itemize}

Below is a graphical picture of the above definition in the second case ($S^\oc = S^\ocf$ or $S^\oc = S^\ocu$) when all its sons (for the tree relation induced by  $\cutrel$) are of the form $S_i^\ocg$ (which illustrates both cases of the definition in one picture) :
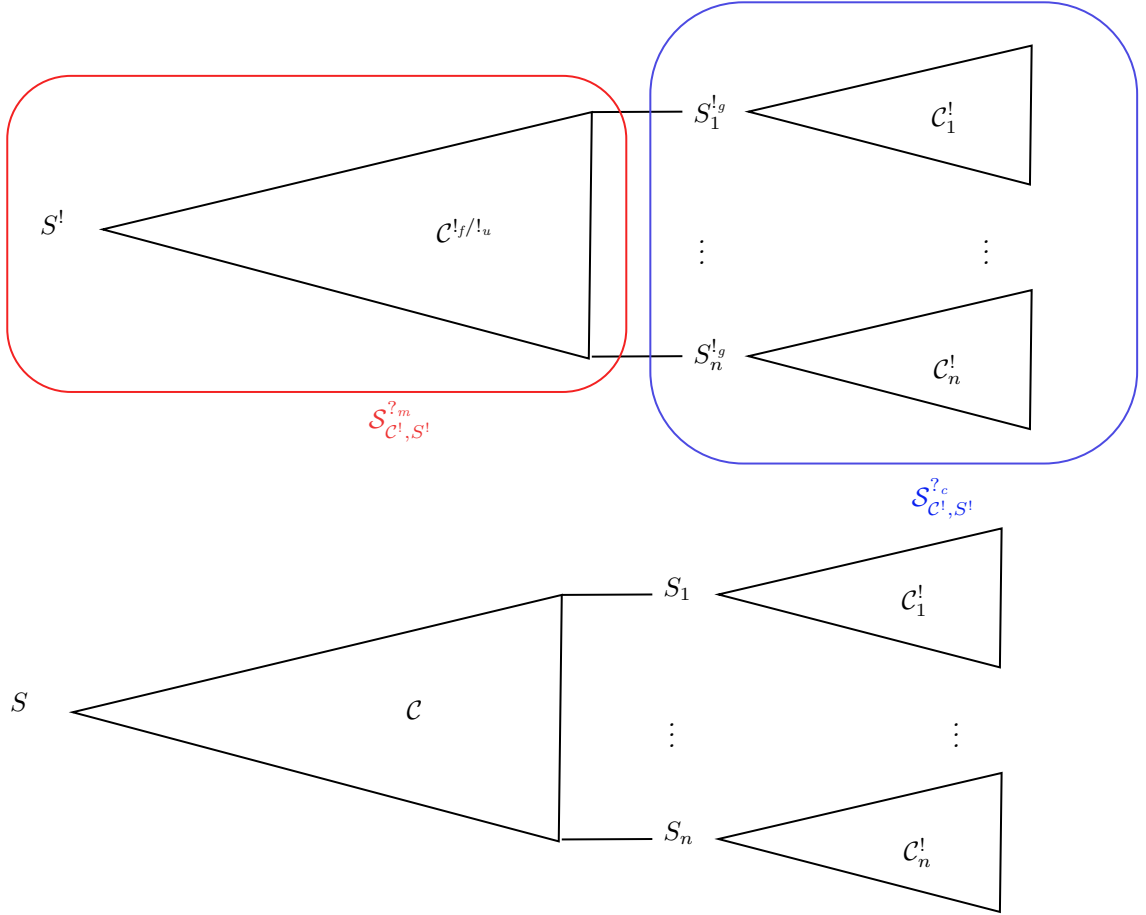
\tikzset{every picture/.style={line width=0.75pt}} 

\begin{tikzpicture}[x=0.75pt,y=0.75pt,yscale=-1,xscale=1]

\draw   (81.67,137.23) -- (324.28,202.18) -- (325.79,78.22) -- cycle ;
\draw    (325.79,78.22) -- (371,78) ;
\draw   (404,78) -- (544.49,114.67) -- (545.34,44.77) -- cycle ;
\draw    (325.79,201.22) -- (371,201) ;
\draw   (404,201) -- (544.49,237.67) -- (545.34,167.77) -- cycle ;
\draw  [color={rgb, 255:red, 243; green, 34; blue, 34 }  ,draw opacity=1 ] (34,91.8) .. controls (34,74.24) and (48.24,60) .. (65.8,60) -- (311.2,60) .. controls (328.76,60) and (343,74.24) .. (343,91.8) -- (343,187.2) .. controls (343,204.76) and (328.76,219) .. (311.2,219) -- (65.8,219) .. controls (48.24,219) and (34,204.76) .. (34,187.2) -- cycle ;
\draw  [color={rgb, 255:red, 76; green, 74; blue, 226 }  ,draw opacity=1 ] (598,69.4) .. controls (598,43.77) and (577.23,23) .. (551.6,23) -- (401.4,23) .. controls (375.77,23) and (355,43.77) .. (355,69.4) -- (355,208.6) .. controls (355,234.23) and (375.77,255) .. (401.4,255) -- (551.6,255) .. controls (577.23,255) and (598,234.23) .. (598,208.6) -- cycle ;

\draw (49,126) node [anchor=north west][inner sep=0.75pt]   [align=left] {$\displaystyle S^{\oc}$};
\draw (247,130) node [anchor=north west][inner sep=0.75pt]   [align=left] {$\C^{\ocf/\ocu}$};
\draw (375,68) node [anchor=north west][inner sep=0.75pt]   [align=left] {$\displaystyle S_{1}^{\ocg}$};
\draw (375,191) node [anchor=north west][inner sep=0.75pt]   [align=left] {$\displaystyle S_{n}^{\ocg}$};
\draw (494,72) node [anchor=north west][inner sep=0.75pt]   [align=left] {$\displaystyle \ocProofs_1 $};
\draw (495,199) node [anchor=north west][inner sep=0.75pt]   [align=left] {$\ocProofs_n$};
\draw (519,133) node [anchor=north west][inner sep=0.75pt]   [align=left] {$\displaystyle \vdots $};
\draw (377,133) node [anchor=north west][inner sep=0.75pt]   [align=left] {$\displaystyle \vdots $};
\draw (213,222) node [anchor=north west][inner sep=0.75pt]  [color={rgb, 255:red, 238; green, 52; blue, 52 }  ,opacity=1 ] [align=left] {$\mathcal{S}^{\mpx{}}_{\mathcal{C}^\oc, S^\oc}$};
\draw (484,261) node [anchor=north west][inner sep=0.75pt]   [align=left] {\textcolor[rgb]{0.04,0.13,0.94}{{$\mathcal{S}^{\contr{}}_{\mathcal{C}^\oc, S^\oc}$}}};

\end{tikzpicture}

\tikzset{every picture/.style={line width=0.75pt}} 

\begin{tikzpicture}[x=0.75pt,y=0.75pt,yscale=-1,xscale=1]

\draw   (81.67,137.23) -- (324.28,202.18) -- (325.79,78.22) -- cycle ;
\draw    (325.79,78.22) -- (371,78) ;
\draw   (404,78) -- (544.49,114.67) -- (545.34,44.77) -- cycle ;
\draw    (325.79,201.22) -- (371,201) ;
\draw   (404,201) -- (544.49,237.67) -- (545.34,167.77) -- cycle ;

\draw (49,126) node [anchor=north west][inner sep=0.75pt]   [align=left] {$\displaystyle S$};
\draw (247,130) node [anchor=north west][inner sep=0.75pt]   [align=left] {$\C$};
\draw (375,68) node [anchor=north west][inner sep=0.75pt]   [align=left] {$\displaystyle S_{1}$};
\draw (375,191) node [anchor=north west][inner sep=0.75pt]   [align=left] {$\displaystyle S_{n}$};
\draw (494,72) node [anchor=north west][inner sep=0.75pt]   [align=left] {$\displaystyle \ocProofs_1 $};
\draw (495,199) node [anchor=north west][inner sep=0.75pt]   [align=left] {$\ocProofs_n$};
\draw (519,133) node [anchor=north west][inner sep=0.75pt]   [align=left] {$\displaystyle \vdots $};
\draw (377,133) node [anchor=north west][inner sep=0.75pt]   [align=left] {$\displaystyle \vdots $};

\end{tikzpicture}

\end{defi}

Finally, we have the multiplexing principal case:
\begin{lem}[Justification for step~{\hyperref[mpxcomm]{$(\text{comm}_{\mpx{}})$}}]
Let
$$
\AIC{\C_\B}
\AIC{\vdash \overbrace{A, \dots, A}^i, \B}
\AIC{\e(\mpx{i})}
\RL{\mpx{i}}
\BIC{\vdash \wn_\e A, \B}
\AIC{\ocProofs_{\wn_\e A}}
\RL{\rmcutpar}
\TIC{\vdash \A, \wn_{\g'}\A', \wn_{\g''}\A''}
\DP$$
be a $\musuperLLinf(\Sig, \leqg, \leqf, \lequ)$-proof with $\A$ being sent on $\C_{\B}\cup{\B}$ by $\iota$ ; $\wn_{\vec{\g''}} \A''$ being sent on sequent of $\mathcal{S}^{\mpx{}}_{\ocProofs, S^\oc}$ ; and $\wn_{\vec{\g'}}\A'$ being sent on $\mathcal{S}^{\contr{}}_{\ocProofs, S^\oc}$, where $S^\oc :=\oc_\e A, \wn_{\vec{\f'}}\B'$ is the sequent cut-connected to $\vdash \wn_\e A, \B$ on the formula $\wn_\e A$.
We have that
$$
\hspace{-0.7cm}
{\scriptsize\AIC{\C_\B}
\AIC{\vdash \overbrace{A, \dots, A}^i, \B}
\AIC{\overbrace{\opmpx_{S^\oc}(\ocProofs_{\wn_\e A}) \quad \dots \quad \opmpx_{S^\oc}(\ocProofs_{\wn_\e A})}^i}
\RL{\mathsf{mcut(\iota', \perp\!\!\!\perp')}}
\TIC{\vdash\A,\overbrace{\A', \dots, \A'}^i, \overbrace{\wn_{\vec{\g''}}\A'', \dots, \wn_{\vec{\g''}}\A''}^i}
\AIC{\hspace{-1.5cm}\bar{\vec{\g'}}(\mpx{i})}
\doubleLine
\RL{\mpx{i}^{\bar{\vec{\g'}}}}
\BIC{\vdash \A, \wn_{\vec{\g'}}\A', \overbrace{\wn_{\vec{\g''}}\A'', \dots, \wn_{\vec{\g''}}\A''}^i}
\AIC{\hspace{-1cm}\bar{\vec{\g''}}(\contr{i})}
\doubleLine
\RL{\contr{i}^{\bar{\vec{\g''}}}}
\BIC{\vdash \A, \wn_{\vec{\g'}}\A', \wn_{\vec{\g''}}\A''}
\DP}
$$
is also a $\musuperLLinf(\Sig, \leqg, \leqf, \lequ)$-proof.
\end{lem}
\begin{proof}
We prove that for each sequent $\vdash \oc_{\e''} A'', \wn_{\vec{\f''}} \B''$ of $\mathcal{S}^{\contr{}}_{\ocProofs, S^\oc}$, $\e\leqg\vec{\f''}$ and that for each sequent $\vdash \oc_{\e''} A'', \wn_{\vec{\f''}} \B''$ of $\mathcal{S}^{\mpx{}}_{\ocProofs, S^\oc}$, $\e\leqf\vec{\f''}$ or $\e\lequ\vec{\f''}$. The $\cutrel$-relation defines a tree rooted on $\S^\oc$, we do a proof by induction:
\begin{itemize}
\item If $\vdash \oc_{\e''} A'', \wn_{\vec{\f''}}\B''$ is in $\mathcal{S}^{\mpx{}}_{\ocProofs, S^\oc}$, then its antecedent is also in $\mathcal{S}^{\mpx{}}_{\ocProofs, S^\oc}$, by induction, we have the $\e\leqf \e''$ or $\e\lequ\e''$. Moreover, the promotion applied on $\vdash \oc_{\e''} A'', \wn_{\vec{f''}}\B''$ is an $\ocf$ or an $\ocu$ promotion. We therefore have either by axiom \reflequs{}, either by axiom \refleqTrans, either by axiom \refleqfu, that $\e\leqf \vec{\f''}$ or $\e\lequ\vec{\f''}$.
\item If $\vdash \oc_{\e''} A'', \wn_{\vec{\f''}}\B''$ is in $\mathcal{S}^{\contr{}}_{\ocProofs, S^\oc}$, and that its antecedent is in $\mathcal{S}^{\mpx{}}_{\ocProofs, S^\oc}$, then by induction, we have that $\e\leqf\e''$ or $\e\leqf\e''$. Moreover, the promotion applied on $\vdash \oc_{\e''} A'', \wn_{\vec{f''}}\B''$ is an $\ocg$ promotion. Therefore, we have by axiom \reflequs{} or \refleqfg{} that $\e\leqg\vec{\f''}$.
\item If $\vdash \oc_{\e''} A'', \wn_{\vec{\f''}}\B''$ is in $\mathcal{S}^{\contr{}}_{\ocProofs, S^\oc}$, and that its antecedent is in $\mathcal{S}^{\contr{}}_{\ocProofs, S^\oc}$, then by induction, we have that $\e\leqg\e''$. Therefore, by axiom \refleqgs{}, $\e\leqg\vec{\f''}$.
\end{itemize}

Finally we get that for all sequents $\vdash \oc_{\e''} A, \wn_{\vec{\f''}}\B''$ of $\mathcal{S}^{\mpx{}}_{\ocProofs, S^{'\oc}}$, $\bar{\vec{\f''}}(\mpx{i})$ are true, as $\e\leq_s\vec{\f''}$, $\mpx{i}(\e)$ ($s\in\{f, u\}$) and by lemma~\reffumpxAx.
We also get that for all sequents $\vdash \oc_{\e''} A, \wn_{\vec{\f''}}\B$ of $\mathcal{S}^{\contr{}}_{\ocProofs, S^{'\oc}}$, $\bar{\vec{\f''}}(\contr{i})$ are true as $\e\leqg\vec{\f''}$, $\contr{i}(\e)$ and by lemma~\refgmpxAx.

From the condition on the proof of the statement and from property~\ref{prop:closureDerivabilityDerivationequivalence}, we get that $\bar{\vec{g'}}(\mpx{i})$ and $\bar{\vec{g''}}(\contr{i})$ are true and so that the right proof is correct.
\end{proof}

\subsection{Rule permutations}
\label{app:rulePerm}
\begin{defi}[Permutation of rules]
We define \emph{one-step rule permutation} on (pre-)proofs of \muLLinf{} with rules of figure~\ref{fig:oneStepRulePermutation}.

Given a \muLLinf{} (pre-)proof $\pi$ and $p\in\{l,r,i\}^*$ a path in the proof, we define $\perm(\pi, p)$ by induction on $p$:
\begin{itemize}
    \item the proof $\perm(\pi, \epsilon)$ is the proof obtained by applying the one-step rule permutation at the root of $\pi$ if it is possible, either it is not defined;
    \item we define $\perm(q(\pi'), i\cdot p') := r(\perm(\pi', q'))$ if $\perm(\pi', q')$ is defined, otherwise it is not defined;
    \item we define$\perm(q(\pi_l, \pi_r), l\cdot q') := q(\perm(\pi_l, q'), \pi_r)$ if $\perm(\pi_l, q')$ is defined, otherwise it is not defined;
    \item we define$\perm(q(\pi_l, \pi_r), r\cdot q') := q(\perm(\pi_l, q'), \pi_r)$ if $\perm(\pi_l, q')$ is defined, otherwise it is not defined;
    \item for each other cases, $\perm(\pi, p)$ is not defined.
\end{itemize}

A \emph{sequence of rule permutation} starting from a \muLLinf{} pre-proof $\pi$ is a (possibly empty) sequence $(p_i)_{i\in \lambda}$ ($\lambda\in \omega$), where $p_i\in\{l, r, i\}$ such that if we set $\pi_0:=\pi$, then the sequence $(\pi_i)_{i\in 1+\lambda}$ defined by induction by $\pi_{i+1}:=\perm(\pi_i, p_i)$ are all defined. We say that the sequence $(\pi_i)_{i\in 1+\lambda}$ is the \emph{sequence of proofs associated to the sequence of rule permutation}.
We say that the sequence \emph{ends on} $\pi_{\lambda}$ if $\lambda$ is finite, we also write it $\perm(\pi, (p_i)_{i\in\lambda})$.
\end{defi}

\begin{figure*}
    \centering
$
\hspace*{-2cm}
\AIC{\pi}
\noLine
\UIC{\vdash \wn A, \wn A, \wn B, \wn B, \A}
\RL{\wncontr}
\UIC{\vdash \wn A, \wn B, \wn B, \A}
\RL\wncontr
\UIC{\vdash \wn A, \wn B, \A}
\DP\rightsquigarrow
\AIC{\pi}
\noLine
\UIC{\vdash \wn A, \wn A, \wn B, \wn B, \A}
\RL{\wncontr}
\UIC{\vdash \wn A, \wn A, \wn B, \A}
\RL\wncontr
\UIC{\vdash \wn A, \wn B, \A}
\DP\qquad
\AIC{\pi}
\noLine
\UIC{\vdash \wn A, \wn A, B, \A}
\RL{\wncontr}
\UIC{\vdash \wn A, B, \A}
\RL\wnde
\UIC{\vdash \wn A, \wn B, \A}
\DP\leftrightsquigarrow
\AIC{\pi}
\noLine
\UIC{\vdash \wn A, \wn A, B, \A}
\RL{\wnde}
\UIC{\vdash \wn A, \wn A, \wn B, \A}
\RL\wncontr
\UIC{\vdash \wn A, \wn B, \A}
\DP
$\\[2ex]
$
\AIC{\pi}
\noLine
\UIC{\vdash \wn A, \wn A, \A}
\RL{\wncontr}
\UIC{\vdash \wn A, \A}
\RL\wnwk
\UIC{\vdash \wn A, \wn B, \A}
\DP\leftrightsquigarrow
\AIC{\pi}
\noLine
\UIC{\vdash \wn A, \wn A, \A}
\RL{\wnwk}
\UIC{\vdash \wn A, \wn A, \wn B, \A}
\RL\wncontr
\UIC{\vdash \wn A, \wn B, \A}
\DP\qquad
\AIC{\pi}
\noLine
\UIC{\vdash \A}
\RL{\wnwk}
\UIC{\vdash \wn A, \A}
\RL\wnwk
\UIC{\vdash \wn A, \wn B, \A}
\DP\rightsquigarrow
\AIC{\pi}
\noLine
\UIC{\vdash \A}
\RL{\wnwk}
\UIC{\vdash \wn B, \A}
\RL\wnwk
\UIC{\vdash \wn A, \wn B, \A}
\DP
$\\[2ex]
$
\AIC{\pi}
\noLine
\UIC{\vdash A, \A}
\RL{\wnwk}
\UIC{\vdash A, \wn B, \A}
\RL\wnde
\UIC{\vdash \wn A, \wn B, \A}
\DP\leftrightsquigarrow
\AIC{\pi}
\noLine
\UIC{\vdash A, \A}
\RL{\wnde}
\UIC{\vdash \wn A, \A}
\RL\wnwk
\UIC{\vdash \wn A, \wn B, \A}
\DP\qquad\qquad
\AIC{\vdash A, B, \A}
\RL{\wnde}
\UIC{\vdash A, \wn B,\A}
\RL{\wnde}
\UIC{\vdash \wn A, \wn B, \A}
\DP\rightsquigarrow
\AIC{\vdash A, B, \A}
\RL{\wnde}
\UIC{\vdash \wn A, B,\A}
\RL{\wnde}
\UIC{\vdash \wn A, \wn B, \A}
\DP
$

    \caption{One-step rule permutation}
    \label{fig:oneStepRulePermutation}
\end{figure*}

\begin{lem}[Robustness of the proof structure to rule permutation]\label{lem:proofStructureRobustnessPerm}
    One-step rule permutation does not modify the structure of the proof.
\end{lem}
\begin{proof}
This lemma is immediate as the substitutions are defined between unary rule.
\end{proof}

\begin{defi}[Finiteness of permutation of rules]
Let $\pi$ be a \muLLinf{} (pre-)proof, and let $(p_i)_{i\in \lambda}$ be a sequence of rule permutation starting from $\pi$ and let $(\pi_i)_{i\in 1+\lambda}$ be the sequence of proofs associated to it, let $q\in\{l,r,i\}^*$ be a path to the conclusion sequent of a rule $(r)$ of $\pi$, we define the \emph{sequence of residuals} $(q_i)_{i\in 1+\lambda}$ of $(r)$ in $\pi_i$ to be a sequence of path defined by induction on $i$:
\begin{itemize}
    \item if $i=0$, $q_0=q$;
    \item if $p_i=q_i$, then $q_{i+1}:=q_i \cdot i$.
    \item if $q_i=p_i\cdot i$ then $q_{i+1}:=p_i$.
    \item else $q_{i+1}:=q_i$.
\end{itemize}
We say that a rule $(r)$ in $\pi$ is \emph{finitely permuted} if its sequence of residuals is ultimately constant.
We say that $(p_i)_{i\in\lambda}$ is a \emph{rule permutation sequence with finite permutation of rules} if each rule of $\pi_0$ is finitely permuted.
\end{defi}

\begin{prop}[Convergence of permutation with finite permutation of rules]
\label{prop:convergenceFinitePerm}
Let $\pi$ be a \muLLinf{} pre-proof and let $(p_i)_{i\in \omega}$ be a permutation sequence with finite permutation of rules starting from $\pi$, then the sequence is converging.
\end{prop}
\begin{proof}
    Let $(\pi_i)_{i\in\omega}$ be the sequence of proofs associated to the sequence.
Let's suppose for the sake of contradiction that the sequence is not converging. It implies, using lemma~\ref{lem:proofStructureRobustnessPerm}, that there is an infinite sequence of strictly increasing indexes $\varphi(i)$ such that the $(r_{\varphi(i)})$ are all at the same position. 
This implies by finiteness of permutation of one rules, than there are an infinite number of rules of $\pi_0$ which have $(r_{\varphi(i)})$ in their residuals, implying that one of the rules below the position of $(r_{\varphi(i)})$ in $\pi_0$ has infinitely many residuals being equal to $(r_i)$ or below $(r_i)$ contradicting the finitess of permutation of one rule hypothesis.
\end{proof}

\begin{prop}[Preservation of validity for permutations with finite permutation of rules]
\label{app:validFinitePerm}
Let $\pi$ be a \muLLinf{} pre-proof and let $(p_i)_{i\in \omega}$ be a permutation sequence with finite permutation of rules starting from $\pi$ and converging (thanks to lemma~\ref{prop:convergenceFinitePerm} to a pre-proof $\pi'$. Then $\pi$ is valid if and only if $\pi'$ is.
\end{prop}
\begin{proof}
From lemma~\ref{lem:proofStructureRobustnessPerm}, we have that the structure of the trees of the sequence stays the same, therefore the structure of $\pi$ is the same than the structure of $\pi'$, moreover the threads of $\pi$ and $\pi'$ are the same if we remove indexes where the thread is not active.
Therefore validity is easily preserved both ways.
\end{proof}

\subsection{Details on \Cref{redSeqTranslationFiniteness}}
\label{app:detailsredSeqTranslationFiniteness}

\begin{lem}
    \label{app:redSeqTranslationFiniteness}
    Let $\pi_0$ be a $\musuperLLinf(\Sig, \leqg, \leqf, \lequ)$ proof and let $\pi_0\rightsquigarrow \pi_1$ be a $\musuperLLinf(\Sig, \leqg, \leqf, \lequ)$ step of reduction. There exist a finite number of \muLLinf{} proofs $\theta_0, \dots, \theta_n$ such that $\theta_0\redseq\dots\redseq\theta_n,\quad
    \pi_0^\circ = \theta_0$ and $\theta_n = \pi_1^{\circ}$ up to a finite number of rule permutations, done only on rules that just permuted down the $(\mcut)$.
    \end{lem}
    
    To prove this lemma, we need the following one.
    This lemma prove that when starting from the translation of a proof containing derelictions promotions and functorial promotions, there exist an order of execution of cut-elimination step that will make them disappear or commute under the cut.
    This order depends on how the proof is translated, for instance the following (opened) proof:
    $$
    \AIC{\vdash A, B, C}
    \RL{\ocf}
    \UIC{\vdash \oc A, \wn B, \wn C}
    \AIC{\vdash C^\perp}
    \RL{\ocf}
    \UIC{\vdash \oc C^\perp}
    \RL{\rmcutpar}
    \BIC{\vdash \oc A, \wn B}
    \DP
    $$
    has two translations:
    $$
    {\small\AIC{\vdash A, B, C}
    \RL{\wnde}
    \UIC{\vdash A, B, \wn C}
    \RL{\wnde}
    \UIC{\vdash A, \wn B, \wn C}
    \RL{\ocprom}
    \UIC{\vdash \oc A, \wn B, \wn C}
    \AIC{\vdash C^\perp}
    \RL{\ocprom}
    \UIC{\vdash \oc C^\perp}
    \RL{\rmcutpar}
    \BIC{\vdash \oc A, \wn B}
    \DP\qquad
    \AIC{\vdash A, B, C}
    \RL{\wnde}
    \UIC{\vdash A, \wn B, C}
    \RL{\wnde}
    \UIC{\vdash A, \wn B, \wn C}
    \RL{\ocprom}
    \UIC{\vdash \oc A, \wn B, \wn C}
    \AIC{\vdash C^\perp}
    \RL{\ocprom}
    \UIC{\vdash \oc C^\perp}
    \RL{\rmcutpar}
    \BIC{\vdash \oc A, \wn B}
    \DP}
    $$
    To eliminate cuts, we apply in both the same cut-elimination steps but in a different order. We apply in both an $(\ocprom)$ commutative step, then apply in the first one a dereliction commutative step and a $(\ocprom)/(\wnde)$ principal case; whereas in the second one we first apply the $(\ocprom)/(\wnde)$ principal case then the dereliction commutative step.
    \begin{lem}
    \label{redSeqTranslationFinitenessIntermediaryLemma}
    Let $n\in\mathbb{N}$, let $d_1, \dots, d_n\in\mathbb{N}$ and let $p_1, \dots, p_n\in\{0, 1\}$. Let $\pi$ be a \muLLinf{}-proof concluded by an (\mcut)-rule, on top of which there is a list of $n$ proofs $\pi_1, \dots, \pi_n$. We ask for each $\pi_i$ to be of one of the following forms depending on $p_i$:
    \begin{itemize}
    \item If $p_i=1$, the $d_i+1$ last rules of $\pi_i$ are $d_i$ derelictions and then a promotion rule. We ask for the principal formula of this promotion to be either a formula of the conclusion, or to be cut with a formula being principal in a proof $\pi_j$ on one of the last $d_j+p_j$ rules.
    
    \item If $p_i=0$, the $d_i$ last rules of $\pi_i$  are $d_i$ derelictions.
    \end{itemize}
    In each of these two cases, we ask for $\pi_i$ that each principal formulas of the $d_i$ derelictions to be either a formula of the conclusion of the multicut, either a cut-formula being cut with a formula appearing in $\pi_j$ such that $p_j=1$.
    We prove that $\pi$ reduces through a finite number of \mcut{}-reductions to a proof where each of the last $d_i+p_i$ rules either were eliminated by a $(\ocprom/\wnde)$-principal case, or were commuted below the cut.
    \end{lem}
    \begin{proof}
        We prove the property by induction on the sum of all the $d_i$ and of all the $p_i$:
        \begin{itemize}
        \item (Initialization). As the sum of the $d_i$ and $p_i$  is $0$, all $d_i$ and $p_i$ are equal to $0$, meaning that our statement is vacuously true.
        
        \item (Heredity). We have several cases:
        \begin{itemize}
        \item If the last rule of a proof $\pi_i$ is a promotion or a dereliction for which the principal formula is in the conclusion of the (\mcut), we do a commutation step on this rule obtaining $\pi'$. We apply our induction hypothesis on the proof ending with the (\mcut); and with parameters $d'_1, \dots, d'_n$ as well as $p'_1, \dots, p'_n$ and proofs $\pi'_1, \dots, \pi'_n$. To describe these parameters we have two cases:
        \begin{itemize}
        \item If the rule is a promotion. We take for each $j\in\llbracket 1, n\rrbracket$, $d'_j=d_j$; $p'_j = p_j$ if $j\neq i$, $p'_i=0$; $\pi'_j=\pi_j$ if $j\neq i$.
        \item If the rule is a dereliction. We take for each $j\in\llbracket 1, n\rrbracket$, $d'_j=d_j$ if $j\neq i$, $d'_i=d_i-1$; $p'_j=p_j$.
        \end{itemize}
        The $\pi'_j$ will be the hypotheses of the (\mcut) of $\pi''$. Note that $\sum d'_j +\sum p'_j = \sum d_j+ \sum p_j-1$ meaning that we can apply our induction hypothesis.
        Combining our reduction step with the reduction steps of the induction hypothesis, we obtain the desired result.
        
        \item If there are no rules from the conclusion but that one $\pi_i$ ends with $d_i>0$ and $p_i=0$, meaning that the proof ends by a dereliction on a formula $\wn F$. This means that there is proof $\pi_j$ such that $p_j=1$ and such that $\wn F$ is cut with one of the formula of $\pi_j$, namely $\oc F^\perp$. As there are only one $\oc$-formula, and as $p_j=1$, $\oc F^\perp$ is the principal rule of the last rule applied on $\pi_j$. We therefore can perform an $(\ocprom/\wnde)$ principal case on the last rules from $\pi_i$ and $\pi_j$, leaving us with a proof $\pi'$ with an (\mcut) as conclusion. We apply the induction hypothesis on this proof with parameters $d'_1=d_1, \dots d'_i=d'_i-1 \dots, d'_n=d'_n$, $p'_1=p_1, \dots, p'_j=p'_j-1, \dots, p'_n=p_n$ and with the proofs being the hypotheses of the multicut.
        Combining our steps with the steps from the induction hypotheses, we obtain the desired result.
        
        \item We will show that the case where there are no rules from the conclusion and that no $\pi_i$ are such that $d_i>0$ and $p_i=0$, is impossible. Supposing, for the sake of contradiction, that this case is possible. We will construct an infinite sequence of proofs $(\theta_i)_{i\in\mathbb{N}}$ all different and all being hypotheses of the multi-cut, which is impossible. We know that there exist a proof $\theta_0:=\pi_j$ ending with a promotion on a formula $\oc A$ and that this formula is not a formula from the conclusion. This proof is in relation by the $\cutrel$-relation to another proof $\theta_1:=\pi_{j'}$. We know that this proof cannot be $\pi_j$ because the $\cutrel$-relation extended to sequents is acyclic. This proof also ends with a promotion on a principal formula which is not from the conclusion. By repeating this process, we obtain the desired sequence $(\theta_i)_{i\in\mathbb{N}}$, giving us a contradiction.
        \end{itemize}
        \end{itemize}
        The statement is therefore true by induction
        \end{proof}

    \begin{proof}[Proof of lemma~\ref{redSeqTranslationFiniteness}]
        Reductions from the non-exponential part of $\musuperLLinf(\Sig, \leqg, \leqf, \lequ)$ translates easily to one step of reduction in \muLLinf.
        To prove the result on exponential part, we will describe each translation of the reductions of figure \ref{fig:musuperllexpcommcutstep} and \ref{fig:musuperllexpprincipcutstep}.
        For the commutative steps no commutation of rules are necessary.
        \begin{itemize}
        \item Step~\hyperref[ocgcomm]{$({\text{comm}}_\ocg)$}. This step translates to the commutation of one $(\oc)$-rule in \muLLinf{}, which is one step of reduction.
        
        \item Step~\hyperref[ocfcommocgEmpty]{$({\text{comm}}^1_\ocf)$}. We prove that lemma~\ref{redSeqTranslationFinitenessIntermediaryLemma} applies to step~\hyperref[ocfcommocgEmpty]{$({\text{comm}}^1_\ocf)$}. Taking the left proof from step~\hyperref[ocfcommocgEmpty]{$({\text{comm}}^1_\ocf)$} and translating it in \muLLinf{}, we obtain a proof:
        $$
        \AIC{\pi_1^\circ}
        \noLine
        \UIC{\vdash A_1^\circ, \B_1^\circ}
        \RL{\wnde}
        \doubleLine
        \UIC{\vdash A_1^\circ, \wn\B_1^\circ}
        \RL{\ocprom}
        \UIC{\vdash \oc A_1^\circ, \wn\B_1^\circ}
        \AIC{\dots}
        \noLine
        \UIC{}
        \noLine
        \UIC{}
        \noLine
        \UIC{}
        \AIC{\pi_n^\circ}
        \noLine
        \UIC{\vdash A_n^\circ, \B_n^\circ}
        \RL{\wnde}
        \doubleLine
        \UIC{\vdash A_n^\circ, \wn\B_n^\circ}
        \RL{\ocprom}
        \UIC{\vdash \oc A_n^\circ, \wn\B_n^\circ}
        \RL{\rmcutpar}
        \TIC{\vdash \oc A^\circ, \wn\A^\circ}
        \DP
        $$
        with $\iota(1)=(i, 1)$ for some $i$ and $n=1+\#(\C)$.
        We apply our result on this proof with all the $p_i$ being equal to $1$ and with $d_i=\#(\B_i)$.
        Moreover, we notice that there will be only one promotion rule commuting under the cut and that it commutes before any dereliction, giving us 
        the 
        translation of the functorial promotion under the multicut.
        
        \item Step~\hyperref[ocfcommocgNonEmpty]{$({\text{comm}}^2_\ocf)$}. As for~\hyperref[ocgcomm]{$({\text{comm}}_\ocg)$}, this step only translates to the commutation of one $(\oc)$-rule in \muLLinf{}, which is one step of reduction.
        
        \item Step~\hyperref[ocucommOnlyocu]{$({\text{comm}}^1_\ocu)$}.
        This step translates to the commutation of one $(\ocprom)$-rule, followed by $\#(\ocuProofs)$ $(\oc/\wnde)$ principal steps and finally one $(\wnde)$ commutation giving us the translation of a unary promotion under the multicut.
        
        \item Step~\hyperref[ocucommocgEmpty]{$({\text{comm}}^2_\ocu)$}.
        We prove this step using lemma~\ref{redSeqTranslationFinitenessIntermediaryLemma} as for step~\hyperref[ocfcommocgEmpty]{$({\text{comm}}^1_\ocf)$}.
        
        \item Step~\hyperref[ocucommocgNonEmptyocfFirst]{$({\text{comm}}^3_\ocu)$} and~\hyperref[ocucommocgNonEmptyocgFirst]{$({\text{comm}}^4_\ocu)$}. Both of these steps translate to the commutation of one $(\ocprom)$, followed by $\#(\ocuProofs_1)+1$ $(\oc/\wnde)$ principal steps.
        
        \item Step~\hyperref[mpxcomm]{$(\text{comm}_{\mpx{}})$}. We must distinguish three cases based on $i$:
        \begin{itemize}
        \item $i=0$. This step translate to one $(\wnwk)$-commutative step.
        \item $i=1$. This step translate to one $(\wnde)$-commutative step.
        \item $i>1$. This step translates to $i-1$ commutation of $(\wncontr)$ and $i$ commutation of $(\wnde)$.
        \end{itemize}
        
        \item Step~\hyperref[contrcomm]{$(\text{comm}_{\contr{}})$}. This step translates to $i-1$ commutation of $(\wncontr)$.
        
        \item Step~\hyperref[contrPrincip]{$(\text{principal}_{\contr{}})$}. This step translates to  $i-1$ contraction principal cases. At the end we obtain the following derivation under the multi-cut:
        $$
        \AIC{\vdash\A^\circ, \overbrace{\wn {\A'}^\circ, \dots, \wn {\A'}^\circ}^i}
        \RL{\wncontr}
        \doubleLine
        \UIC{\vdash\A^\circ, \overbrace{\wn {\A'}^\circ, \dots, \wn {\A'}^\circ}^{i-1}}
        \noLine
        \UIC{\vdots}
        \noLine
        \UIC{\vdash\A^\circ, \wn {\A'}^\circ, \wn{\A'}^\circ}
        \RL{\wncontr}
        \doubleLine
        \UIC{\vdash \A^\circ, \wn{\A'}^\circ}
        \DP
        $$
        which we can re-arrange to get the translation of $\#\A'$ $\contr{i}^{\bar{\vec{\g}}}$ rules on each formulas of $\wn{\A'}^\circ$.
        Note that for $i=2$ no rule permutation are needed.
        
        \item Step~\hyperref[mpxPrincip]{$(\text{principal}_{\mpx{}})$}. If $i\geq 1$, this step translates in two phases:
        \begin{enumerate}
        \item First $i-1$ contraction principal cases;
        \item followed by $\#(\mathcal{S}^{\mpx{}}_{\ocProofs, S^{'\oc}})$ $(\wnde/\oc)$-principal cases, and $\#(\A'')$ dereliction commutative cases.
        \end{enumerate}
        To prove the second phase we re-use lemma~\ref{redSeqTranslationFinitenessIntermediaryLemma} as for steps \hyperref[ocucommocgEmpty]{$({\text{comm}}^2_\ocu)$} and \hyperref[ocfcommocgEmpty]{$({\text{comm}}^1_\ocf)$}.
        
        Finally, the obtained proof under the multi-cut look like this:
        $$
        \AIC{\vdash \A^\circ, \overbrace{\wn {\A''}^\circ, \dots, \wn {\A''}^\circ}^i, \overbrace{{\A'}^\circ,\dots, {\A'}^\circ}^i}
        \RL{\wnde}
        \doubleLine
        \UIC{\vdash \A^\circ, \overbrace{\wn {\A''}^\circ, \dots, \wn {\A''}^\circ}^i, \overbrace{{\A'}^\circ,\dots, {\A'}^\circ}^{i-1}, \wn{\A'}^\circ}
        \noLine
        \UIC{\vdots}
        \noLine
        \UIC{\vdash\A^\circ, \overbrace{\wn {\A''}^\circ, \dots, \wn {\A''}^\circ}^i, {\A'}^\circ, \overbrace{\wn{\A'}^\circ, \dots, \wn{\A'}^\circ}^{i-1}}
        \RL{\wnde}
        \doubleLine
        \UIC{\vdash\A^\circ, \overbrace{\wn {\A''}^\circ, \dots, \wn {\A''}^\circ}^i, \overbrace{\wn{\A'}^\circ, \dots, \wn{\A'}^\circ}^i}
        \RL{\wncontr}
        \doubleLine
        \UIC{\vdash\A^\circ, \overbrace{\wn {\A''}^\circ, \dots, \wn {\A''}^\circ}^{i-1}, \overbrace{\wn{\A'}^\circ, \dots, \wn{\A'}^\circ}^{i-1}}
        \noLine
        \UIC{\vdots}
        \noLine
        \UIC{\vdash\A^\circ, \wn {\A''}^\circ, \wn{\A''}^\circ, \wn{\A'}^\circ, \wn{\A'}^\circ}
        \RL{\wncontr}
        \doubleLine
        \UIC{\vdash \A^\circ, \wn{\A''}^\circ, \wn{\A'}^\circ}
        \DP
        $$
        \end{itemize}
        which we can re-arrange to get the translation of $\#\A'$ $\mpx{i}^{\bar{\vec{\g''}}}$, followed by the translation of $\#\A''$ $\contr{i}^{\bar{\vec{\g'}}}$.
        \item If $i=0$, this step translates to a weakening principal case, giving us the translation of $\#\A'$ $\mpx{0}^{\bar{\vec{\g''}}}$ and $\#\A''$ $\contr{0}^{\bar{\vec{\g'}}}$ with no commutation of rules necessary.
        \end{proof}

\subsection{Details on \Cref{mcutonestepCompleteness}}

\begin{lem}[Completeness of the (\mcut)-reduction system]
    \label{app:mcutonestepCompleteness}
If there is a \muLLinf{}-redex $\mathcal{R}$ sending $\pi^\circ$ to ${\pi'}^\circ$ then there exists a $\musuperLLinf(\Sig, \leqg, \leqf, \lequ)$-redex $\mathcal{R}'$ sending $\pi$ to a proof $\pi''$, such that in the translation of $\mathcal{R}'$, $\mathcal{R}$ is applied.
\end{lem}
\begin{proof}
We only prove the exponential cases, the non-exponential cases being immediate. We have several cases:
\begin{itemize}
\item If the case is the commutative step of a contraction or a dereliction or weakening $(r)$, as it is on top of a (\mcut), it necessarily means that $(r)$ comes from the translation of a multiplexing or a contraction rule $(r')$ which is also on top of an (\mcut) in $\pi$, we can take $\mathcal{R}'$ as the step commutating $(r')$ under the cut.
\item If it is a principal case again, we have that there is a contraction or a dereliction or weakening rule $(r)$ on top of a (\mcut) on a formula $\wn A$. It also means that each proofs cut-connected to $\wn A$ ends with a promotion. As $\pi^\circ$ is the translation of a $\musuperLLinf(\Sig, \leqg,\leqf, \lequ)$-proof, it means that $(r)$ is contained in the translation of a multiplexing or contraction rule $(r')$ on a formula $\wn_\e A$ on top of a (\mcut). It also means that all the proofs cut-connected for this (\mcut) to $\wn_\e A$  are translations of promotions (no other rules than a promotion in $\musuperLLinf{}(\Sig, \leqg, \leqf, \lequ)$ translates to a derivation ending with a promotion). Therefore the principal case on $(r')$ is possible, we can take $\mathcal{R}'$ as it.
\item If it is the commutative step of a promotion $(r)$, it means that all the proofs of the contexts of the (\mcut) are promotions.
Meaning that $(r)$ is contained in the translation of a promotion $(r')$ on top of (\mcut).
We also have that the context of this (\mcut) are only proof ending with a promotion for the same reasons that last point.
We therefore need to make sure that each (\mcut) with a context full of promotions are covered by the $\rightsquigarrow$-relation. Looking back at figure~\ref{fig:musuperllexpcommcutstep} together with conditions given by each corresponding lemmas, we have that:
\begin{itemize}
\item Each $(\ocg)$-commutation is covered by the first case.
\item Each $(\ocf)$-commutation is covered by the two cases that follows: the second of the two covers the case where there is an $(\ocg)$-promotion in hypotheses of the multicut with non-empty context, whereas the first one covers the case where there are no such $(\ocg)$-promotions in the hypotheses.
\item The $(\ocu)$-commutation is covered by all the remaining cases:
\begin{itemize}
\item The first one covers $(\ocu)$-commutation when the hypotheses are all concluded by an $(\ocu)$-rule.
\item $(\ocu)$-commutation with $(\ocf)$-rules and (possibly) $(\ocg)$-rule with empty context are covered by the second case.
\item $(\ocu)$-commutation with $(\ocf)$-rules and $(\ocg)$-rule with non-empty contexts is covered by the third and the fourth cases: the third case covering all the cases where the chain of $(\ocu)$ encounters a $(\ocf)$ first, the fourth one when it encounter a $(\ocg)$ first.
\item $(\ocu)$-commutation without $(\ocf)$ rules but with $(\ocg)$ with or without empty contexts is covered by last case.
\end{itemize}
\end{itemize}
\end{itemize}
\end{proof}
    
\subsection{Details on the translation of fair reduction sequences}

\begin{coro}
    \label{app:FairredSeqTranslationFiniteness}
    For every fair $\musuperLLinf(\Sig, \leqg, \leqf, \lequ)$ reduction sequences $(\pi_i)_{i\in\omega}$, there exists:
    \begin{itemize}
        \item a fair \muLLinf{} reduction sequence $(\theta_i)_{i\in\omega}$;
        \item a sequence of strictly increasing $(\varphi(i))_{i\in\omega}$ natural numbers;
        \item for each $i$, an integer $k_i$ and a finite sequence of rule permutations $(p_i^k)_{k\in\llbracket 0, k_i-1\rrbracket}$ starting from $\pi_i^\circ$ and ending $\theta_{\varphi(i)}$. For convenience in the proof, let's denote by $(\pi_i^k)_{k\in\llbracket 0, k_i\rrbracket}$ be the sequence of proofs associated to the permutation;
        \item for all $i>i'$, $p_i^k>p_i^{k'}$ if $k'\in\llbracket 0, k_{i'}-1$ and $k\geq k_{i'}$;
        \item for all $i,k$, $p_i^k$ are positions lower than the multicuts in $\pi_i^\circ$.
        \item for each $i'\geq i$ and for each $k\in\llbracket 0, k_i-1\rrbracket, p^k_{i'}=p^k_{i}$
    \end{itemize}
    \end{coro}
    \begin{proof}
    We construct the sequence by induction on the steps of reductions of $(\pi_i)_{i\in\omega}$.
    \begin{itemize}
        \item For $i=0$: we take $\theta_0=\pi_0^\circ$, $\varphi(0)=0$ and $k_0=0$.
    
        \item For $i+1$, suppose we constructed everything up to rank $i$. We use lemma~\ref{redSeqTranslationFiniteness} on the step $\pi_i \to \pi_{i+1}$
        and get a finite sequence of reduction $\theta'_0 \to \dots \to \theta'_n$, such that there is a permutation of rules $(p_1, \dots, p_m)$ ($m\in\mathbb{N}$) starting on
        $\pi_{i+1}^\circ$ and ending on $\theta'_n$ such that $p_1, \dots, p_m$ are at the depths of rules that just commuted down the multicut during the sequence $\theta'_0 \to \dots \to \theta'_n$.
    We have that $\theta'_0
    =\pi_i^\circ$, therefore $(p_i^0, \dots, p_i^{k_i-1})$ is a sequence of reduction starting from $\theta'_0$ and ending on $\theta_{\varphi(i)}$.
    As ${\theta'}_0$ and $\theta'_j$ are equal under the multicut rules of $\theta'_0$ (for each $j\in\llbracket 0, n\rrbracket$) 
    and that depths $p_i^j, j\in\llbracket 0, k_i-1\rrbracket$ are under the multicuts of $\pi_i$
    , we have that $(p_i^0, \dots, p_i^{k_i-1})$ is a sequence of rule permutation starting on proof $\theta'_j$.
    Let's denote by ${\theta'}_j^{0}, \dots, {\theta'}_j^{k_i}$ the sequence of proof associated to it.
    We have that for the same reason, ${\theta'}_j$ is equal to ${\theta'}_{j}^{k_i}$ on top of the depths of multicuts of $\theta'_j$. We therefore have that ${\theta'}_0^{k_i}, \dots, {\theta'}_{n}^{k_i}$ is an (\mcut) reduction sequence of $\muLLinf{}$ starting from $\theta_{\varphi(i)}$.
    As the two sequences of reductions $p_1, \dots, p_m$ and $p_i^0, \dots, p_i^{k_i-1}$ have disjoint sets of rules with non-empty traces,
    we have that $p_i^0, \dots, p_i^{k_i-1}, p_1, \dots, p_m$ is a sequence of rule permutation starting from $\pi'_{i+1}$ and ending on the same proof than the proof ending the sequence $p_1, \dots, p_m, p_i^0, \dots, p_i^{k_i-1}$, namely ${\theta'}_n^{k_i}$.
    By setting $\varphi(i+1):=\varphi(i)+n, \quad \theta_{\varphi(i)+j} := {\theta'}_j^{k_i}$ (for $j\in\llbracket 0, n\rrbracket$), \quad $p^j_{i+1} = p^j_i$ for $j\leq k_i-1$\quad and \quad $p^{k_i-1+j}_{i+1} = p_j$ for $j\in\llbracket 1, m\rrbracket$, we have our property.
    
    Here is a summary of the objects used in the inductive step:
    \[\begin{tikzcd}
        {\pi_i} &&&& {\pi_{i+1}} \\
        &&&& {\pi_{i+1}^\circ} \\
        \\
        {\pi_i^{\circ}=\theta'_0} & \dots & {\theta'_j} & \dots & {\theta'_n} \\
        {{\theta'}^1_0} && {{\theta'}^1_j} && {{\theta'}^1_n} \\
        \vdots && \vdots && \vdots \\
        {{\theta'}_0^{k_i} = \theta_{\varphi(i)}} & \dots & {{\theta'}_j^{k_i}} & \dots & {\theta_{\varphi(i+1)}}
        \arrow[from=1-1, to=4-1]
        \arrow[from=4-1, to=4-2]
        \arrow[from=1-5, to=2-5]
        \arrow[from=1-1, to=1-5]
        \arrow["{p_i^0}"', from=4-1, to=5-1]
        \arrow[from=4-5, to=5-5]
        \arrow[from=7-4, to=7-5]
        \arrow[from=6-1, to=7-1]
        \arrow[from=5-1, to=6-1]
        \arrow[from=7-1, to=7-2]
        \arrow[from=4-4, to=4-5]
        \arrow[from=4-2, to=4-3]
        \arrow[from=4-3, to=4-4]
        \arrow[from=4-3, to=5-3]
        \arrow[from=5-3, to=6-3]
        \arrow[from=6-3, to=7-3]
        \arrow[from=7-2, to=7-3]
        \arrow[from=7-3, to=7-4]
        \arrow["{p_1, \dots, p_m}"{description}, from=2-5, to=4-5]
        \arrow[from=6-5, to=7-5]
        \arrow[from=5-5, to=6-5]
    \end{tikzcd}
    \]
    \end{itemize}
    
    We get fairness of $(\theta_i)_{i\in\omega}$ from lemma~\ref{mcutonestepCompleteness} and from the fact that after the translation of an (\mcut)-step, $\pi^\circ\rightsquigarrow {\pi'}^\circ$, each residual of a redex $\mathcal{R}$ of $\pi^\circ$, is contained in the translations of residuals of the associated redex $\mathcal{R}'$ of lemma~\ref{mcutonestepCompleteness}.
    \end{proof}

\subsection{Details on the main theorem}

    \begin{thm}
        \label{app:musuperLLmodinfCutElim}
        Every fair (\mcut)-reduction sequence of $\musuperLLinf(\Sig, \leqg, \leqf, \lequ)$ converges to a $\musuperLLinf(\Sig, \leqg, \leqf, \lequ)$ cut-free proof.
        \end{thm}
        \begin{proof}
        Consider a $\musuperLLinf(\Sig, \leqg, \leqf, \lequ)$ fair reduction sequence $(\pi_i)_{i\in 1+\lambda}$ ($\lambda\in\omega+1$).
        If the sequence is finite, we use lemma~\ref{redSeqTranslationFiniteness} and we are done.
        If the sequence is infinite, using corollary~\ref{FairredSeqTranslationFiniteness} we get a fair infinite \muLLinf{} reduction sequence $(\theta_i)_{i\in\omega}$ and a sequence $(\varphi(i))_{i\in\omega}$ of natural numbers.
        By theorem~\ref{thm:mullcutelim}, we know that $(\theta_i)_{i\in\omega}$ converges to a cut-free proof $\theta$ of \muLLinf{}.
        We now prove that the sequence $(\pi_i)_{i\in\omega}$ converges to a $\musuperLLinf(\Sig, \leqg, \leqf, \lequ)$ pre-proof $\pi$ such that $\pi^\circ=\theta$ up to a permutation of rules (the permutations of one particular rule being finite).
        
        First, we prove that for each depth $d$, there is an $i$ such that there are no $(\mcut)$-rules under depth $d$ in $\pi_i$. Suppose for the sake of contradiction that there exist a depth $d$ such that there always exist a $(\mcut)$ at depth $d$.
        There is a rank $i'$ and an $(\mcut)$ rule in $\pi_{i'}$ such that for each $i\geq i'$, $\pi_i$ will always contain this $(\mcut)$ and (therefore) the branch $b$ to it never changes.
        The translations $\pi_{i'}^\circ$ contains the translation of the branch $b$ which also ends with an $\mcut$.
        Since $\pi_{i'}^\circ$ is equal to $\theta_{\varphi(i')}$ up to the permutations of rules under the multicut and that these permutations do not change the depths of the $(\mcut)$ rules, we have that the $\theta_{\varphi(i)}$ all contains a $(\mcut)$ at a depth equal to the depth of the translation of $b$.
        This contradicts the productivity of this sequence of reduction, we therefore have that $(\pi_i)$ converges to a pre-proof $\pi$.
        
        Second, we prove that $\pi^\circ$ is equal to $\theta$ up to a permutation of rules (the permutations of one particular rule being finite).
        The condition on the sequence given by corollary~\ref{FairredSeqTranslationFiniteness} defines a sequence of rule permutation starting from $\pi^\circ$:
        $$p_0^0, \dots, p_0^{k_0-1}, p_1^{k_0}, \dots, p_1^{k_1-1}, \dots, p_n^{k_{n-1}}, \dots, p_{n+1}^{k_n}, \dots,$$
        moreover we have that this is a permutation of rules with finite permutation, therefore this sequence of rule permutation converges to a \muLLinf{} pre-proof $\pi'$.
        We have for each $i$, that the end of the sequence of rule permutation 
        $$p_0^0, \dots, p_0^{k_0-1}, p_1^{k_0}, \dots, p_1^{k_1-1}, \dots, p_i^{k_{i-1}}, \dots, p_i^{k_i-1}$$ 
        starting from $\pi^\circ$ is equal to $\pi_i^{k_i}$ under the multicuts
        . Therefore we have that the sequence $(\pi_i^{k_i})_{i\in\omega}=(\theta_{\varphi(i)})_{i\in\omega}$ converges to $\pi'$ and therefore that $\pi'=\theta$.
        As rule permutation with finite permutation and $(-)^\circ$ translation are robust to validity (both ways), we have that $\pi$ is valid.
        \end{proof}

\subsection{Details on corollary~\ref{coro:superllcutelim2}}

\begin{coro}[Cut Elimination for \superLL]
    \label{app:superllcutelim2}
Cut elimination holds for $\superLL(\Sig,\leqg,\leqf,\lequ)$ as soon as the 8 cut-elimination axioms of definition~\ref{cutElimAxs} are satisfied.
\end{coro}
\begin{proof}
Any $\superLL(\Sig, \leqg, \leqf, \lequ)$-proof is also $\musuperLLinf(\Sig, \leqg, \leqf, \lequ)$-proof therefore any sequence of (\mcut)-reductions converges to a cut-free proof. A cut-free proof of sequents containing only $\superLL(\Sig, \leqg, \leqf, \lequ)$-formulas and valid rules from \linebreak $\musuperLLinf(\Sig, \leqg, \leqf, \lequ)$ is necessarily a $\superLL(\Sig, \leqg, \leqf, \lequ)$ (cut-free) proof.
\end{proof}


\clearpage
\tableofcontents

\end{document}